\documentclass[
aps,
prx,
reprint,
floatfix,
superscriptaddress,
amsmath,
amssymb,
amsfonts,
nofootinbib,
longbibliography
]{revtex4-2}  

\usepackage{graphicx}     
\usepackage{dcolumn}   
\usepackage{bm}           
\usepackage{amsmath}      
\usepackage{amssymb}      
\usepackage{amsfonts}     
\usepackage{amsthm}       
\usepackage{amstext}
\usepackage{amsbsy}
\usepackage{color}
\usepackage{xcolor}       
\usepackage{wasysym}
\usepackage{enumerate}
\usepackage{relsize}
\usepackage{pstricks}
\usepackage{tikz}
\usetikzlibrary{shapes.geometric, arrows}
\usepackage{siunitx}
\usepackage{mhchem}
\usepackage{stackengine}
\usepackage{scalerel}
\usepackage{textcomp}
\usepackage{float}
\usepackage{verbatim}
\usepackage{dsfont}
\usepackage{cancel}
\usepackage{comment}
\usepackage{array, multirow, bigdelim, makecell, booktabs}
\usepackage[misc]{ifsym}
\usepackage{sidecap}
\usepackage{pifont}
\usepackage{mathtools}
\usepackage{enumitem}

\usepackage[colorlinks=true,linkcolor=blue,urlcolor=blue,filecolor=magenta,citecolor=blue,breaklinks]{hyperref}
\hypersetup{
    pdftitle={Overleaf Example},
    pdfpagemode=FullScreen
}
\usepackage[capitalise]{cleveref}

\definecolor{darkblue}{HTML}{004D6B}
\definecolor{darkred}{HTML}{8c1515}
\definecolor{darkgreen}{HTML}{006400}

\newcommand{\vdimers}{\scalebox{0.8}{\,\pmb{{\textnormal{|\phantom{\tiny{m}}|}}}\,}}
\newcommand{\hdimers}{\,\pmb{\overline{\underline{\text{\phantom{\raisebox{2pt}{\tiny{ee}}}}}}}\,}
\newcommand{\svdimers}{\scalebox{0.5}{\,\pmb{{\text{|\phantom{\tiny{m}}|}}}\,}}
\newcommand{\shdimers}{\,\pmb{\overline{\underline{\text{\phantom{\raisebox{0.1pt}{\tiny{n}}}}}}}\,}
\newcommand{\crossdimers}{\pmb{\mathlarger{\times}}}

\newtheorem{theorem}{Theorem}[section]

\newtheorem{lemma}[theorem]{Lemma}

\theoremstyle{remark}

\allowdisplaybreaks

\begin{document}

\title{Microscopic Spin-1 Parent Hamiltonians for Emergent Valence-Bond Loop
Manifolds}
\author{Hari Borutta}
\affiliation{Department of Physics, Indian Institute of Technology Madras, Chennai 600036, India}
\author{Yasir Iqbal}
\affiliation{Department of Physics, Indian Institute of Technology Madras, Chennai 600036, India}
\author{Kirill Shtengel}
\affiliation{Department of Physics and Astronomy, University of California,
Riverside, CA 92521, USA}
\affiliation{Department of Physics, Indian Institute of Technology Madras, Chennai 600036, India}

\date{\today}

\begin{abstract}
We construct an exact spin-$1$ parent Hamiltonian for constrained
valence-bond loop manifolds on the checkerboard and pyrochlore
lattices. The Hamiltonian is local, SU(2)- and time-reversal-invariant,
and built from positive-semidefinite projectors acting on triangular
faces.  Each projector removes only the maximally polarized state of a
triangle, so the model is frustration-free.  Its zero-energy states are
generated by an AKLT-like construction in which each spin-$1$ moment is
resolved into two virtual spin-$\tfrac12$ degrees of freedom, singlets
are formed inside every crossed plaquette or tetrahedron, and the
physical spin-$1$ Hilbert space is recovered by projection.  The
resulting ground states are fully packed singlet-loop states on the
corner-sharing lattice.  Thus, a loop or dimer constraint, usually
introduced as part of an effective Rokhsar--Kivelson description,
appears here as the exact zero-energy manifold of a microscopic spin
Hamiltonian.  We analyze spin correlations within this manifold and
show that, for a fixed loop covering, they are determined by loop
connectivity.  We also project symmetry-allowed perturbations into the
ground-state manifold and derive the resulting low-energy pseudospin
dynamics.  The checkerboard and pyrochlore cases differ sharply.  On
the pyrochlore lattice, tetrahedral symmetry removes simple local bias
terms, and the leading nontrivial next-nearest-neighbour Heisenberg
perturbation gives an emergent spin-$\tfrac12$ XY model on the diamond
lattice of tetrahedron centers.  These results give an exact spin-$1$
microscopic starting point for constrained valence-bond physics in two
and three dimensions, and show how loop, dimer, and gauge-theoretic
descriptions can be approached from a small-spin, SU(2)-invariant
frustrated magnet.
\end{abstract}

\maketitle

\tableofcontents

\section{Introduction}

Many of the most useful descriptions of frustrated quantum magnets begin with constrained degrees of freedom. Quantum dimer models impose a local dimer constraint and then reveal resonating valence-bond physics, topological sectors, and emergent gauge structure. Quantum spin ice starts from an ice rule and gives rise to Coulombic correlations and fractionalized excitations. These effective descriptions have been central to our understanding of quantum-disordered phases in frustrated magnets~\cite{Balents2010}. They also pose a basic microscopic question: can such constrained manifolds arise directly as the exact ground-state space of a local spin Hamiltonian?

This question is especially sharp in three dimensions. Exact ground states are rare in frustrated quantum magnets, where geometric frustration and strong quantum fluctuations often suppress conventional order but make controlled microscopic analysis difficult. Parent-Hamiltonian constructions provide one of the few settings where firm statements can be made. The Affleck--Kennedy--Lieb--Tasaki (AKLT) construction gives exact valence-bond ground states from local spin Hamiltonians~\cite{Affleck1987}, but in three-dimensional lattices with large coordination number it typically requires large on-site spins fixed by the lattice connectivity~\cite{Parameswaran2009}. Klein-type projector Hamiltonians offer a complementary mechanism: local projectors enforce singlets on elementary motifs and can produce extensively degenerate ground-state manifolds with sector structures reminiscent of gauge theories~\cite{Nussinov2006,Nussinov2007}. The challenge is to combine these virtues in a small-spin, symmetry-preserving model on a genuinely frustrated three-dimensional lattice.

The pyrochlore lattice, and its two-dimensional checkerboard analogue, provide a natural setting for this problem. Their corner-sharing geometry imposes strong local constraints, while their loop structure connects them closely to valence-bond coverings, quantum dimer models, and emergent gauge descriptions. In particular, the pyrochlore antiferromagnet is a canonical setting for three-dimensional Coulomb phases and emergent U(1) gauge fields~\cite{Hermele2004a,Moessner2006a}. At the same time, recent work on constrained dimer models on checkerboard and related crossed-medial lattices has shown that non-planar dimer Hilbert spaces can support exactly solvable liquid states and gauge-theory descriptions~\cite{Wildeboer2020,Wildeboer2026}. These developments make it timely to ask whether analogous loop or dimer constraints can be obtained not by imposing an effective Hilbert space, but as the kernel of a microscopic SU(2)-invariant spin Hamiltonian.

In this work, we answer this question for spin one. We introduce a local SU(2)- and time-reversal-invariant spin-1 Hamiltonian on the checkerboard and pyrochlore lattices. The Hamiltonian is a sum of positive-semidefinite projectors acting on triangular faces of the corner-sharing units. Each projector penalizes only the maximally polarized total-spin sector of a triangle. The model is therefore frustration-free, and its zero-energy states can be constructed explicitly using an AKLT-inspired fractionalization procedure: each physical spin-1 moment is resolved into two virtual spin-$\tfrac12$ degrees of freedom, singlets are formed inside every crossed plaquette or tetrahedron, and the two virtual spins on each physical site are projected back into the triplet sector.

The resulting ground states are fully packed singlet-loop states on the underlying corner-sharing geometry. In this sense, the model combines features that are usually separated. Like AKLT Hamiltonians, it has an explicit valence-bond parent structure. Unlike the standard AKLT construction on these lattices, it works already for spin one, rather than requiring the larger on-site spins dictated by coordination number. Like Klein-type models, it has an extensively degenerate constrained ground-state manifold rather than an essentially unique valence-bond solid. And unlike Rokhsar--Kivelson quantum dimer models, the loop or dimer constraint is not assumed from the outset: it emerges as the exact zero-energy condition of a microscopic spin Hamiltonian.

Our first main result is an explicit characterization of this ground-state manifold in terms of singlet-loop states generated by the virtual-spin construction. The local triangular projector constraint eliminates the maximally symmetric spin sector, while the allowed zero-energy states are spanned by configurations in which every triangular face is protected by at least one virtual singlet. On the checkerboard and pyrochlore lattices this condition organizes the ground-state space into fully packed loop configurations. This provides a concrete microscopic realization of a constrained valence-bond manifold in both two and three dimensions.

Our second result concerns correlations inside this exact manifold. Using Schwinger bosons and spin coherent states, we obtain analytical control over overlaps and two-spin correlations for fixed singlet-loop coverings. In such states, spin correlations are governed by loop connectivity: two spins are correlated only when they lie on the same loop, with a strength that decays exponentially with their separation along that loop. This gives a direct geometric interpretation of correlations in the ground-state manifold.

Our third result is the derivation of effective low-energy dynamics generated by perturbations projected into the exact manifold. We focus on perturbations allowed by symmetry---specifically, additional Heisenberg exchanges---and show how they lift the degeneracy among loop states. The checkerboard and pyrochlore cases behave differently. On the checkerboard lattice, local terms can bias the valence-bond orientation on each crossed plaquette. On the pyrochlore lattice, the higher local symmetry suppresses such simple bias terms, pointing instead toward collective loop processes on the diamond lattice of tetrahedra. This distinction is important for understanding which valence-bond phases or liquid regimes may be selected away from the exactly solvable point.

The broader significance of the construction is twofold.  First, it
gives an exact small-spin parent Hamiltonian for a constrained
valence-bond loop manifold on corner-sharing lattices, including the
three-dimensional pyrochlore lattice.  Second, it provides a controlled
microscopic starting point from which effective loop, dimer, and
gauge-theoretic descriptions can be derived by projection, rather than
postulated as the initial Hilbert space.  In this way, the construction
connects AKLT-type fractionalization, Klein-type degeneracy, and
dimer-like constrained dynamics within a single local Hamiltonian.  It
also places the checkerboard and pyrochlore lattices on common footing:
the former links naturally to recent exact results on non-planar
constrained dimer models, while the latter brings the construction into
the three-dimensional setting most closely associated with emergent
gauge physics.  The constrained degrees of freedom are therefore not an
assumption of the theory, but an exact consequence of the local
spin Hamiltonian.

The remainder of the paper is organized as follows.  In
Sec.~\ref{sec:model}, we define the spin-$1$ projector Hamiltonian and
explain its relation to Klein and AKLT constructions.  In
Sec.~\ref{sec:ground_states}, we construct the zero-energy manifold,
prove its loop-state spanning structure, and derive spin correlations
in fixed loop states.  In Sec.~\ref{sec:eff_Hamiltians}, we project
additional Heisenberg interactions into the exact ground-state manifold
and obtain the corresponding effective pseudospin Hamiltonians on the
checkerboard and pyrochlore lattices.  In Sec.~\ref{sec:MC}, we study correlations in a
sign-free equal-amplitude loop superposition by Monte Carlo sampling.
We conclude in Sec.~\ref{sec:discussion_outlook} with a discussion of
the implications of the construction and possible directions for
extending it toward loop and gauge-theoretic descriptions.

\section{The model}
\label{sec:model}

We consider a spin-1 model defined on the checkerboard and pyrochlore lattices. The Hamiltonian is given by
\begin{equation}
    \mathcal{H}_0 = \mathcal{J} \sum_\triangle P_3\left({S}_\triangle\right)
    \label{eq:Hamiltonian}
\end{equation}
where the sum is taken over all triangular faces of the tetrahedra or crossed plaquettes (essentially, flat tetrahedra) forming the lattice [see Fig.~\ref{fig:lattices}] and $P_3\left({S}_\triangle\right)$ is a projector of the combined spin state of each triangle onto the states with the maximum possible value of their total spin, which is 3. (Here $\mathbf{S}_\triangle = \sum_{i\in\triangle}\mathbf{S}_i$.)

\begin{figure}[hbt]
  \centering
     \includegraphics[width=0.4 \columnwidth]{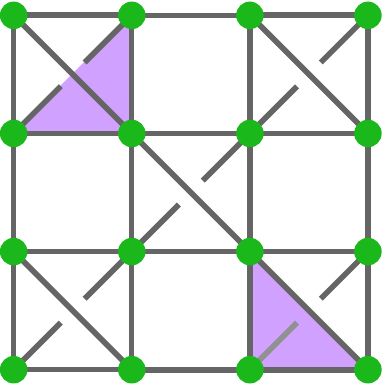}
    \hspace{5mm}
    \includegraphics[width=0.48 \columnwidth]{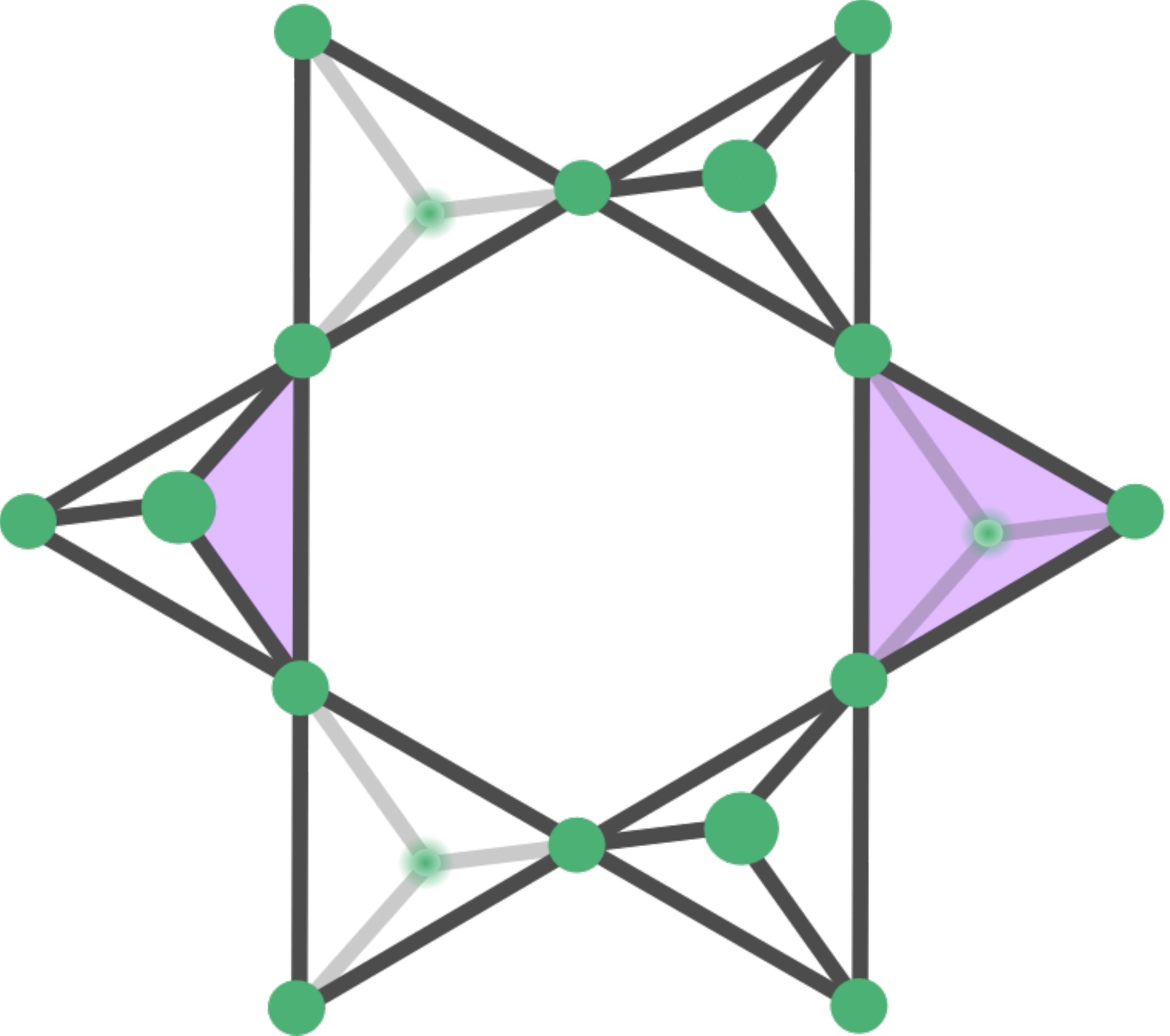}
  \caption{Checkerboard and pyrochlore lattices.  The checkerboard
lattice is the two-dimensional corner-sharing analogue of the
pyrochlore lattice, whose local units are tetrahedra.  Representative
triangular faces are shaded in purple; triangular faces of all
orientations carry the projector interaction in
Eq.~\eqref{eq:Hamiltonian}.  The shortest singlet loops have length
four on the checkerboard lattice and length six on the pyrochlore
lattice, where the latter are formed by bonds belonging to different
tetrahedra.}
\label{fig:lattices}
\end{figure}

This Hamiltonian has features in common with both
Klein~\cite{Klein1982,Nussinov2006} and AKLT~\cite{Affleck1987,Affleck1988}
Hamiltonians.  Since the projector onto the total-spin-$3$ sector of
three spin-$1$ moments can be written as a polynomial in
$\mathbf S_\triangle^2$,
\[
    P_3(\mathbf S_\triangle)
    =
    \frac{
    \mathbf S_\triangle^2
    \left(\mathbf S_\triangle^2-2\right)
    \left(\mathbf S_\triangle^2-6\right)
    }{720},
\]
the Hamiltonian contains multi-spin interactions when expanded in
powers of the microscopic spin operators.  Nevertheless, it is local,
SU(2)-invariant, and time-reversal invariant.  As in Klein and AKLT
constructions, its ground states can be constructed and analyzed
exactly.

A key difference from the standard AKLT construction is that the present
model is defined already for spin one on the checkerboard and
pyrochlore lattices.  The usual AKLT construction on a lattice with
coordination number $z$ requires on-site spin $z/2$, or an integer
multiple thereof; on the checkerboard and pyrochlore lattices this
would require spin three or larger~\cite{Parameswaran2009}.  This
difference is important because large-spin AKLT constructions in three
and higher dimensions are generally more susceptible to conventional
ordering, although the pyrochlore lattice itself appears to evade this
tendency because of its strong geometric frustration~\cite{Parameswaran2009}.

A second key difference is that the present Hamiltonian has an extensive
ground-state manifold rather than a unique valence-bond solid.  In this
respect it is closer in spirit to Klein-type projector Hamiltonians,
while retaining an explicit AKLT-inspired virtual-spin construction.
The structure and consequences of this ground-state manifold are the
subject of the next section.

\section{Properties of the ground states}
   \label{sec:ground_states}

\subsection{Ground state manifold}
   \label{sec:gs_manifold}
The ground states of the Hamiltonian given by Eq.~(\ref{eq:Hamiltonian}) are constructed through a procedure similar to that used for the original AKLT chain~\cite{Affleck1987,Affleck1988} and outlined below:
\begin{enumerate}[label=(\roman*)]
    \item Break each spin-1 degree of freedom into two spin-$\tfrac12$ constituents, splitting these constituents between two tetrahedra or crossed plaquettes that share the original spin.
    \item Form two singlets between four spin-$\tfrac12$ constituents within each tetrahedron/crossed plaquette. (There are three ways of doing this in each such plaquette or tetrahedron.)
    \item Project the combined state of all pairs of constituent spin-$\tfrac12$ on each site of the original lattice onto a triplet, thus restoring the spin-1 nature of the model's physical degrees of freedom.
\end{enumerate}
An example of this procedure in a single crossed plaquette is shown in Fig.~\ref{fig:AKLT_chequer}.
\begin{figure}[bht]
  \centering
     \includegraphics[width=0.45 \columnwidth]{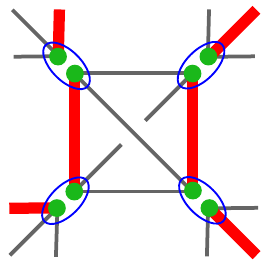}
  \caption{AKLT-like construction on a single crossed plaquette or
tetrahedron.  Each physical spin-$1$ degree of freedom is resolved into
two virtual spin-$\tfrac12$ variables.  Within the crossed plaquette or
tetrahedron, the virtual spins are paired into singlets, shown as thick
red bonds.  The two virtual spins associated with the same physical
site are then symmetrized, indicated by blue ovals, to recover the
physical spin-$1$ Hilbert space.}
  \label{fig:AKLT_chequer}
\end{figure}
The resulting three ground states for each crossed plaquette/tetrahedron can be written using Schwinger bosons~\cite{Auerbach1998}:
\begin{subequations}\label{eq:dimer_states}
\begin{equation}\label{eq:hor_dimers}
\lvert \hdimers \rangle = \frac{1}{2} \left(a^\dag_1  b^\dag_2 - b^\dag_1  a^\dag_2\right)\left(a^\dag_3  b^\dag_4 - b^\dag_3  a^\dag_4 \right)\lvert 0 \rangle \, ;
\end{equation}
\begin{equation}\label{eq:vert_dimers}
\lvert \vdimers\rangle =\frac{1}{2} \left(a^\dag_3  b^\dag_2 - b^\dag_3  a^\dag_2\right)\left(a^\dag_1  b^\dag_4 - b^\dag_1  a^\dag_4\right)\lvert 0 \rangle \, ;
\end{equation}
\begin{multline}\label{eq:diag_dimers}
\lvert \crossdimers\rangle =  \frac{1}{2} \left(a^\dag_1  b^\dag_3 - b^\dag_1  a^\dag_3\right)
\left(a^\dag_2  b^\dag_4 - b^\dag_2  a^\dag_4 \right)\lvert 0 \rangle\\ 
= \lvert \hdimers\rangle - \lvert \vdimers\rangle \, .
\end{multline}
\end{subequations}
with the following numbering convention: on the checkerboard lattice, the sites of the crossed plaquette are numbered anticlockwise beginning with the bottom left; this results in a sign convention where all horizontal and vertical singlets are oriented from sites of sublattice A of the underlying square lattice (sites 1 and 3) to the sites of sublattice B (sites 2 and 4). The diagonal singlets are oriented upward.
On the pyrochlore lattice, we can simply ``designate'' two pairs of mutually orthogonal edges as ``horizontal'' and ``vertical''; focusing on these four (out of six) edges renders the lattice effectively bipartite, as shown in Fig.~\ref{fig:pyrochlore_configs}. 
It is already clear from Eq.~(\ref{eq:diag_dimers}) that these states are not linearly independent. Furthermore, $\langle \hdimers | \hdimers \rangle = \langle \vdimers | \vdimers \rangle = 1$ whereas $\langle \hdimers | \vdimers \rangle = 1/2$. 
A pair of orthonormal basis states can then be defined as
\begin{subequations}\label{eq:basis_states}
\begin{equation}\label{eq:A_state}
\lvert A\rangle = \frac{\sqrt{3}}{3} \left(\lvert \hdimers\rangle + \lvert \vdimers\rangle\right) \, ;
\end{equation}
\begin{equation}\label{eq:B_state}
\lvert B\rangle =  \left(\lvert \hdimers\rangle - \lvert \vdimers\rangle\right) = \lvert \crossdimers\rangle \, .
\end{equation}
\end{subequations}

\begin{figure}[bht]
    \centering
    \includegraphics[width=\columnwidth]{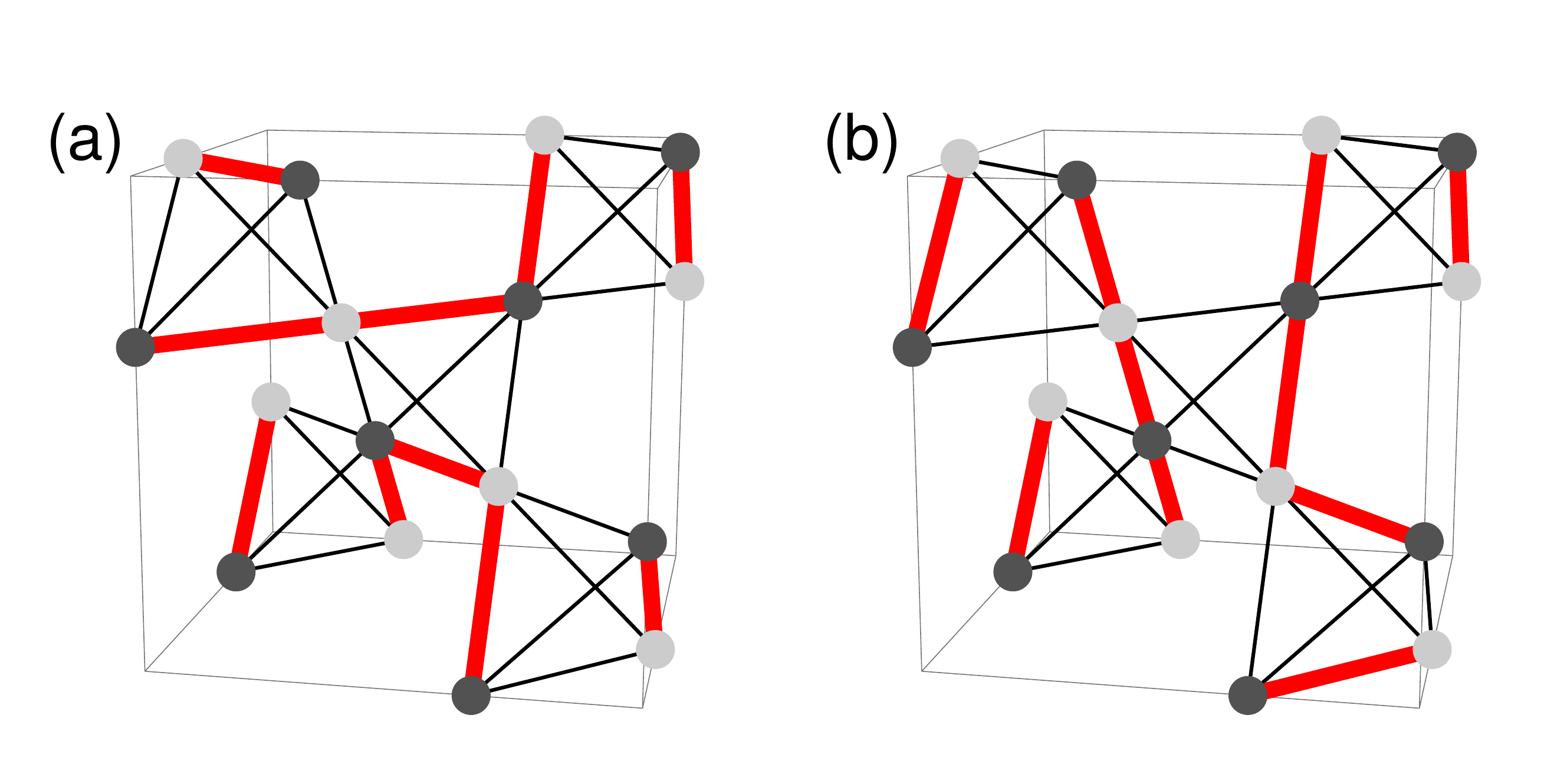}
    \caption{Two representative singlet configurations on the same portion
of the pyrochlore lattice.  The gray and black sites denote the two
sublattice types used to orient the singlets consistently.  In each
tetrahedron, the singlet pairing is chosen from the two independent
opposite-pair states $\lvert\hdimers\rangle$ and
$\lvert\vdimers\rangle$ of Eq.~\eqref{eq:dimer_states}, which are used
to construct the orthonormal local basis in Eq.~\eqref{eq:basis_states}.
All singlets shown connect opposite sublattices.  Boundary singlets
outside the displayed cluster are not shown.}
\label{fig:pyrochlore_configs}
\end{figure}

However, for the remainder of this section, we will ignore the orthonormal bases for each crossed plaquette and instead focus on the ``singlet loop'' states defined as a direct product of \emph{dimer} states $\lvert \hdimers \rangle$, $\lvert \vdimers \rangle$ and $\lvert \crossdimers\rangle$ on each crossed plaquette/tetrahedron. We claim that such states, henceforth denoted as $\lvert\mathcal{L} \rangle$, form an overcomplete basis for \emph{all} ground states of Hamiltonian in Eq.~(\ref{eq:Hamiltonian}). 

While the statement that all such singlet loop states are ground states is self-evident (the presence of a singlet in each triangle excludes the possibility of its three spins combining into a spin-3 state), to prove the (over)completeness we need to show that no additional ground states exist outside of the Hilbert subspace spanned by the loop states.

Inspired by the proof of uniqueness of the ground state of the standard AKLT model on a 1D chain laid out in~\cite{Tasaki2020}, we begin with the following lemmas (for both the checkerboard and pyrochlore lattices):

Before proving the local statement, it is useful to record the
angular-momentum structure on a single triangle.  Consider a triangle
$
\tikz[baseline=(current bounding box.center), scale=0.7]{
    \draw[gray, ultra thin] (0,0) -- (0.8,0) -- (0.4,0.69) -- (0,0);
    \node at (0.4,1.0) {$3$};
    \node at (-0.2,0) {$1$};
    \node at (1.0,0) {$2$};
}
$
of spin-$1$ degrees of freedom.  Each physical spin-$1$ is obtained
by symmetrizing two virtual spin-$\tfrac12$ variables.  Thus the local
Hilbert space may be written as
\begin{equation}
    \mathcal H_\triangle
    =
    \mathrm{Sym}^2 V_1
    \otimes
    \mathrm{Sym}^2 V_2
    \otimes
    \mathrm{Sym}^2 V_3,
    \qquad
    V_i \simeq \mathbb C^2 .
    \label{eq:triangle_virtual_hilbert_space}
\end{equation}
Equivalently, in terms of physical spin representations,
\begin{equation}
    \mathcal H_\triangle
    =
    1_1 \otimes 1_2 \otimes 1_3 .
\end{equation}
Here and below a symbol such as $s$ denotes the spin-$s$ irreducible
representation.  The decomposition into total-spin sectors is
\begin{equation}
    1_1 \otimes 1_2 \otimes 1_3
    =
    3
    \oplus
    2 \cdot 2
    \oplus
    3 \cdot 1
    \oplus
    0 .
    \label{eq:three_spin_one_decomposition}
\end{equation}
The full triangle Hilbert space has dimension $3^3=27$, while the
total-spin-$3$ sector has dimension $2\times 3+1=7$.  Hence
\begin{equation}
    \dim \ker P_3(\mathbf S_\triangle)
    =
    27-7
    =
    20 .
    \label{eq:triangle_kernel_dimension}
\end{equation}
For a fixed edge, say $(12)$, the subspace generated by one intersite
virtual singlet has the form
\begin{equation}
    \mathcal E_{12}
    =
    \left\{
    \mathcal S_{12}^\dagger |\phi\rangle
    \;\middle|\;
    |\phi\rangle \in
    \mathrm{Sym}^2 V_1\otimes \mathrm{Sym}^2 V_2\otimes \mathrm{Sym}^2 V_3
    \right\}
    \label{eq:E12_dimension_count}
\end{equation}
where
\begin{equation}
    \mathcal S_{12}^\dagger
    =
    a_1^\dagger b_2^\dagger
    -
    b_1^\dagger a_2^\dagger
    \label{eq:virtual_singlet_12}
\end{equation}
creates a virtual spin-$\tfrac12$ singlet between sites $1$ and $2$.
This single-edge subspace has dimension
\begin{equation}
    \dim \mathcal E_{12}
    =
    2 \times 2 \times 3
    =
    12 .
\end{equation}
Thus a single edge-singlet space is not the whole triangular kernel.
The correct statement is instead that the full kernel is the linear
span of the three edge-singlet spaces.  This is the content of the
following lemma.

\begin{lemma}[Local triangular kernel]
\label{lemma1}
Consider a triangle $\triangle=(1,2,3)$ of spin-$1$ degrees of
freedom.  Let
\begin{equation}
    \mathbf S_\triangle
    =
    \mathbf S_1+\mathbf S_2+\mathbf S_3 ,
\end{equation}
and let $P_3(\mathbf S_\triangle)$ be the projector onto the
total-spin $S_\triangle=3$ sector of $\mathcal H_\triangle$.  Then
\begin{equation}
    \ker P_3(\mathbf S_\triangle)
    =
    \mathcal E_{12}
    +
    \mathcal E_{23}
    +
    \mathcal E_{31},
    \label{eq:triangle_kernel_span}
\end{equation}
where $\mathcal E_{ij}$ is the subspace spanned by states containing
an intersite virtual spin-$\tfrac12$ singlet on the edge $(ij)$.
Equivalently, every local zero-energy state on a triangle can be
written as a linear combination of states in which at least one edge
of the triangle carries a virtual spin-$\tfrac12$ singlet.
\end{lemma}

\begin{proof}
We use the Schwinger-boson representation of the virtual spinors.
On site $i$, let $a_i^\dagger$ and $b_i^\dagger$ create the two
components of a spin-$\tfrac12$ spinor.  The physical spin-$1$ Hilbert
space is the two-boson subspace on each site.  For an edge $(ij)$,
define the SU(2)-invariant virtual-singlet creation operator
\begin{equation}
    \mathcal S_{ij}^\dagger
    =
    a_i^\dagger b_j^\dagger
    -
    b_i^\dagger a_j^\dagger .
    \label{eq:virtual_singlet_operator}
\end{equation}
The subspace $\mathcal E_{ij}$ consists of states of the form
\begin{equation}
    \mathcal S_{ij}^\dagger
    F_{ij;k}
    \left(
    a_i^\dagger,b_i^\dagger;
    a_j^\dagger,b_j^\dagger;
    a_k^\dagger,b_k^\dagger
    \right)
    |0\rangle ,
    \label{eq:Eij_definition}
\end{equation}
where $k$ is the remaining site of the triangle.  The polynomial
$F_{ij;k}$ is homogeneous of degree $1$ in the bosons on site $i$,
degree $1$ in the bosons on site $j$, and degree $2$ in the bosons
on site $k$.  Thus the total boson number on each site is $2$, as
required for a physical spin-$1$ state.

We first show that
\begin{equation}
    \mathcal E_{12}
    +
    \mathcal E_{23}
    +
    \mathcal E_{31}
    \subseteq
    \ker P_3(\mathbf S_\triangle).
    \label{eq:E_subset_kernel}
\end{equation}
The total-spin $S_\triangle=3$ sector is the maximally symmetric
representation obtained by symmetrizing all six virtual
spin-$\tfrac12$ variables on the triangle.  A state in
$\mathcal E_{ij}$ contains an antisymmetric pair of virtual spinors
on sites $i$ and $j$, namely the singlet
$\mathcal S_{ij}^\dagger$.  Such a state has no component in the
completely symmetric six-spinor sector.  Therefore it is annihilated
by $P_3(\mathbf S_\triangle)$, proving Eq.~\eqref{eq:E_subset_kernel}.

It remains to prove the reverse inclusion.  We first recall the
corresponding two-site statement.  For two physical spin-$1$ sites,
\begin{equation}
    \mathcal H_i^{(S=1)}
    \otimes
    \mathcal H_j^{(S=1)}
    =
    1_i \otimes 1_j
    =
    2 \oplus 1 \oplus 0 .
    \label{eq:two_spin_decomposition}
\end{equation}
The $S_{ij}=2$ subspace is the completely symmetric subspace of the
four virtual spinors on sites $i$ and $j$.  Its orthogonal complement,
namely the $S_{ij}=1\oplus 0$ subspace, is precisely the space
generated by one intersite virtual singlet $\mathcal S_{ij}^\dagger$
multiplied by one remaining virtual spinor on each of the two sites.
This singlet construction is orthogonal to the fully symmetric
four-spinor sector and has dimension
\begin{equation}
    \dim (S_{ij}=1\oplus 0)
    =
    3+1
    =
    4 ,
\end{equation}
which is the dimension of the complement of the $S_{ij}=2$ subspace
inside $1_i\otimes 1_j$.  Hence, for each pair $(ij)$, the subspace
generated by an intersite virtual singlet is exactly the complement
of the maximal pair-spin sector $S_{ij}=2$.

Now let
\begin{equation}
    |\psi\rangle
    \in
    \mathcal H_\triangle
\end{equation}
be orthogonal to
\begin{equation}
    \mathcal E_{12}
    +
    \mathcal E_{23}
    +
    \mathcal E_{31}.
\end{equation}
Then $|\psi\rangle$ is orthogonal to the states in $\mathcal E_{ij}$ for every
edge $(ij)$.  By the two-site result above, this implies that every
pair of spins in $|\psi\rangle$ lies in its maximal pair-spin sector:
\begin{align}
    \left(\mathbf S_i+\mathbf S_j\right)^2 |\psi\rangle
    &=
    2(2+1)\,|\psi\rangle
    =
    6\,|\psi\rangle ,
    \nonumber\\
    &\hspace{2.0cm}
    (ij)=(12),(23),(31).
    \label{eq:pair_spin_maximal}
\end{align}
Since $\mathbf S_i^2=\mathbf S_j^2=1(1+1)=2$,
Eq.~\eqref{eq:pair_spin_maximal} is equivalent to
\begin{equation}
    \mathbf S_i\cdot \mathbf S_j
    |\psi\rangle
    =
    |\psi\rangle ,
    \qquad
    (ij)=(12),(23),(31).
    \label{eq:pair_dot_product}
\end{equation}
Therefore
\begin{align}
    \mathbf S_\triangle^2 |\psi\rangle
    &=
    \left(
    \sum_{i=1}^{3} \mathbf S_i^2
    +
    2\sum_{i<j}\mathbf S_i\cdot \mathbf S_j
    \right)
    |\psi\rangle
    \nonumber\\
    &=
    \left(
    3\times 2
    +
    2\times 3
    \right)
    |\psi\rangle
    \nonumber\\
    &=
    12\,|\psi\rangle .
    \label{eq:total_spin_three}
\end{align}
Thus $|\psi\rangle$ has total spin $S_\triangle=3$, since
$3(3+1)=12$.  We have therefore shown that the orthogonal complement
of
\begin{equation}
    \mathcal E_{12}
    +
    \mathcal E_{23}
    +
    \mathcal E_{31}
\end{equation}
is exactly the $S_\triangle=3$ subspace.  Consequently,
\begin{equation}
    \mathcal E_{12}
    +
    \mathcal E_{23}
    +
    \mathcal E_{31}
    =
    \left(
    \mathcal H_\triangle^{S_\triangle=3}
    \right)^\perp
    =
    \ker P_3(\mathbf S_\triangle).
\end{equation}
This proves the lemma.
\end{proof}

Diagrammatically, the lemma justifies representing the local kernel by
configurations in which at least one edge of the triangle carries a
virtual singlet, as in
$
\tikz[baseline=(current bounding box.center), scale=0.6]{
    \draw[gray, ultra thin] (0.8,0) -- (0.4,0.69) -- (0,0);
    \draw[ultra thick] (0,0) -- (0.8,0);
    \node at (0.4,0.9) {$1$};
    \node at (-0.2,0) {$\frac12$};
    \node at (1.0,0) {$\frac12$};
}
$;
this graphical language should be understood as a statement about a
spanning set, not as a unique bond assignment for an arbitrary
superposition in the kernel.

\begin{lemma}[Local patterns that extend through the lattice]
\label{lemma:extendable_patterns}
Consider a zero-energy state of $\mathcal H_0$ on the checkerboard or
pyrochlore lattice with periodic boundary conditions or in the infinite lattice.  Resolve each
physical spin-$1$ into two virtual spin-$\tfrac{1}{2}$ variables. Then one
out of the two virtual spin-$\tfrac{1}{2}$ has to be assigned to each of the two adjacent corner-sharing units and further any
virtual-singlet pattern that can be extended through the full lattice
while satisfying all triangular projector constraints has, on every
tetrahedron or crossed plaquette, one of the three opposite-pair
valence-bond patterns
\[
\begin{array}{ccc}
\tikz[baseline=(current bounding box.center), scale=0.6]{
\draw[gray, ultra thin] (0,0) -- (0.8,0) -- (0.8,0.8) -- (0,0.8) -- (0,0);
\draw[ultra thick] (0,0) -- (0.8,0.8);
\draw[ultra thick] (0.8,0) -- (0,0.8);
\node at (-0.2,-0.1) {$\frac{1}{2}$};
\node at (1.0,-0.1) {$\frac{1}{2}$};
\node at (1.0,0.9) {$\frac{1}{2}$};
\node at (-0.2,0.9) {$\frac{1}{2}$};
}
&
\tikz[baseline=(current bounding box.center), scale=0.6]{
\draw[gray, ultra thin] (0,0) -- (0.8,0);
\draw[gray, ultra thin] (0.8,0.8) -- (0,0.8);
\draw[gray, ultra thin] (0,0) -- (0.8,0.8);
\draw[gray, ultra thin] (0.8,0) -- (0,0.8);
\draw[ultra thick] (0,0) -- (0,0.8);
\draw[ultra thick] (0.8,0) -- (0.8,0.8);
\node at (-0.2,-0.1) {$\frac{1}{2}$};
\node at (1.0,-0.1) {$\frac{1}{2}$};
\node at (1.0,0.9) {$\frac{1}{2}$};
\node at (-0.2,0.9) {$\frac{1}{2}$};
}
&
\tikz[baseline=(current bounding box.center), scale=0.6]{
\draw[ultra thick] (0,0) -- (0.8,0);
\draw[ultra thick] (0.8,0.8) -- (0,0.8);
\draw[gray, ultra thin] (0,0) -- (0.8,0.8);
\draw[gray, ultra thin] (0.8,0) -- (0,0.8);
\draw[gray, ultra thin] (0,0) -- (0,0.8);
\draw[gray, ultra thin] (0.8,0) -- (0.8,0.8);
\node at (-0.2,-0.1) {$\frac{1}{2}$};
\node at (1.0,-0.1) {$\frac{1}{2}$};
\node at (1.0,0.9) {$\frac{1}{2}$};
\node at (-0.2,0.9) {$\frac{1}{2}$};
}
\end{array}
\]
corresponding respectively to
\begin{equation}
    |\hdimers\rangle,\qquad
    |\vdimers\rangle,\qquad
    |\crossdimers\rangle .
\end{equation}
Thus the local patterns that extend through the full lattice are
spanned by the opposite-pair valence-bond states.
\end{lemma}

\begin{proof}
Since $H_0$ is a sum of positive-semidefinite projectors, a
zero-energy state must be annihilated by every triangular projector:
\begin{equation}
    P_3(\mathbf S_\triangle)|\Psi\rangle=0
\end{equation}
for every triangular face $\triangle$.  By Lemma~\ref{lemma1}, the
kernel of a triangular projector is spanned by states in which at
least one edge of the triangle carries an intersite virtual
spin-$\tfrac12$ singlet.  We shall use this statement as a local
propagation rule: every triangular face must be protected by at least
one such edge singlet.

Let us first rule out a pattern in which a physical site uses both of
its virtual spin-$\tfrac12$ variables inside the same tetrahedron
or crossed plaquette.  Such a site is then unavailable to the other
corner-sharing unit that also contains it.  In that neighboring unit,
call the unavailable corner $x$, and call the other three corners
$a,b,c$.  The three triangular faces containing $x$ are
\begin{equation}
    (xab),\qquad (xac),\qquad (xbc).
\end{equation}
Because no singlet can use the unavailable corner $x$, these three
triangles must be protected by singlets on the opposite edges
\begin{equation}
    (ab),\qquad (ac),\qquad (bc),
\end{equation}
respectively.  Thus the neighboring unit is forced into the
three-edge singlet-triangle pattern
\[
\tikz[baseline=(current bounding box.center), scale=0.7]{
\draw[gray, ultra thin] (0,0) -- (1,0) -- (1,1) -- (0,1) -- (0,0);
\draw[gray, ultra thin] (0,0) -- (1,1);
\draw[gray, ultra thin] (1,0) -- (0,1);
\draw[ultra thick] (0,0) -- (1,0) -- (1,1) -- (0,0);
\node at (-0.2,1.1) {$x$};
\node at (-0.2,-0.1) {$a$};
\node at (1.2,-0.1) {$b$};
\node at (1.2,1.1) {$c$};
\node at (-0.2,0.55) {$\varnothing$};
}
\]
where $\varnothing$ indicates that no virtual spin remains available
at the corner $x$.

This singlet-triangle pattern uses both virtual spin-$\tfrac12$
variables on the three sites $a,b,c$ inside the same unit.  The
constraint therefore propagates to the neighboring units attached to
$a,b,c$.  On the checkerboard lattice this propagation is illustrated
below.  Starting from a singlet triangle in the crossed plaquette
labeled $1$, the adjacent crossed plaquettes labeled $2$ and $3$ are
forced into singlet-triangle patterns in order to protect all their
triangular faces:
\[
\begin{array}{c}
\begin{tikzpicture}[scale=0.8]
\draw[step=1cm,gray,very thin] (-1,-1) grid (2,2);
\draw[ultra thick] (-1,0) -- (0,0) -- (0,1) -- (-1,0);
\draw[gray, ultra thin] (-1,1) -- (0,0);
\draw[gray, ultra thin] (0,2) -- (1,1);
\draw[gray, ultra thin] (0,1) -- (1,2);
\draw[gray, ultra thin] (0,0) -- (1,-1);
\draw[gray, ultra thin] (0,-1) -- (1,0);
\draw[gray, ultra thin] (1,0) -- (2,1);
\draw[gray, ultra thin] (1,1) -- (2,0);
\node at (-1.2,0.5) {$1$};
\node at (0.5,2.2) {$2$};
\node at (0.5,-1.2) {$3$};
\node at (2.2,0.5) {$4$};
\node at (0.2,0.8) {$a$};
\node at (0.8,0.8) {$b$};
\node at (0.8,0.2) {$c$};
\node at (2.2,1.2) {$d$};
\node at (2.2,-0.2) {$e$};
\end{tikzpicture}
\\[2mm]
\begin{tikzpicture}[scale=0.8]
\draw[step=1cm,gray,very thin] (-1,-1) grid (2,2);
\draw[ultra thick] (-1,0) -- (0,0) -- (0,1) -- (-1,0);
\draw[gray, ultra thin] (-1,1) -- (0,0);
\draw[ultra thick] (0,2) -- (1,1) -- (1,2) -- (0,2);
\draw[gray, ultra thin] (0,1) -- (1,2);
\draw[gray, ultra thin] (0,0) -- (1,-1);
\draw[ultra thick] (0,-1) -- (1,0) -- (1,-1) -- (0,-1);
\draw[gray, ultra thin] (1,0) -- (2,1);
\draw[gray, ultra thin] (1,1) -- (2,0);
\node at (-1.2,0.5) {$1$};
\node at (0.5,2.2) {$2$};
\node at (0.5,-1.2) {$3$};
\node at (2.2,0.5) {$4$};
\node at (0.2,0.8) {$a$};
\node at (0.8,0.8) {$b$};
\node at (0.8,0.2) {$c$};
\node at (2.2,1.2) {$d$};
\node at (2.2,-0.2) {$e$};
\end{tikzpicture}
\end{array}
\]
However, after plaquettes $2$ and $3$ are forced in this way, the
sites $b$ and $c$ are already used up by those neighboring
singlet-triangle patterns.  In the crossed plaquette labeled $4$,
there are then triangular faces containing the unavailable corners
$b$ and $c$.  Such a triangle cannot be protected by any intersite
virtual singlet on one of its edges, because two of its corners have
no available virtual spin.  The same local obstruction occurs on the
pyrochlore lattice, where the corresponding crossed plaquettes are
replaced by tetrahedra.  The argument uses only the corner-sharing
connectivity and the triangular faces of the local units.

Equivalently, the forbidden local situation has the form
\[
\tikz[baseline=(current bounding box.center), scale=0.7]{
\draw[gray, ultra thin] (0,0) -- (1,0) -- (1,1) -- (0,1) -- (0,0);
\draw[gray, ultra thin] (0,0) -- (1,1);
\draw[gray, ultra thin] (1,0) -- (0,1);
\node at (-0.2,0.0) {$\varnothing$};
\node at (1.2,0.0) {$\varnothing$};
\node at (1.2,1.0) {$\frac{1}{2}$};
\node at (-0.2,1.0) {$\frac{1}{2}$};
\node at (0.5,-0.35) {forbidden};
}
\]
because at least one triangular face contains two unavailable corners.
This contradicts Lemma~\ref{lemma1}. Hence, no state that uses both
virtual spin-$\frac{1}{2}$ variables of a physical site inside one
tetrahedron or crossed plaquette can be continued to a zero-energy
state on the full periodic lattice.

It remains to identify the local patterns once such exhausted
arrangements are excluded.  Each tetrahedron or crossed plaquette then
receives exactly one available virtual spin-$\frac{1}{2}$ from each
of its four sites.  A local virtual-singlet pattern is therefore a
matching on four vertices.  To satisfy the triangular constraint on
all four faces, every triangular face must contain at least one edge
of this matching.  A single edge of a crossed plaquette/tetrahedron is not enough, since the triangle
formed by the two vertices not incident on that edge and either
endpoint is left unprotected.  Therefore the matching must contain
two disjoint edges covering all four vertices.  There are exactly
three such perfect matchings:
\begin{equation}
    (12)(34),\qquad
    (32)(14),\qquad
    (13)(24).
\end{equation}
These are precisely
\begin{equation}
    |\hdimers\rangle,\qquad
    |\vdimers\rangle,\qquad
    |\crossdimers\rangle .
\end{equation}
Thus the only local patterns that extend through the full lattice are
the three opposite-pair valence-bond patterns.  This proves the
lemma.
\end{proof}

The three local valence-bond states are not linearly independent.  In
Schwinger-boson notation, with
\begin{equation}
    \mathcal S_{ij}^{\dagger}
    =
    a_i^{\dagger} b_j^{\dagger}
    -
    b_i^{\dagger} a_j^{\dagger},
\end{equation}
the singlet operators obey the Pluecker identity
\begin{equation}
    \mathcal S_{12}^{\dagger}\mathcal S_{34}^{\dagger}
    -
    \mathcal S_{32}^{\dagger}\mathcal S_{14}^{\dagger}
    -
    \mathcal S_{13}^{\dagger}\mathcal S_{24}^{\dagger}
    =
    0 .
\end{equation}
With the definitions used above, this gives
\begin{equation}
    |\crossdimers\rangle
    =
    |\hdimers\rangle
    -
    |\vdimers\rangle .
\end{equation}
The three opposite-pair states therefore span a two-dimensional local
valence-bond space.

\begin{lemma}[Valence-bond loop manifold]
\label{lemma:loop_manifold}
On the checkerboard or pyrochlore lattice with periodic boundary
conditions or in the infinite limit, the zero-energy states obtained from the globally
extendable virtual-singlet patterns are fully packed singlet-loop
states.  They are obtained by choosing, on every tetrahedron or
crossed plaquette, one state from the local valence-bond space spanned
by
\[
|\hdimers\rangle,\qquad
|\vdimers\rangle,\qquad
|\crossdimers\rangle,
\]
and then projecting the two virtual spin-$\frac{1}{2}$ variables at
every physical site onto the symmetric spin-$1$ subspace.
\end{lemma}

\begin{proof}
By Lemma~\ref{lemma:extendable_patterns}, every locally allowed
pattern that extends through the full lattice is one of the three
opposite-pair valence-bond patterns.  Hence each corner-sharing unit
contributes two singlet bonds, connecting its four sites in pairs.
Since each physical site belongs to exactly two adjacent
corner-sharing units, it is touched by exactly two virtual singlet
bonds, one from each adjacent unit.  After projection onto the
physical spin-$1$ Hilbert space, every triangular face still lies in
the span of states with an intersite virtual singlet on at least one
of its edges.  By Lemma~\ref{lemma1}, every triangular projector
$P_3(\mathbf S_\triangle)$ annihilates the resulting state.  Since
$H_0$ is a sum of positive-semidefinite triangular projectors, all
such states have zero energy.

The singlet bonds form a graph in which every physical site has
boson degree two.  On a lattice with periodic boundary conditions, every
connected component of this graph is a closed loop.  The resulting
zero-energy states are therefore fully packed singlet-loop states.
This proves the lemma.
\end{proof}

The three lemmas together should be read as a spanning statement for the
global zero-energy space.  Lemma~III.1 shows that the kernel of each
triangular projector is spanned by states containing at least one
virtual singlet on an edge of that triangle.  Hence, any global
zero-energy state can be expanded in terms of virtual-singlet patterns
that locally satisfy the triangular constraints.  Lemma~III.2 then
identifies which such local patterns can be extended consistently
through the full corner-sharing lattice with the prescribed assignment
of one virtual spin-$\tfrac12$ to each adjacent unit.  Lemma~III.3 finally
shows that these globally extendable patterns are precisely fully packed
singlet-loop configurations.  Thus, the loop states form a spanning set,
generally overcomplete, for the zero-energy manifold obtained from the
virtual-spin construction.

It is worth mentioning that the preceding statement concerns states that can be extended through the full corner-sharing lattice.  It is not a claim that the isolated four-spin kernel of a single tetrahedron or crossed plaquette is only
two-dimensional.  Rather, the point is that the additional local
patterns, such as three-edge singlet-triangle arrangements or other
patterns that use both virtual spins of a site inside the same unit,
cannot be continued through the full lattice while satisfying all
triangular projector constraints.  The sector that survives this
extension is the valence-bond loop manifold described above.

\subsection{Spin--spin correlations in loop ground states}
\label{sec:gs_correlations}

We now compute spin correlations in a fixed loop state. This calculation is useful for two reasons. First, it gives a direct diagnostic of the magnetic correlations present in each basis state of the exact valence-bond manifold. Second, it shows that the correlations are controlled entirely by the geometry of the loop passing through the two spins. This is the spin-$1$ analogue of the familiar AKLT result:
along a loop the correlations are short-ranged and staggered, while spins on different loops are exactly uncorrelated in a fixed loop state.

Let $\mathcal L$ denote a fully packed singlet-loop configuration, and let $|\mathcal L\rangle$ be the corresponding valence-bond state. The state may be written in Schwinger-boson notation as
\begin{multline}
\label{eq:loop_state}
|\mathcal L\rangle
=
\prod_{\langle i,j\rangle\in\mathcal L}
\frac{1}{\sqrt{2}}
\left(
a_i^\dagger b_j^\dagger
-
b_i^\dagger a_j^\dagger
\right)|0\rangle
\\
=
\frac{1}{2^{N/2}}
\prod_{\langle i,j\rangle\in\mathcal L}
\left(
a_i^\dagger b_j^\dagger
-
b_i^\dagger a_j^\dagger
\right)|0\rangle .
\end{multline}
Here, $N$ is the number of physical spin-$1$ sites. Since every physical site is touched by two virtual singlet bonds in a fully packed loop configuration, the number of singlet bonds in $\mathcal L$ is also $N$. The product in Eq.~\eqref{eq:loop_state} runs over all virtual singlet bonds in the loop configuration, including diagonal bonds across crossed plaquettes on the checkerboard lattice when such bonds are present. On the checkerboard lattice, two
loops may cross geometrically at a crossed plaquette through diagonal bonds; this does not imply a branching of the loop graph, since every site still has degree two.

The definition of each singlet requires an ordering of the two sites on the bond. While definitions for the checkerboard lattice with diagonal bonds and pyrochlore lattice were made previously [see under Eqs.~\eqref{eq:dimer_states}], retaining the same definition or changing it does not affect the results below. All diagonal quantities computed below are therefore independent of any choice of singlet definition both in checkerboard and pyrochlore lattices.

We evaluate the correlations using normalized spin-$1$ coherent states,
\begin{equation}
    |\boldsymbol{\Omega}_i\rangle
    =
    \frac{1}{\sqrt{2}}
    \left(
    u_i a_i^\dagger + v_i b_i^\dagger
    \right)^2
    |0\rangle ,
\end{equation}
where
\begin{equation}
    u_i=\cos\left(\frac{\theta_i}{2}\right)
    e^{i\phi_i/2},
    \qquad
    v_i=\sin\left(\frac{\theta_i}{2}\right)
    e^{-i\phi_i/2},
\end{equation}
and
\begin{equation}
    \boldsymbol{\Omega}_i
    =
    (\sin\theta_i\cos\phi_i,
     \sin\theta_i\sin\phi_i,
     \cos\theta_i).
\end{equation}
For the coherent product state
$|\Omega\rangle=\prod_i|\boldsymbol{\Omega}_i\rangle$, the overlap
with a loop state is
\begin{multline}
\label{eq:overlap1}
\Psi_{\mathcal L}(\Omega)
\equiv
\langle \mathcal L|\Omega\rangle
\\
=
\left\langle 0\right|
\frac{1}{2^{N/2}}
\prod_{\langle i,j\rangle\in\mathcal L}
\left(
a_i b_j-b_i a_j
\right)
\prod_{k=1}^{N}
\frac{
\left(
u_k a_k^\dagger+v_k b_k^\dagger
\right)^2
}{\sqrt{2}}
\left|0\right\rangle
\\
=
\frac{1}{2^N}
\prod_{\langle i,j\rangle\in\mathcal L}
\left(
u_i v_j-v_i u_j
\right).
\end{multline}
The elementary identity
\begin{equation}
\label{eq:overlap_aux}
\left|
v_i u_j-u_i v_j
\right|^2
=
\frac{1}{2}
\left(
1-\boldsymbol{\Omega}_i\cdot\boldsymbol{\Omega}_j
\right)
\end{equation}
then gives
\begin{equation}
\label{eq:overlap2}
\left|
\Psi_{\mathcal L}(\Omega)
\right|
=
\frac{1}{2^{3N/2}}
\prod_{\langle i,j\rangle\in\mathcal L}
\sqrt{
1-\boldsymbol{\Omega}_i\cdot\boldsymbol{\Omega}_j
}\, .
\end{equation}

Following the coherent-state representation of spin operators,
the spin--spin correlation function in a fixed loop state can be
written as
\begin{multline}
\label{eq:corr_function_gen}
\frac{
\langle\mathcal L|
\mathbf S_i\cdot\mathbf S_j
|\mathcal L\rangle
}{
\langle\mathcal L|\mathcal L\rangle
}
\\
=
\frac{2(2-\delta_{ij})}{Z_{\mathcal L}}
\int
\prod_k d\boldsymbol{\Omega}_k\,
\left|
\Psi_{\mathcal L}(\Omega)
\right|^2
\,
\boldsymbol{\Omega}_i\cdot\boldsymbol{\Omega}_j ,
\end{multline}
where
\begin{equation}
    Z_{\mathcal L}
    =
    \int
    \prod_k d\boldsymbol{\Omega}_k\,
    \left|
    \Psi_{\mathcal L}(\Omega)
    \right|^2 .
\end{equation}
The measure is normalized so that
$\int d\boldsymbol{\Omega}=4\pi$.  The prefactor gives the correct
spin-$1$ coherent-state value: for $i\neq j$ it is $4$, while for
$i=j$ it gives $\langle \mathbf S_i^2\rangle=2$.

We now focus on $i\neq j$.  Because
$\left|\Psi_{\mathcal L}(\Omega)\right|^2$ factorizes into a product
over loops, all loops that do not contain either $i$ or $j$ cancel
between numerator and denominator.  Thus
\begin{multline}
\label{eq:corr_function_mod}
\frac{
\langle\mathcal L|
\mathbf S_i\cdot\mathbf S_j
|\mathcal L\rangle
}{
\langle\mathcal L|\mathcal L\rangle
}
\\
=
4\,
\frac{
\displaystyle
\int
\prod_k' d\boldsymbol{\Omega}_k\,
\boldsymbol{\Omega}_i\cdot\boldsymbol{\Omega}_j
\prod_{\langle m,n\rangle\in\mathcal L}'
\left(
1-\boldsymbol{\Omega}_m\cdot\boldsymbol{\Omega}_n
\right)
}{
\displaystyle
\int
\prod_k' d\boldsymbol{\Omega}_k\,
\prod_{\langle m,n\rangle\in\mathcal L}'
\left(
1-\boldsymbol{\Omega}_m\cdot\boldsymbol{\Omega}_n
\right)
}.
\end{multline}
The primes indicate that the products are restricted to the loop or loops containing $i$ and $j$.

If $i$ and $j$ belong to different loops, the numerator factorizes into a product of two independent loop integrals, one carrying a single vector insertion from the loop containing $i$ and the other carrying a single vector insertion from the loop containing $j$. Each such integral is odd under $\boldsymbol{\Omega}_k\to-\boldsymbol{\Omega}_k$ on that loop and therefore vanishes. Hence, in a fixed loop state, two spins are correlated only if they lie on the same loop. This is the first key physical consequence of the loop representation: spin correlations
are not controlled by Euclidean distance alone, but by loop connectivity.

Let the common loop be denoted by $\Lambda$, and let its length be $\ell$. The loop length is even. This follows from the fact that the centers of crossed plaquettes on the checkerboard lattice and tetrahedra on the pyrochlore lattice form bipartite lattices (square and diamond, respectively); the singlet loops alternate between the two sublattices of this dual lattice.

The denominator in Eq.~\eqref{eq:corr_function_mod} is then
\begin{multline}
\label{eq:corr_function_denom}
\int
\prod_{k\in\Lambda} d\boldsymbol{\Omega}_k
\prod_{\langle m,n\rangle\in\Lambda}
\left(
1-\boldsymbol{\Omega}_m\cdot\boldsymbol{\Omega}_n
\right)
\\
=
\int
\prod_{k\in\Lambda} d\boldsymbol{\Omega}_k
\left[
1+
\prod_{\langle m,n\rangle\in\Lambda}
\boldsymbol{\Omega}_m\cdot\boldsymbol{\Omega}_n
\right]
\\
=
(4\pi)^\ell
\left(
1+3^{1-\ell}
\right).
\end{multline}
In the second line, only two terms survive the angular integrations: the empty term and the term containing every bond of the loop. All other terms contain at least one site with an odd number of $\boldsymbol{\Omega}$ factors and vanish.  The last line follows from
\begin{equation}
\label{eq:omega_integrals}
    \int d\boldsymbol{\Omega}\,
    \Omega^\mu
    =
    0,
    \qquad
    \int d\boldsymbol{\Omega}\,
    \Omega^\mu\Omega^\nu
    =
    \frac{4\pi}{3}\delta_{\mu\nu}.
\end{equation}

Now suppose that the two sites $i$ and $j$ are separated by $s\neq0$ steps along the shorter arc of the loop. The numerator receives two nonzero contributions, corresponding to the two paths on the loop that connect $i$ and $j$. These paths have lengths $s$ and $\ell-s$.
Using Eq.~\eqref{eq:omega_integrals}, one obtains
\begin{multline}
\label{eq:corr_function_num}
\int
\prod_{k\in\Lambda} d\boldsymbol{\Omega}_k\,
\boldsymbol{\Omega}_i\cdot\boldsymbol{\Omega}_j
\prod_{\langle m,n\rangle\in\Lambda}
\left(
1-\boldsymbol{\Omega}_m\cdot\boldsymbol{\Omega}_n
\right)
\\
=
(4\pi)^\ell
(-1)^s
\left(
3^{-s}+3^{s-\ell}
\right).
\end{multline}
The factor $(-1)^s$ is the staggered sign accumulated along the loop. The two terms in parentheses come from the two arcs connecting the two sites. Since, $\ell$ is even, the two paths carry the same overall staggered sign.

Combining Eqs.~\eqref{eq:corr_function_denom} and \eqref{eq:corr_function_num}, the spin--spin correlation function in a fixed loop state is
\begin{equation}
\label{eq:corr_function_loop}
\frac{
\langle\mathcal L|
\mathbf S_i\cdot\mathbf S_j
|\mathcal L\rangle
}{
\langle\mathcal L|\mathcal L\rangle
}
=
4(-1)^s
\frac{
3^{-s}+3^{s-\ell}
}{
1+3^{1-\ell}
},
\qquad
i,j\in\Lambda .
\end{equation}
If $i$ and $j$ lie on different loops, the same correlation function vanishes exactly:
\begin{equation}
\label{eq:corr_function_different_loops}
\frac{
\langle\mathcal L|
\mathbf S_i\cdot\mathbf S_j
|\mathcal L\rangle
}{
\langle\mathcal L|\mathcal L\rangle
}
=
0,
\qquad
i,j \ \text{on different loops}.
\end{equation}

Equation~\eqref{eq:corr_function_loop} reduces, in the limit $\ell\to\infty$, to the usual spin-$1$ AKLT-chain result along an infinite loop,
\begin{equation}
\label{eq:aklt_limit}
\frac{
\langle\mathcal L|
\mathbf S_i\cdot\mathbf S_j
|\mathcal L\rangle
}{
\langle\mathcal L|\mathcal L\rangle
}
\longrightarrow
4(-1)^s 3^{-s}.
\end{equation}
Thus, the correlation length along a loop is
\begin{equation}
    \xi_{\rm loop}=\frac{1}{\ln 3}.
\end{equation}
The importance of this result is that each fixed loop state is magnetically short-ranged, and the only memory of the underlying geometry is the connectivity of the loop passing through the two spins. In particular, two sites that are close in real space but lie on different loops are exactly uncorrelated in a fixed loop state, whereas two sites far apart in real space can remain weakly correlated
if they belong to the same long loop.

The form of Eq.~\eqref{eq:corr_function_loop} is reminiscent of the loop rules that appear in the valence-bond basis, where overlaps and spin correlations are organized by transition graphs~\cite{BeachSandvik2006}. There is, however, an important distinction from quantum dimer models: in the latter, the constrained dimer Hilbert space is imposed as the starting point \cite{Rokhsar1988,
MoessnerRaman2011QDM}, whereas here the loop manifold arises as the zero-energy sector of a microscopic SU(2)-invariant spin-$1$ parent Hamiltonian.

The extension from a fixed loop state to a general state in the ground-state manifold is more subtle. A product state in the local orthonormal basis,
\begin{equation}
    |\mathrm{GS}\rangle
    =
    \bigotimes_{\alpha\in A}
    |A_\alpha\rangle
    \bigotimes_{\beta\in B}
    |B_\beta\rangle ,
\end{equation}
with $A\cup B$ equal to the set of crossed plaquettes or tetrahedra, can be rewritten as a superposition of loop states. In such a superposition, the loop configurations do not all appear with the same sign. For example, when the state $|B\rangle$ is written as a difference of horizontal and vertical valence-bond configurations, vertical dimers contribute relative minus signs. Consequently,
\[
    |\Psi_{\mathrm{GS}}(\Omega)|^2
\]
contains off-diagonal terms of the form
\begin{equation}
    \left[
    \Psi_{\mathcal L'}(\Omega)
    \right]^*
    \Psi_{\mathcal L}(\Omega),
    \qquad
    \mathcal L'\neq\mathcal L,
\end{equation}
which are not positive term by term. Moreover, these terms involve
the complex factors $u_i v_j-v_i u_j$ themselves, rather than only their absolute values. The replacement used in Eq.~\eqref{eq:overlap2} is therefore valid for diagonal loop contributions or fixed loops (when $\mathcal L' = \mathcal L$), but not by itself for interference terms between different loop states.

For this reason, Eq.~\eqref{eq:corr_function_loop} should be read as an exact result for fixed loop states. It strongly suggests short-ranged spin correlations throughout the loop manifold, but a separate analysis is needed for general superpositions, where off-diagonal loop overlaps enter. In Appendix~\ref{app:stat_mech_loops} we show how these off-diagonal
coherent-state factors can be handled.  This distinction is the reason
for the numerical analysis in Sec.~\ref{sec:MC}: the fixed-loop result
gives exact geometric correlations for each basis state, while the
equal-amplitude superposition probes how these correlations survive
after summing over off-diagonal loop overlaps with positive weights.

\section{Additional interactions and effective Hamiltonians}
\label{sec:eff_Hamiltians}

\subsection{Additional nearest-neighbour Heisenberg interactions}
\label{sec::nn_Heisenberg}

The results of Sec.~\ref{sec:gs_correlations} show that spin
correlations in a fixed loop state are controlled by loop
connectivity.  For a general superposition of loop states, however,
the correlations depend on the relative amplitudes and signs with
which the different loop configurations enter the wavefunction.  It is
therefore useful to ask which superpositions are selected once the
parent Hamiltonian is perturbed.  We address this question by
projecting symmetry-allowed additional interactions into the
zero-energy manifold of $\mathcal H_0$.

The simplest perturbations are additional Heisenberg exchanges between
pairs of physical spin-$1$ moments.  The parent Hamiltonian
$\mathcal H_0$ already contains bilinear antiferromagnetic exchange
terms as part of the polynomial projector, together with higher-order
spin interactions.  Adding explicit nearest-neighbor Heisenberg terms therefore changes
the relative weight of the bilinear part within the full microscopic
Hamiltonian.  We write the nearest-neighbour perturbation as
\begin{equation}
\mathcal{H}_1 = \sum_{\langle i,j\rangle} J_{ij} \mathbf{S}_i\cdot\mathbf{S}_j
\label{eq:nn_Heisenberg}    
\end{equation}
and consider the total Hamiltonian $\mathcal H_0+\mathcal H_1$.

Throughout this subsection, the additional exchanges are assumed to
respect the spatial symmetries of the lattice.  On the checkerboard
lattice this allows two distinct nearest-neighbour couplings: $J_s$
for the side bonds of a crossed square and $J_d$ for the diagonal
bonds.  On the pyrochlore lattice, all nearest-neighbour bonds of a
tetrahedron are symmetry equivalent, so the corresponding perturbation
has a single nearest-neighbour coupling.

Our aim is to obtain the leading effective Hamiltonian generated by
$\mathcal H_1$ within the zero-energy manifold of $\mathcal H_0$.  The
global loop states form a convenient spanning set, but they are not an
orthonormal global basis.  We therefore use them to extract local
matrix elements, which are then expressed in the local orthonormal
basis $|A\rangle,|B\rangle$ defined in Eq.~\eqref{eq:basis_states}.
This gives the leading local contribution to the projected
perturbation.

The normalized matrix element of a spin bilinear between two loop
states can be written in the coherent-state representation as
\begin{multline}\label{eq:matrix_element_gen}
\langle{\mathcal{L}'} \lvert \mathbf{S}_i\cdot\mathbf{S}_j\rvert {\mathcal{L}}\rangle =
  \frac{\langle\Psi^{\mathcal{L}'} \lvert \mathbf{S}_i\cdot\mathbf{S}_j\rvert\Psi^{\mathcal{L}}\rangle}{\sqrt{\langle\Psi^{\mathcal{L}'} |\Psi^{\mathcal{L}'}\rangle\langle\Psi^{\mathcal{L}} |\Psi^{\mathcal{L}}\rangle}}\\
  = \frac{4}{\sqrt{Z'Z}} \int \prod_{k} d \boldsymbol{\Omega}_k  \left( \Psi^{\mathcal{L}'}(\Omega) \right)^\ast \!\Psi^{\mathcal{L}}(\Omega) \,\boldsymbol{\Omega}_i\cdot\boldsymbol{\Omega}_j \, ,
\end{multline}
where $Z$ and $Z'$ are the coherent-state norms of the
corresponding loop states.  For diagonal matrix elements,
$\mathcal L'=\mathcal L$, this expression reduces to the fixed-loop
correlation function derived in Eq.~\eqref{eq:corr_function_loop}.
These diagonal elements already determine a substantial part of the
local projected Hamiltonian.

Let us focus on a crossed plaquette or tetrahedron containing sites
$1,2,3,4$, with the labelling convention shown in
Fig.~\ref{fig:two_chequers}.  The same labelling will be used in the
next subsection, where we consider interactions between neighbouring
crossed plaquettes or tetrahedra.

\begin{figure}[h!]
  \centering
  \includegraphics[width=3cm]{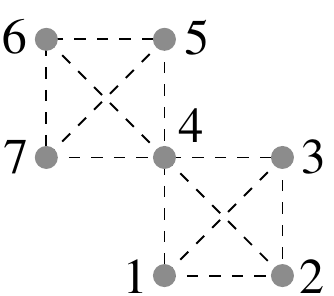}
  \caption{Two neighbouring crossed plaquettes with the site-labeling
convention used for the projected matrix elements.  Sites $1,2,3,4$
belong to the lower crossed plaquette, while sites $4,5,6,7$ belong to
the upper crossed plaquette.  The shared site is labelled $4$.}
\label{fig:two_chequers}
\end{figure}

In applying Eq.~\eqref{eq:corr_function_loop} to local matrix elements,
we keep the leading contribution from the shortest segment of the loop
connecting the two spins under consideration.  Equivalently, we neglect
the contribution from the complementary path around the loop, together
with the corresponding finite-loop correction in the denominator of
Eq.~\eqref{eq:corr_function_loop}.  The matrix elements below should
therefore be read as leading local-loop contributions to the projected
perturbation.  We indicate this by approximate equalities and comment
on the size of the omitted contribution where it is relevant.

We begin with the diagonal matrix element of
$\mathbf{S}_1\cdot\mathbf{S}_3$ in the diagonal valence-bond state
$\lvert\crossdimers\rangle$.  In this state, sites $1$ and $3$ are
nearest neighbours along the loop.  Thus Eq.~\eqref{eq:corr_function_loop}
with $s=1$ gives
\begin{equation}\label{eq:diag_element_13}
  \langle\crossdimers\left|\mathbf{S}_1\cdot\mathbf{S}_3\right|\crossdimers\rangle \approx - \frac{4}{3}
\end{equation}
within the leading local approximation.  The shortest loop containing
a diagonal dimer on the checkerboard lattice has length six, and the
shortest loop on the pyrochlore lattice also has length six.  Thus, the
omitted contribution from the complementary path along the loop is at
most of order $1/3^5$ in this case.  Similarly,
\begin{subequations}\label{eq:diag_element_1x}
\begin{equation}\label{eq:diag_element_12}
  \langle\hdimers\left|\mathbf{S}_1\cdot\mathbf{S}_2\right|\hdimers\rangle \approx - \frac{4}{3} \, ,
\end{equation}
\begin{equation}\label{eq:diag_element_14}
  \langle\vdimers\left|\mathbf{S}_1\cdot\mathbf{S}_4\right|\vdimers\rangle \approx - \frac{4}{3} \, ,
\end{equation}
\end{subequations}
where the corresponding correction on the checkerboard lattice can be
larger, of order $1/3^3$, because the shortest loop made only of side
bonds has length $\ell=4$.

With the same level of approximation, pairs of spins that are not
connected by the local segment of the loop give
\begin{subequations}\label{eq:offdiag_element_1x}
\begin{equation}\label{eq:offdiag_element_12}
  \langle\vdimers\left|\mathbf{S}_1\cdot\mathbf{S}_2\right|\vdimers\rangle
  = \langle\crossdimers\left|\mathbf{S}_1\cdot\mathbf{S}_2\right|\crossdimers\rangle \approx 0
\end{equation}
\begin{equation}\label{eq:offdiag_element_14}
 \langle\hdimers\left|\mathbf{S}_1\cdot\mathbf{S}_4\right|\hdimers\rangle
  = \langle\crossdimers\left|\mathbf{S}_1\cdot\mathbf{S}_4\right|\crossdimers\rangle \approx 0
\end{equation}
\begin{equation}\label{eq:offdiag_element_13}
 \langle\hdimers\left|\mathbf{S}_1\cdot\mathbf{S}_3\right|\hdimers\rangle
  = \langle\vdimers\left|\mathbf{S}_1\cdot\mathbf{S}_3\right|\vdimers\rangle \approx 0 \, .
\end{equation}
\end{subequations}
These vanishing matrix elements should also be understood in the
leading local-loop sense.  A nonzero contribution can arise only from
a longer path through the rest of the loop, which is not retained in
this local estimate.

We now translate these local loop-state matrix elements into the
orthonormal basis $|A\rangle,|B\rangle$.  The relevant relations are
\begin{equation}
\begin{aligned}
    |\crossdimers\rangle &= |B\rangle,\\
    |\hdimers\rangle &=
    \frac{\sqrt{3}|A\rangle+|B\rangle}{2},\\
    |\vdimers\rangle &=
    \frac{\sqrt{3}|A\rangle-|B\rangle}{2}.
\end{aligned}
\end{equation}
First consider $\mathbf{S}_1\cdot\mathbf{S}_3$.  The preceding matrix
elements imply
\begin{subequations}\label{eq:ort_elements_13}
\begin{equation}\label{eq:hor_element_13}
\frac{1}{4} \left[\sqrt{3}\langle A \vert + \langle B \vert \right]\mathbf{S}_1\cdot\mathbf{S}_3 \left[\sqrt{3}\vert A \rangle + \vert B \rangle \right]
   = 0
\end{equation}
\begin{equation}\label{eq:vert_element_13}
 \frac{1}{4} \left[\sqrt{3}\langle A \vert - \langle B \vert \right]\mathbf{S}_1\cdot\mathbf{S}_3 \left[\sqrt{3}\vert A \rangle - \vert B \rangle \right]
   =0
\end{equation}
\begin{equation}\label{eq:cross_element_13}
 \langle B\left|\mathbf{S}_1\cdot\mathbf{S}_3\right| B \rangle  = -\frac{4}{3} \, ,
\end{equation}
\end{subequations}
from which one obtains
\begin{subequations}\label{AB_elements_13}
\begin{equation}\label{eq:AA_element_13}
 \langle A\left|\mathbf{S}_1\cdot\mathbf{S}_3\right| A \rangle  = \frac{4}{9} \, ,
\end{equation}
\begin{equation}\label{eq:BB_element_13}
 \langle B\left|\mathbf{S}_1\cdot\mathbf{S}_3\right| B \rangle  = -\frac{4}{3} \, ,
\end{equation}
\begin{equation}\label{eq:AB_element_13}
\operatorname{Re}\left(\langle A\left|\mathbf{S}_1\cdot\mathbf{S}_3\right| B \rangle\right)  = 0 \, .
\end{equation}
\end{subequations}
The first two equations in Eq.~\eqref{eq:ort_elements_13} give the
same diagonal constraint with opposite signs for the real
off-diagonal matrix element.  Together with
Eq.~\eqref{eq:cross_element_13}, they fix the matrix elements in
Eq.~\eqref{AB_elements_13}.

The same procedure for the side bond $(12)$ gives
\begin{subequations}\label{eq:ort_elements_12}
\begin{equation}\label{eq:hor_element_12}
\frac{1}{4} \left[\sqrt{3}\langle A \vert + \langle B \vert \right]\mathbf{S}_1\cdot\mathbf{S}_2 \left[\sqrt{3}\vert A \rangle + \vert B \rangle \right]
   = -\frac{4}{3}
\end{equation}
\begin{equation}\label{eq:vert_element_12}
 \frac{1}{4} \left[\sqrt{3}\langle A \vert - \langle B \vert \right]\mathbf{S}_1\cdot\mathbf{S}_2 \left[\sqrt{3}\vert A \rangle - \vert B \rangle \right]
   =0
\end{equation}
\begin{equation}\label{eq:cross_element_12}
 \langle B\left|\mathbf{S}_1\cdot\mathbf{S}_2\right| B \rangle =0 \, ,
\end{equation}
\end{subequations}
and therefore
\begin{subequations}\label{AB_elements_12}
\begin{equation}\label{eq:AA_element_12}
 \langle A\left|\mathbf{S}_1\cdot\mathbf{S}_2\right| A \rangle  = -\frac{8}{9} \, ,
\end{equation}
\begin{equation}\label{eq:BB_element_12}
 \langle B\left|\mathbf{S}_1\cdot\mathbf{S}_2\right| B \rangle  = 0 \, ,
\end{equation}
\begin{equation}\label{eq:AB_element_12}
\operatorname{Re}\left(\langle A\left|\mathbf{S}_1\cdot\mathbf{S}_2\right| B \rangle\right)  = -\frac {4\sqrt{3}}{9} \, .
\end{equation}
\end{subequations}
For the other side bond attached to site $1$, one obtains
\begin{subequations}\label{AB_elements_14}
\begin{equation}\label{eq:AA_element_14}
 \langle A\left|\mathbf{S}_1\cdot\mathbf{S}_4\right| A \rangle  = -\frac{8}{9} \, ,
\end{equation}
\begin{equation}\label{eq:BB_element_14}
 \langle B\left|\mathbf{S}_1\cdot\mathbf{S}_4\right| B \rangle  = 0 \, ,
\end{equation}
\begin{equation}\label{eq:AB_element_14}
\operatorname{Re}\left(\langle A\left|\mathbf{S}_1\cdot\mathbf{S}_4\right| B \rangle\right)  = +\frac {4\sqrt{3}}{9} \, .
\end{equation}
\end{subequations}
The opposite signs of the real off-diagonal matrix elements in
Eqs.~\eqref{eq:AB_element_12} and \eqref{eq:AB_element_14} are
important: they are what allows these terms to cancel once the two
side bonds are assigned the same coupling by lattice symmetry.

The diagonal loop-state matrix elements determine the real diagonal
part of the local pseudospin Hamiltonian.  To complete the local
projection, one also needs the off-diagonal dimer matrix elements.
For two spins in the same crossed plaquette there are two independent
ones:
\begin{equation}\label{eq:offdiag_nn_elements_1}
  \langle\vdimers\left|\mathbf{S}_1\cdot\mathbf{S}_2\right|\hdimers\rangle
  \approx -\frac{2}{3}
\end{equation}
\begin{equation}\label{eq:offdiag_nn_elements_2}
  \langle\vdimers\left|\mathbf{S}_1\cdot\mathbf{S}_3\right|\hdimers\rangle
  \approx \frac{2}{3} \, ,
\end{equation}
with the remaining matrix elements are related to these by lattice
symmetries.  Appendix~\ref{app:matrix_elements_calculation} gives a
representative calculation of these off-diagonal dimer matrix
elements.

When these off-diagonal dimer elements are re-expressed in the
orthonormal $|A\rangle,|B\rangle$ basis, one finds
\begin{equation}\label{eq:offdiag_nn_elements_Im}
\operatorname{Im}\left(\langle A\left|\mathbf{S}_1\cdot\mathbf{S}_2\right| B \rangle\right)
=
\operatorname{Im}\left(\langle A\left|\mathbf{S}_1\cdot\mathbf{S}_3\right| B \rangle\right)
=
0 .
\end{equation}

We now combine the bond contributions subject to lattice symmetry.  On
the checkerboard lattice, the two side bonds
$J_{12}\mathbf{S}_1\cdot\mathbf{S}_2$ and
$J_{14}\mathbf{S}_1\cdot\mathbf{S}_4$ have the same coupling,
\[
    J_{12}=J_{14}=J_s .
\]
Consequently, the off-diagonal contributions in
Eqs.~\eqref{eq:AB_element_12} and \eqref{eq:AB_element_14} cancel in
the symmetry-preserving nearest-neighbour perturbation.  The diagonal
bond $J_{13}\mathbf{S}_1\cdot\mathbf{S}_3$ may have a different
coupling on the checkerboard lattice; we write
\[
    J_{13}=J_d ,
\]
where $s$ and $d$ denote side and diagonal bonds, respectively.

The effective Hamiltonian in the local orthonormal basis
$|A\rangle,|B\rangle$ can therefore be written in terms of a pseudospin
$\tau^z$, with
\[
    \tau^z |A\rangle = + |A \rangle,
    \qquad
    \tau^z |B\rangle = - |B\rangle .
\]
On the checkerboard lattice, each crossed plaquette contains four side
bonds and two diagonal bonds, and these bonds are not shared with
neighbouring crossed plaquettes.  Thus the full nearest-neighbour
contribution is obtained by summing all six bond contributions within
each crossed plaquette.  Using Eqs.~\eqref{eq:AB_element_12} and
\eqref{eq:AB_element_14}, the four side bonds give
$-(16/9)J_s\,\tau^z_\alpha$, while Eq.~\eqref{eq:AB_element_13}
shows that the two diagonal bonds give
$(16/9)J_d\,\tau^z_\alpha$, up to an additive constant.  Therefore,
discarding this additive constant, the leading nearest-neighbour
contribution on the checkerboard lattice is
\begin{equation}
\label{eq:effective_ham_nn}
  H_\text{eff}
  =
  -\frac{16}{9}\left(J_s-J_d\right)
  \sum_{\alpha} \tau_\alpha^z .
\end{equation}
Here, the sum runs over crossed plaquettes $\alpha$.

This counting also makes clear why no transverse one-body term appears
in Eq.~\eqref{eq:effective_ham_nn}.  The diagonal bonds have no real
off-diagonal matrix element in the $|A\rangle,|B\rangle$ basis, as
shown in Eq.~\eqref{eq:AB_element_13}.  The side bonds do have real
off-diagonal matrix elements, but the contributions from symmetry-related
horizontal and vertical side bonds cancel pairwise, as follows from
Eqs.~\eqref{eq:AB_element_12} and \eqref{eq:AB_element_14}.  Hence, no
$\tau^x$ term remains in the symmetry-preserving nearest-neighbour
perturbation. The absence of a $\tau^y$ term follows separately from
Eq.~\eqref{eq:offdiag_nn_elements_Im}: the relevant off-diagonal matrix
elements are real. This is consistent with time-reversal symmetry of
the perturbation. Since, the local valence-bond basis is chosen to be
real, a single $\tau^y$ term would require an imaginary off-diagonal
matrix element in this basis.

Equation~\eqref{eq:effective_ham_nn} has a simple physical meaning:
the imbalance between side and diagonal Heisenberg exchanges acts as a
local pseudospin field.  If $J_s>J_d$, the field favours
$\tau^z=+1$, corresponding to the $|A\rangle$ state on every crossed
plaquette, namely the equal-amplitude superposition of horizontal and
vertical singlets.  If $J_d>J_s$, the field favours $\tau^z=-1$,
corresponding to the diagonal state
$|B\rangle=|\crossdimers\rangle$.

The pyrochlore case is qualitatively different.  There all
nearest-neighbour bonds of a tetrahedron are symmetry equivalent, so
the analogue of $J_s-J_d$ vanishes.  Hence a symmetry-preserving
nearest-neighbour Heisenberg perturbation does not split the local
$|A\rangle,|B\rangle$ doublet at this order.  This is the local
symmetry reason behind the absence of a nearest-neighbour pseudospin
field on the pyrochlore lattice: a term linear in $\tau^z$ would
select a particular valence-bond orientation and thereby break the
tetrahedral rotational symmetry.  To lift the degeneracy in a
symmetry-preserving way on the pyrochlore lattice, one must therefore
consider less local or higher-order perturbations, such as the
next-nearest-neighbour interactions discussed below.

\subsection{Next-nearest-neighbour additional Heisenberg interactions}
\label{sec:nnn_effective_H}

The nearest-neighbour perturbation discussed in
Sec.~\ref{sec::nn_Heisenberg} has an important consequence: on the
pyrochlore lattice it does not lift the local degeneracy of the
$|A\rangle,|B\rangle$ doublet, because all nearest-neighbour bonds of
a tetrahedron are symmetry equivalent.  The first nontrivial
symmetry-allowed degeneracy lifting on the pyrochlore lattice must
therefore come from less local perturbations.  The simplest such
perturbations are additional next-nearest-neighbour Heisenberg
interactions between physical spin-$1$ moments.

We first perform the calculation for the checkerboard lattice, where
the lower point-group symmetry allows several independent
next-nearest-neighbour couplings.  The pyrochlore result then follows
by imposing the stronger tetrahedral symmetry constraints.  The
calculation follows the same logic as in the nearest-neighbour case:
we compute the leading local matrix elements in the dimer basis, then
convert them to the orthonormal pseudospin basis
$|A\rangle,|B\rangle$.

For next-nearest-neighbour spins we again use
Eqs.~\eqref{eq:corr_function_loop} and
\eqref{eq:matrix_element_gen}, now with separation $s=2$ along the
local loop segment.  As before, we retain only the shortest segment of
the loop connecting the two spins under consideration.  In
Eq.~\eqref{eq:corr_function_loop}, this means keeping the $3^{-s}$
term in the numerator and neglecting the complementary contribution
$3^{s-\ell}$, as well as the finite-loop correction $3^{1-\ell}$ in
the denominator.  The denominator approximation is well controlled,
since $\ell\geq 4$ on the checkerboard lattice and $\ell\geq 6$ on the
pyrochlore lattice.  The numerator requires a little more care.

On the checkerboard lattice, the pair of spins $3$ and $5$ in
Fig.~\ref{fig:two_chequers} can be connected by a loop of length
$\ell=4$ around an empty plaquette.  In that case the two paths between
the spins have equal length, so the complementary contribution is equal
in magnitude to the shortest-segment contribution that we keep.  This
does not signal a failure of the construction.  Rather, it reflects the
fact that the projected Hamiltonian is being organized as a sum of
local terms associated with neighbouring crossed-plaquette pairs.  The
same pair of physical spins, which are diagonally opposite across an empty plaquette, belongs to two distinct neighbouring
crossed-plaquette pairs, and therefore appears in two local terms of
the effective Hamiltonian.  Keeping only one segment in each local term
is the correct local decomposition of this contribution; including both
segments in a single matrix element would double-count the same
physical interaction when the effective Hamiltonian is assembled.
All other spin pairs considered below contribute to a single local term
in the effective Hamiltonian; their shortest loops have length at least
$8$ on the checkerboard lattice and $6$ on the pyrochlore lattice.

\begin{table*}[t] 
    \caption{Leading local matrix elements in the dimer representation for
next-nearest-neighbour spin interactions on the checkerboard and
pyrochlore lattices.  The column headings specify the physical spin
bilinear.  In a state such as
$|\vdimers,\hdimers\rangle$, the first dimer configuration refers to
the lower crossed plaquette, or lower tetrahedron, in
Fig.~\ref{fig:two_chequers} or Fig.~\ref{fig:two_tetrahedra}(a), while
the second refers to the upper one.  The same ordering convention is
used for bra states.  The entries are the leading shortest-loop
contributions obtained from the loop matrix-element rules of
Appendix~\ref{app:matrix_elements_calculation}; matrix elements not
listed vanish within this local approximation.}
\label{table:matrix_elements}
    \resizebox{\textwidth}{!}{%
    \begin{tabular}{|c||*{9}{c|}}
        \hline
        Matrix elements & $\mathbf{S}_1 \cdot \mathbf{S}_5$ & $\mathbf{S}_1 \cdot \mathbf{S}_6$ & $\mathbf{S}_1 \cdot \mathbf{S}_7$ & $\mathbf{S}_2 \cdot \mathbf{S}_5$ & $\mathbf{S}_2 \cdot \mathbf{S}_6$ & $\mathbf{S}_2 \cdot \mathbf{S}_7$ & $\mathbf{S}_3 \cdot \mathbf{S}_5$ & $\mathbf{S}_3 \cdot \mathbf{S}_6$ & $\mathbf{S}_3 \cdot \mathbf{S}_7$\tabularnewline
        \hline
        $\langle \vdimers, \vdimers |\dots| \vdimers, \vdimers \rangle$   & ${4}/{9}$    & $0$ & $0$ & $0$ & $0$ & $0$ & $0$ & $0$ & $0$\tabularnewline
        $\langle \vdimers, \hdimers |\dots| \vdimers, \hdimers \rangle$ & $0$ & $0$ & ${4}/{9}$ & $0$ & $0$ & $0$ & $0$ & $0$ & $0$\tabularnewline
        $\langle \hdimers, \vdimers |\dots| \hdimers, \vdimers \rangle$ & $0$ & $0$ & $0$ & $0$ & $0$ & $0$ & ${4}/{9}$ & $0$ & $0$\tabularnewline
        $\langle \hdimers, \hdimers |\dots| \hdimers, \hdimers \rangle$ & $0$ & $0$ & $0$ & $0$ & $0$ & $0$ & $0$ & $0$ & ${4}/{9}$\tabularnewline
        \hline
        $\langle \vdimers, \vdimers |\dots| \vdimers, \hdimers \rangle$ & ${2}/{9}$ & $-{2}/{9}$ & ${2}/{9}$ & $0$ & $0$ & $0$ & $0$ & $0$ & $0$\tabularnewline
        $\langle \vdimers, \vdimers |\dots| \hdimers, \vdimers \rangle$ & ${2}/{9}$ & $0$ & $0$ & $-{2}/{9}$ & $0$ & $0$ & ${2}/{9}$ & $0$ & $0$\tabularnewline
        $\langle \vdimers, \hdimers |\dots| \hdimers, \hdimers \rangle$ & $0$ & $0$ & ${2}/{9}$ & $0$ & $0$ & $-{2}/{9}$ & $0$ & $0$ & ${2}/{9}$\tabularnewline
        $\langle \hdimers, \vdimers |\dots| \hdimers, \hdimers \rangle$ & $0$ & $0$ & $0$ & $0$ & $0$ & $0$ & ${2}/{9}$ & $-{2}/{9}$ & ${2}/{9}$\tabularnewline
        \hline
        $\langle \vdimers, \vdimers |\dots| \hdimers, \hdimers \rangle$ & ${1}/{9}$ & $-{1}/{9}$ & ${1}/{9}$ & $-{1}/{9}$ & ${1}/{9}$ & $-{1}/{9}$ & ${1}/{9}$ & $-{1}/{9}$ & ${1}/{9}$\tabularnewline
        $\langle \vdimers, \hdimers |\dots| \hdimers, \vdimers \rangle$ & ${1}/{9}$ & $-{1}/{9}$ & ${1}/{9}$ & $-{1}/{9}$ & ${1}/{9}$ & $-{1}/{9}$ & ${1}/{9}$ & $-{1}/{9}$ & ${1}/{9}$\tabularnewline
        \hline
    \end{tabular}%
    }
\end{table*}

The dimer-basis matrix elements are listed in
Table~\ref{table:matrix_elements}.  In the leading local approximation
described above, all matrix elements not shown in the table either vanish or are related by complex conjugation.
The representative calculation of these entries is given in
Appendix~\ref{app:matrix_elements_calculation}.  The conversion from
the dimer basis to the local orthonormal basis $|A\rangle,|B\rangle$
is given in Appendix~\ref{app:matrix_elements_basis_conversion}.

For the checkerboard lattice, the symmetries of a pair of neighbouring
crossed plaquettes allow four independent next-nearest-neighbour
couplings:
\begin{subequations}\label{eq:checkerboard_coupling_def}
\begin{align}
    J^{(1)} &= J_{17} = J_{35} \, ,\\
    J^{(2)} &= J_{15} = J_{37} \, ,\\
    J^{(3)} &= J_{16} = J_{36} = J_{25} = J_{27} \, ,\\
    J^{(4)} &= J_{26} \, .
\end{align}
\end{subequations}
These four classes correspond to symmetry-inequivalent ways of
connecting the two crossed plaquettes by a next-nearest-neighbour spin
pair.

We now express the result in the two-pseudospin basis.  With the already adopted convention $|A\rangle\equiv |\!\uparrow\rangle$, $|B\rangle\equiv |\!\downarrow\rangle$ for one crossed
plaquette or tetrahedron, for a pair of neighbouring plaquettes/tetrahedra $(\alpha,\beta)$ we use the ordered basis
\[
    |\!\uparrow \uparrow\rangle,\quad
    |\!\downarrow \uparrow\rangle,\quad
    |\!\uparrow \downarrow\rangle,\quad
    |\!\downarrow \downarrow\rangle .
\]
Here, the first pseudospin belongs to the lower crossed plaquette or
tetrahedron $\alpha$, and the second to the neighbouring upper crossed
plaquette or tetrahedron $\beta$ [of Figs.~\ref{fig:two_chequers} and~\ref{fig:two_tetrahedra}].  This ordering convention is used
in constructing the matrices below.  Combining the matrices derived in
Appendix~\ref{app:matrix_elements_basis_conversion}, we obtain
\begin{multline}
    H_{\text{eff}} = J^{(1)} H^{(1)} + J^{(2)} H^{(2)} + J^{(3)} H^{(3)} + J^{(4)} H^{(4)}\\
    \; = J^{(1)} \begin{pmatrix}
                    \frac{32}{81} & 0 & 0 & -\frac{8}{27}\\
                    0 & 0 & -\frac{8}{27} & 0\\
                    0 & -\frac{8}{27} & 0 & 0\\
                    -\frac{8}{27} & 0 & 0 & 0\\
                \end{pmatrix} \qquad\qquad\qquad\qquad\\
          + J^{(2)} \begin{pmatrix}
                    \frac{32}{81} & 0 & 0 & \frac{8}{27}\\
                    0 & 0 & \frac{8}{27} & 0\\
                    0 & \frac{8}{27} & 0 & 0\\
                    \frac{8}{27} & 0 & 0 & 0\\
                \end{pmatrix}
         + J^{(3)} \begin{pmatrix}
                    -\frac{32}{81} & 0 & 0 & 0\\
                    0 & \frac{16}{27} & 0 & 0\\
                    0 & 0 & \frac{16}{27} & 0\\
                    0 & 0 & 0 & 0\\
                \end{pmatrix}\\
         + J^{(4)} \begin{pmatrix}
                    \frac{4}{81} & 0 & 0 & 0\\
                    0 & -\frac{4}{27} & 0 & 0\\
                    0 & 0 & -\frac{4}{27} & 0\\
                    0 & 0 & 0 & \frac{4}{9}\\
                \end{pmatrix} \, .
\end{multline}

Several features are already visible at this stage.  All matrices are
real.  Moreover, they contain only diagonal and anti-diagonal entries.
The diagonal entries will become constants, local fields, and
$\tau^z_\alpha\tau^z_\beta$ interactions, while the anti-diagonal
entries flip both pseudospins simultaneously and therefore generate
$\tau^x_\alpha\tau^x_\beta$ terms.  The absence of other structures is
a useful diagnostic: at this order the next-nearest-neighbour
Heisenberg perturbation generates a very restricted pseudospin
Hamiltonian.

Using
\begin{align*}
\begin{aligned}
\tau^z \otimes \mathbb{I} &=
\begin{pmatrix}
1 & 0 & 0 & 0\\
0 & -1 & 0 & 0\\
0 & 0 & 1 & 0\\
0 & 0 & 0 & -1
\end{pmatrix},
\;\;
\mathbb{I} \otimes \tau^z =
\begin{pmatrix}
1 & 0 & 0 & 0\\
0 & 1 & 0 & 0\\
0 & 0 & -1 & 0\\
0 & 0 & 0 & -1
\end{pmatrix},
\\[1em]
\tau^z \otimes \tau^z &=
\begin{pmatrix}
1 & 0 & 0 & 0\\
0 & -1 & 0 & 0\\
0 & 0 & -1 & 0\\
0 & 0 & 0 & 1
\end{pmatrix},
\;\;
\tau^x \otimes \tau^x =
\begin{pmatrix}
0 & 0 & 0 & 1\\
0 & 0 & 1 & 0\\
0 & 1 & 0 & 0\\
1 & 0 & 0 & 0
\end{pmatrix}
\end{aligned}
\end{align*}
in the $|\!\uparrow \uparrow\rangle$, $|\!\downarrow \uparrow\rangle$, $|\!\uparrow \downarrow\rangle$, $|\!\downarrow \downarrow\rangle$
basis, we obtain the effective Hamiltonian as 

\begin{align}
\begin{split}
    H_{\text{eff}} =& J^{(1)} \left( \frac{8}{81} \tau^z_{\alpha} \otimes \tau^z_{\beta} + \frac{8}{81} \tau^z_{\alpha} \otimes \mathbb{I}_{\beta} + \frac{8}{81} \mathbb{I}_{\alpha} \otimes \tau^z_{\beta} \right. \\
    & \left. \qquad \qquad - \frac{8}{27} \tau^x_{\alpha} \otimes \tau^x_{\beta} + \frac{8}{81} \mathbb{I}_{\alpha} \otimes \mathbb{I}_{\beta} \right)\\
    & + J^{(2)} \left( \frac{8}{81} \tau^z_{\alpha} \otimes \tau^z_{\beta} + \frac{8}{81} \tau^z_{\alpha} \otimes \mathbb{I}_{\beta} \right. \\
    & \left. \quad + \frac{8}{81} \mathbb{I}_{\alpha} \otimes \tau^z_{\beta} + \frac{8}{27} \tau^x_{\alpha} \otimes \tau^x_{\beta} + \frac{8}{81} \mathbb{I}_{\alpha} \otimes \mathbb{I}_{\beta} \right)\\
    & + J^{(3)} \left( -\frac{32}{81} \tau^z_{\alpha} \otimes \tau^z_{\beta} - \frac{8}{81} \tau^z_{\alpha} \otimes \mathbb{I}_{\beta} \right.\\
    & \left. \qquad \qquad - \frac{8}{81} \mathbb{I}_{\alpha} \otimes \tau^z_{\beta} + \frac{16}{81} \mathbb{I}_{\alpha} \otimes \mathbb{I}_{\beta} \right)\\
    & + J^{(4)} \left( \frac{16}{81} \tau^z_{\alpha} \otimes \tau^z_{\beta} - \frac{8}{81} \tau^z_{\alpha} \otimes \mathbb{I}_{\beta} \right. \\
    & \left. \qquad \qquad - \frac{8}{81} \mathbb{I}_{\alpha} \otimes \tau^z_{\beta} + \frac{4}{81} \mathbb{I}_{\alpha} \otimes \mathbb{I}_{\beta} \right) \, .
\end{split}
\end{align}
The anti-diagonal part of these matrices is proportional to
$\tau^x\otimes\tau^x$.  By contrast,
\[
        \tau^y \otimes \tau^y = \begin{pmatrix}
                    0 & 0 & 0 & -1\\
                    0 & 0 & 1 & 0\\
                    0 & 1 & 0 & 0\\
                    -1 & 0 & 0 & 0\\
                \end{pmatrix} \, 
\] 
does not enter the effective Hamiltonian. This absence is not merely a consequence of Hermiticity. 
Hermiticity would allow both $\tau^x_\alpha\tau^x_\beta$ and
$\tau^y_\alpha\tau^y_\beta$. What matters here is the relative sign of
the two anti-diagonal channels. The projected Heisenberg matrix
elements flip the two local valence-bond pseudospins with the sign
structure of $\tau^x_\alpha\tau^x_\beta$, not that of
$\tau^y_\alpha\tau^y_\beta$. Equivalently, in the real
valence-bond basis used here, the perturbation produces real
pseudospin-flip amplitudes whose anti-diagonal signs match
$\tau^x_\alpha\tau^x_\beta$.

\begin{figure}[t!]
    \centering
    \includegraphics[width=\linewidth,
    trim=0 0 0 0, clip
    ]{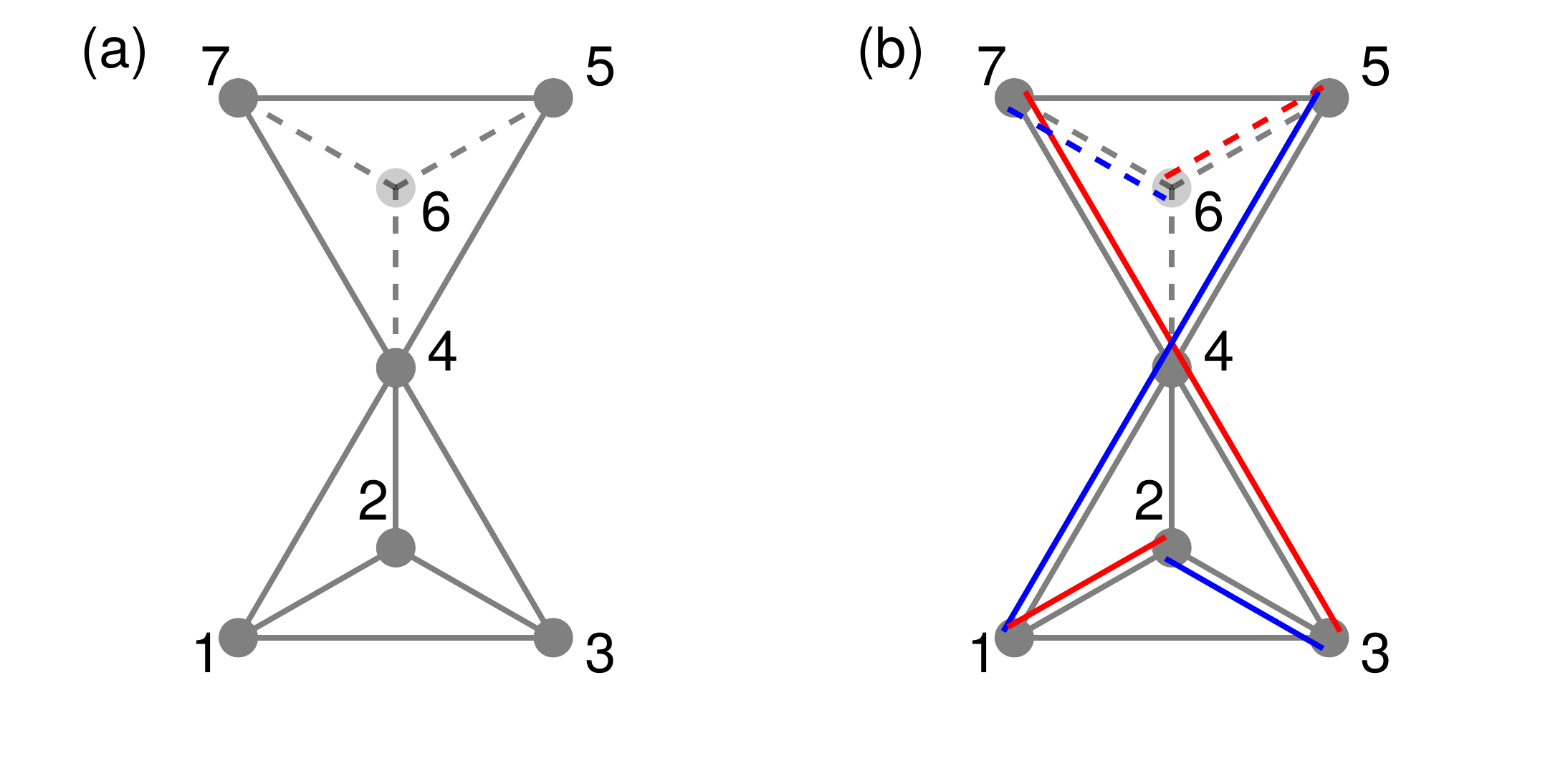}
    \caption{Pair of corner-sharing tetrahedra used for the pyrochlore
projection.  (a) Site labels are chosen to match the convention for the
crossed-plaquette pair in Fig.~\ref{fig:two_chequers}: the site
opposite $1$ is labelled $5$, and similarly for the other sites.
(b) A representative bra--ket singlet configuration on the same pair of
tetrahedra.}
\label{fig:two_tetrahedra}
\end{figure}

The expression above is the projected Hamiltonian for one neighbouring
pair of crossed plaquettes. For such a pair, the one-body part appears
as a field on both pseudospins in the pair. Thus, the local pair
contribution can be written, up to an additive constant, as
\begin{equation*}
\begin{aligned}
    h^\text{cb}_{\alpha\beta}
    &=
    \frac{8}{81}
    \Big[
    3\left(J^{(2)}-J^{(1)}\right)
    \tau^x_\alpha\tau^x_\beta
    \\
    &\hspace{1.4cm}
    +
    \left(J^{(1)}+J^{(2)}-4J^{(3)}+2J^{(4)}\right)
    \tau^z_\alpha\tau^z_\beta
    \Big]
    \\
    &\quad
    +
    \frac{8}{81}
    \left(J^{(1)}+J^{(2)}-J^{(3)}-J^{(4)}\right)
    \left(\tau^z_\alpha+\tau^z_\beta\right).
\end{aligned}
\end{equation*}
The full checkerboard effective Hamiltonian is obtained by summing this
local contribution over all nearest-neighbour pairs of crossed
plaquettes. Since, the crossed-plaquette centers form a square lattice,
each crossed plaquette has coordination number
\(z_\text{cb}=4\). Therefore,
\[
    \sum_{\langle\alpha,\beta\rangle}
    \left(\tau^z_\alpha+\tau^z_\beta\right)
    =
    z_\text{cb}\sum_\alpha \tau^z_\alpha .
\]
With this convention, the full checkerboard Hamiltonian is
\begin{multline} \label{eq:eff_Hamiltonian_checkerboard}
    H^\text{cb}_{\text{eff}}
    =
    \frac{8}{81}
    \sum_{\langle \alpha, \beta \rangle}
    \left[
    3 \left( J^{(2)} -J^{(1)} \right)
    \tau^x_{\alpha} \tau^x_{\beta}
    \right.
    \\
    \left.
    +
    \left( J^{(1)} +  J^{(2)} -4 J^{(3)} + 2 J^{(4)} \right)
    \tau^z_{\alpha} \tau^z_{\beta}
    \right]
    \\
    +
    \frac{32}{81}
    \left( J^{(1)} +  J^{(2)} - J^{(3)} -  J^{(4)} \right)
    \sum_{\alpha} \tau^z_{\alpha}
    +{\rm const.}
\end{multline}
Equation~\eqref{eq:eff_Hamiltonian_checkerboard} acts on pseudospins
$\tau_\alpha$ residing on the square lattice formed by the centers of
the crossed plaquettes.  It contains an Ising-like
$\tau^x_\alpha\tau^x_\beta$ term, a second Ising-like
$\tau^z_\alpha\tau^z_\beta$ term, and a field along the $\tau^z$
direction.  Thus the next-nearest-neighbour perturbation on the
checkerboard lattice produces an anisotropic two-component pseudospin
model rather than a simple resonance term.  In general this is only a
correction to the leading local field already generated by the
nearest-neighbour perturbation in Eq.~\eqref{eq:effective_ham_nn}, so
we do not attempt to map out its full phase diagram here.

The situation is much more constrained, and more revealing, on the
pyrochlore lattice.  There the nearest-neighbour perturbation produces
no local pseudospin field, and tetrahedral symmetry further reduces
the independent next-nearest-neighbour couplings by imposing
$J^{(1)}=J^{(3)}$ and $J^{(2)}=J^{(4)}$.  For the pair of tetrahedra
shown in Fig.~\ref{fig:two_tetrahedra}, the two allowed coupling
classes are
\begin{subequations} \label{eq:pyrochlore_coupling_def}
\begin{align}
    J^{(1)} &= J_{16} = J_{36} = J_{25} = J_{27} = J_{17} = J_{35} \, ,\\
    J^{(2)} &= J_{15} = J_{26} = J_{37} \, .
\end{align}
\end{subequations}
Substituting these symmetry constraints into the two-pseudospin
Hamiltonian produces a substantial simplification.  Three things happen
simultaneously.  First, the one-body pseudospin field cancels.  Second,
the anisotropy between the $\tau^x_\alpha\tau^x_\beta$ and
$\tau^z_\alpha\tau^z_\beta$ channels disappears.  Third, all
orientation-dependent information collapses into the two coupling
classes $J^{(1)}$ and $J^{(2)}$. Thus, a rather complicated set of projected two-spin matrix elements collapses, by tetrahedral symmetry, to a single U(1)-symmetric
nearest-neighbour pseudospin interaction on the diamond lattice:
\begin{equation} \label{eq:eff_Hamiltonian_pyrochlore}
H_{\text{eff}} = \frac{8}{27} \left( J^{(2)} -J^{(1)} \right) \sum_{\langle \alpha, \beta \rangle}  \left(\tau_{\alpha}^x \tau_{\beta}^x + \tau_{\alpha}^z \tau_{\beta}^z\right) + \text{const}.
\end{equation}
The simplification in Eq.~\eqref{eq:eff_Hamiltonian_pyrochlore} is a
key consequence of the pyrochlore symmetry.  In contrast with the checkerboard
case, a local pseudospin field is not generated on the pyrochlore lattice by the next-nearest-neighour Heisenberg perturbations, just as it was not generated by the nearest-neigbour terms, see Section~\ref{sec::nn_Heisenberg}.   The leading degeneracy-lifting term is instead an
interaction between neighbouring tetrahedra.  Thus the two-dimensional
singlet space on each tetrahedron acts as an effective spin-$\tfrac12$
degree of freedom living on the diamond lattice of tetrahedron centers.  The leading symmetry-allowed next-nearest-neighbour
Heisenberg perturbation does not merely select a static valence-bond
orientation; it generates a genuine collective pseudospin dynamics.
Moreover, the dynamics has an emergent continuous symmetry: the
Hamiltonian is invariant under rotations in the
$\tau^x$--$\tau^z$ plane, generated by $\tau^y$.  Equivalently, after
a relabelling of pseudospin axes, it is the standard quantum
spin-$\tfrac12$ XY model on the diamond lattice.  The coupling may be
ferromagnetic or antiferromagnetic depending on the sign of
$J^{(2)}-J^{(1)}$.

The absence of an effective Zeeman field on the pyrochlore lattice is
as important as the form of the interaction itself.  It means that no
particular local valence-bond orientation is selected by the perturbations considered in this and the previous sections. This is a natural consequence of the fact that these perturbations do not break any pyrochlore lattice symmetries and hence cannot result in an effective Hamiltonian that does. 
Instead, the next-nearest-neighbour Heizenberg perturbation promotes coherent mixing of the two local
singlet states on neighbouring tetrahedra.  This is the key difference
from the checkerboard case.

\subsubsection{Effective interactions for different orientations}
\label{sec:nnn_orientations}

The derivation above used the symmetries of the checkerboard and
pyrochlore lattices to restrict the number of independent microscopic
couplings.  
Before interpreting Eqs.~\eqref{eq:eff_Hamiltonian_checkerboard} and~\eqref{eq:eff_Hamiltonian_pyrochlore}, we must
check one further point.  The local pseudospin basis
$|A\rangle,|B\rangle$ is a convenient basis, but it is not invariant
under rotations of a crossed plaquette or tetrahedron.  A spatial
rotation acts nontrivially inside the two-dimensional local
valence-bond space.  Thus, even if the microscopic Heisenberg
interaction is fully symmetric, its projection could in principle
generate bond-dependent directional pseudospin couplings.  In other words, the
projection might have produced a compass-like or spin-orbit-like
effective Hamiltonian rather than the simple form obtained above, as seen in related
pyrochlore quantum dimer models~\cite{Moessner2006a}.

It is therefore nontrivial that the present construction does not lead
to an orientation-dependent pseudospin Hamiltonian at this order.  We
show this explicitly below.  In deriving
Eqs.~\eqref{eq:eff_Hamiltonian_checkerboard} and
\eqref{eq:eff_Hamiltonian_pyrochlore}, we considered the pairs shown
in Figs.~\ref{fig:two_chequers} and \ref{fig:two_tetrahedra}.  Other
pairs are obtained by a $C_4$ rotation for the checkerboard lattice
and a $C_3$ rotation for the pyrochlore lattice.  Since the local
pseudospin basis itself rotates nontrivially, the matrix elements must
be checked rather than assumed to be identical.

\begin{figure*}[t!]
    \centering
      \includegraphics[width=0.9\textwidth
      ]{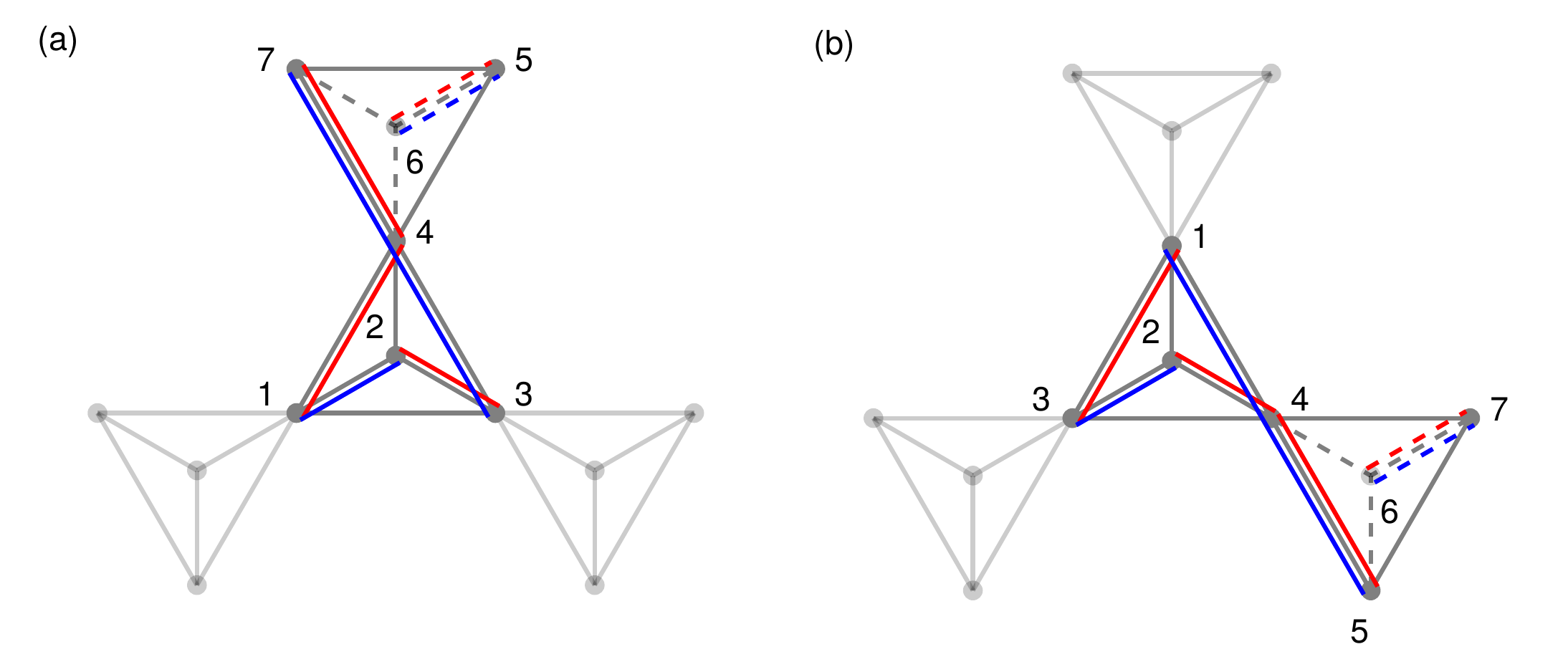}
      \caption{Effect of a $C_3$ rotation on a representative pyrochlore
matrix element.  (a) Bra and ket singlet configurations
$\langle\hdimers,\hdimers|$ and
$|\vdimers,\hdimers\rangle$ on the tetrahedron pair of
Fig.~\ref{fig:two_tetrahedra}.  (b) The corresponding configurations
after a clockwise $120^\circ$ rotation.  The rotation permutes the site
labels and hence the individual matrix elements, but preserves the
symmetry-equivalent coupling classes used in the pyrochlore effective
Hamiltonian.}
\label{fig:tetrahedron_pairs}
\end{figure*}

Figure~\ref{fig:tetrahedron_pairs} shows the same pair of singlet
configurations on two adjacent tetrahedra in two different
orientations.  The numbering of the sites in the two panels is related
by a $2\pi/3$ rotation about the axis perpendicular to the page and
passing through site $2$.  The leading two-spin matrix elements are
calculated by the same loop-resolution procedure as before
[see Appendix~\ref{app:matrix_elements_calculation}].  The results are
listed in Table~\ref{table:rotated_matrix_elements}.

\begin{table*}[t!] 
    \caption{Leading local matrix elements in the dimer representation for
the pyrochlore lattice after the $C_3$ rotation shown in
Fig.~\ref{fig:tetrahedron_pairs}(b).  The entries should be compared
with Table~\ref{table:matrix_elements}.  The rotation permutes the site
labels and hence the individual matrix elements, while keeping them
within the symmetry-equivalent coupling sectors defined in
Eq.~\eqref{eq:pyrochlore_coupling_def}.}
\label{table:rotated_matrix_elements}
    \resizebox{\textwidth}{!}{%
    \begin{tabular}{|c||*{9}{c|}}
        \hline
        Matrix elements & $\mathbf{S}_1 \cdot \mathbf{S}_5$ & $\mathbf{S}_1 \cdot \mathbf{S}_6$ & $\mathbf{S}_1 \cdot \mathbf{S}_7$ & $\mathbf{S}_2 \cdot \mathbf{S}_5$ & $\mathbf{S}_2 \cdot \mathbf{S}_6$ & $\mathbf{S}_2 \cdot \mathbf{S}_7$ & $\mathbf{S}_3 \cdot \mathbf{S}_5$ & $\mathbf{S}_3 \cdot \mathbf{S}_6$ & $\mathbf{S}_3 \cdot \mathbf{S}_7$\tabularnewline
        \hline
        $\langle \vdimers, \vdimers |\dots| \vdimers, \vdimers \rangle$   & $0$ & $0$ & $0$ & $0$ & ${4}/{9}$ & $0$ & $0$ & $0$ & $0$\tabularnewline
        $\langle \vdimers, \hdimers |\dots| \vdimers, \hdimers \rangle$ & $0$ & $0$ & $0$ & ${4}/{9}$ & $0$ & $0$ & $0$ & $0$ & $0$\tabularnewline
        $\langle \hdimers, \vdimers |\dots| \hdimers, \vdimers \rangle$ & $0$ & ${4}/{9}$ & $0$ & $0$ & $0$ & $0$ & $0$ & $0$ & $0$\tabularnewline
        $\langle \hdimers, \hdimers |\dots| \hdimers, \hdimers \rangle$ & ${4}/{9}$ & $0$ & $0$ & $0$ & $0$ & $0$ & $0$ & $0$ & $0$\tabularnewline
        \hline
        
        $\langle \vdimers, \vdimers |\dots| \vdimers, \hdimers \rangle$ & $0$ & $0$ & $0$ &${2}/{9}$ & ${2}/{9}$ & $-{2}/{9}$ & $0$ & $0$ & $0$\tabularnewline
        $\langle \vdimers, \vdimers |\dots| \hdimers, \vdimers \rangle$ & $0$ & ${2}/{9}$ & $0$ & $0$ & ${2}/{9}$ & $0$ & $0$ & $-{2}/{9}$ & $0$\tabularnewline
        $\langle \vdimers, \hdimers |\dots| \hdimers, \hdimers \rangle$ & ${2}/{9}$& $0$ & $0$ & ${2}/{9}$ & $0$ & $0$ & $-{2}/{9}$ & $0$ & $0$\tabularnewline
        $\langle \hdimers, \vdimers |\dots| \hdimers, \hdimers \rangle$ & ${2}/{9}$ & ${2}/{9}$ & $-{2}/{9}$ & $0$ & $0$ & $0$ & $0$ & $0$ & $0$\tabularnewline
        \hline
        
        $\langle \vdimers, \vdimers |\dots| \hdimers, \hdimers \rangle$ & ${1}/{9}$ & ${1}/{9}$ & $-{1}/{9}$ & ${1}/{9}$ & ${1}/{9}$ & $-{1}/{9}$ & $-{1}/{9}$ & $-{1}/{9}$ & ${1}/{9}$\tabularnewline
        $\langle \vdimers, \hdimers |\dots| \hdimers, \vdimers \rangle$ & ${1}/{9}$ & ${1}/{9}$ & $-{1}/{9}$ & ${1}/{9}$ & ${1}/{9}$ & $-{1}/{9}$ & $-{1}/{9}$ & $-{1}/{9}$ & ${1}/{9}$\tabularnewline
        \hline
    \end{tabular}%
    }
\end{table*}

Comparing Table~\ref{table:rotated_matrix_elements} with
Table~\ref{table:matrix_elements}, we see that the entries are cycled
in a definite way.  Sites $1,2,3$ are replaced by $2,3,1$, while sites
$5,6,7$ are replaced by $6,7,5$, respectively.  For example,
$\langle \hdimers, \hdimers |
\mathbf{S}_3\cdot\mathbf{S}_5
| \vdimers, \hdimers \rangle$ in the rotated orientation of
Fig.~\ref{fig:tetrahedron_pairs}(b), which is the complex conjugate of
$\langle \vdimers, \hdimers |
\mathbf{S}_3\cdot\mathbf{S}_5
| \hdimers, \hdimers \rangle$, is equal to
$\langle \hdimers, \hdimers |
\mathbf{S}_2\cdot\mathbf{S}_7
| \vdimers, \hdimers \rangle$ in the original orientation of
Fig.~\ref{fig:tetrahedron_pairs}(a).

Thus, changing from orientation~1 to orientation~2 permutes the
two-spin interaction matrices
\eqref{AB_matrix_15}--\eqref{AB_matrix_37}, but only within the
respective $J^{(1)}$ and $J^{(2)}$ sectors of
Eq.~\eqref{eq:pyrochlore_coupling_def}.  After the matrices in each
sector are summed, the effective Hamiltonian in
Eq.~\eqref{eq:eff_Hamiltonian_pyrochlore} is unchanged.  The same
conclusion holds for the other two orientations of a neighbouring
tetrahedron pair.

On the checkerboard lattice, a $C_4$ rotation cycles the sites as
$1\rightarrow 3\rightarrow 1$ and $5\rightarrow 7\rightarrow 5$,
while sites $2$ and $6$ remain fixed.  This again only permutes the
interaction matrices within the coupling sectors
$J^{(1)},J^{(2)},J^{(3)}$, and $J^{(4)}$ defined in
Eq.~\eqref{eq:checkerboard_coupling_def}.  Hence, after all matrices
within each sector are added, the checkerboard effective Hamiltonian
in Eq.~\eqref{eq:eff_Hamiltonian_checkerboard} is also independent of
the orientation of the plaquette pair.

This orientation independence is a useful structural result.  Although
the local pseudospin basis transforms nontrivially under rotations, the
projected Heisenberg perturbation does not generate bond-dependent
compass or spin-orbit couplings at this order.  The pyrochlore
projection therefore gives the simple diamond-lattice XY model in
Eq.~\eqref{eq:eff_Hamiltonian_pyrochlore}, rather than a more general
bond-dependent pseudospin model.

\subsubsection{Emergent pseudospin dynamics on pyrochlore}
\label{sec:XY-model}

Equation~\eqref{eq:eff_Hamiltonian_pyrochlore} is an XY model even though it involves
$x$ and $z$ pseudospin components rather than the more conventional $x$ and $y$.  This is
only a rotation of axes in the emergent two-dimensional singlet
space on each tetrahedron.  Upon defining
\[
    \sigma^x=\tau^x,\qquad
    \sigma^y=\tau^z,\qquad
    \sigma^z= -\tau^y ,
\]
Eq.~\eqref{eq:eff_Hamiltonian_pyrochlore} takes the conventional form
\begin{equation}
    H_{XY}=-J\sum_{\langle ij\rangle}
    \left(\sigma_i^x \sigma_j^x+
    \sigma_i^y \sigma_j^y\right),
\end{equation}
with
\[
    J=-\frac{8}{27}\left(J^{(2)}-J^{(1)}\right).
\]
The global U(1) symmetry is therefore the rotation
symmetry in the $\sigma^x$--$\sigma^y$ plane, generated by $\sigma^z$.  This
is an emergent pseudospin symmetry acting within the local singlet
doublet, not a physical spin-rotation symmetry of the original spin-$1$
moments. For spin-$\tfrac12$, the XY model is equivalent
to hard-core bosons through the Matsubara--Matsuda mapping, with the
ordered XY phase corresponding to phase coherence of the bosons
\cite{Matsubara1956}.  In spatial dimensions greater than one, the
ground state of the quantum XY model exhibits long-range order in the
standard unfrustrated setting \cite{Kennedy1988a}.  At finite
temperature, dimensionality is crucial: in two spatial dimensions one
obtains a Berezinskii--Kosterlitz--Thouless phase with quasi-long-range
order \cite{Ding1990,Ding1992,Sandvik1999,Melko2004}, whereas in three
dimensions true long-range order persists up to a finite-temperature
transition in the 3D XY universality class
\cite{Dyson1978,Kubo1988}.

Because the diamond lattice is bipartite, the two signs of the XY
coupling are equivalent.  A $\pi$ rotation of the pseudospins on one
diamond sublattice about the $\tau^y$ axis sends
\[
    \tau^x\rightarrow -\tau^x,\qquad
    \tau^z\rightarrow -\tau^z,\qquad
    \tau^y\rightarrow \tau^y .
\]
This changes the sign of
$\tau^x_\alpha\tau^x_\beta+\tau^z_\alpha\tau^z_\beta$ on every
nearest-neighbour diamond bond.  Thus, ferromagnetic and
antiferromagnetic XY couplings are related by a sublattice rotation,
and the sign of $J^{(2)}-J^{(1)}$ does not frustrate the leading
pyrochlore pseudospin dynamics.

The effective Hamiltonian in
Eq.~\eqref{eq:eff_Hamiltonian_pyrochlore} should therefore be viewed
as the leading nontrivial collective dynamics within an exactly
characterized singlet manifold of a local SU(2)-invariant pyrochlore
spin model.  Starting from the frustration-free parent Hamiltonian
\eqref{eq:Hamiltonian}, whose zero-energy states form the loop
manifold described in Sec.~\ref{sec:ground_states}, symmetry forbids a
nearest-neighbour local pseudospin bias.  The first
degeneracy-lifting term appears at next-nearest-neighbour order and
takes the form of an unfrustrated nearest-neighbour XY model for
emergent pseudospins on the diamond lattice of tetrahedron centers.
The pseudospins are not introduced phenomenologically: they are the
two local singlet states that span the exact low-energy Hilbert space
on each tetrahedron.  Thus the projection identifies both the
low-energy degrees of freedom and the lattice on which they interact.

The effective Hamiltonian in Eq.~\eqref{eq:eff_Hamiltonian_pyrochlore}
has the usual tendency of an XY model: it favours phase coherence
between the two local singlet states $|A\rangle$ and $|B\rangle$ on
neighbouring tetrahedra.  In the pseudospin language, this corresponds
to ordering in the $\tau^x$--$\tau^z$ plane.  The resulting low-energy
mode is the Goldstone mode associated with slow spatial variations of
this emergent phase.

This ordered pseudospin state should be distinguished from a Coulombic
U(1) spin liquid.  Equation~\eqref{eq:eff_Hamiltonian_pyrochlore}
identifies the leading coherent dynamics within the exact singlet-loop
manifold, but this leading term by itself does not produce a deconfined
gauge field.  If a Coulombic regime is realized near this parent point,
it must come from additional processes acting within the same constrained
manifold.  Natural candidates are higher-order loop reconnections and
ring-exchange terms around closed paths of the diamond lattice, which
play the role of Rokhsar--Kivelson dynamics in related constrained
dimer models.

This perspective clarifies the relation to earlier approaches.  In the
quantum-ice construction of Hermele, Fisher, and Balents
\cite{Hermele2004a}, the emergent gauge structure arises from
perturbation theory about an anisotropic spin-ice manifold.  In
Rokhsar--Kivelson-type and large-$N$ dimer models on the pyrochlore
lattice \cite{Moessner2006a}, the constrained dimer Hilbert space is
imposed at the outset and the effective dimer dynamics is then
constructed within it.  By contrast, in the present construction the
constrained singlet-loop manifold is the exact zero-energy manifold of
a microscopic SU(2)-invariant spin-$1$ Hamiltonian, and the
pseudospin dynamics is obtained by explicitly projecting physical spin
interactions into that manifold.  In this sense the construction is
close in spirit to Klein-point projector Hamiltonians
\cite{Nussinov2006,Nussinov2007}, but with an AKLT-inspired realization
of the local singlet-loop manifold and an explicit route to the
effective diamond-lattice pseudospin theory.

The two-dimensional checkerboard case is more conventional because the
nearest-neighbour perturbation already generates a local pseudospin
field.  The next-nearest-neighbour terms in
Eq.~\eqref{eq:eff_Hamiltonian_checkerboard} therefore usually act only
as subleading anisotropic pseudospin interactions on top of this field.
Consequently, the singlet sector is generically gapped: the system
locally selects either the $|A\rangle$ or the $|B\rangle$ state before
the weaker inter-plaquette pseudospin dynamics can drive a collective
instability.

A more interesting regime can be approached only if the
$\tau^x_\alpha\tau^x_\beta$ interaction becomes large compared with
both the $\tau^z_\alpha\tau^z_\beta$ interaction and the effective
field along $\tau^z$.  In that regime the combination of
Eqs.~\eqref{eq:eff_Hamiltonian_checkerboard} and
\eqref{eq:effective_ham_nn} reduces approximately to a transverse-field
Ising model.  The Ising coupling is controlled by
$3|J^{(2)}-J^{(1)}|$, up to the common prefactor in
Eq.~\eqref{eq:eff_Hamiltonian_checkerboard}, while the transverse field
receives contributions from the nearest-neighbour perturbation and, in
general, from the one-body term generated by next-nearest-neighbour
interactions.

If the Ising coupling overcomes this field, the system can spontaneously
polarize along the $\tau^x$ direction.  A uniform
$\langle\tau^x\rangle$ corresponds, in the original spin-$1$ language,
to a columnar valence-bond state: the AKLT loops predominantly run in
one lattice direction, spontaneously breaking the rotational symmetry
of the checkerboard lattice.  A staggered
$\langle\tau^x\rangle$ instead corresponds to a plaquette
valence-bond solid in which the loops predominantly encircle
alternating empty plaquettes, thereby breaking translation symmetry.
Such regimes require tuning, since the field appears already at the
nearest-neighbour level whereas the Ising interaction arises only at
next-nearest-neighbour order.  Generically one therefore expects the
field to dominate, leading to a uniform local selection of either
$|A\rangle$ or $|B\rangle$ on every crossed plaquette.

\section{Monte Carlo simulation for equal-amplitude loop superposition}
\label{sec:MC}

The fixed-loop calculation of Sec.~\ref{sec:ground_states} and
Appendix~\ref{app:loops_correlations} gives exact control over spin
correlations in an individual loop state.  A generic ground state,
however, is a superposition of many loop configurations, and its norm
and correlation functions contain off-diagonal overlaps between
different loop states.  The question is therefore whether the
short-ranged correlations found in a fixed loop state survive in a
genuine loop superposition.

We address this question for the simplest sign-free case selected by
the checkerboard nearest-neighbour perturbation, namely the product
state with $|A\rangle$ on every crossed plaquette.  This calculation is
not meant to determine a full phase diagram.  Its purpose is to test,
in a controlled equal-amplitude superposition, how the loop correlations
are modified once off-diagonal overlaps are included.

On the checkerboard lattice, the nearest-neighbour perturbation in
Eq.~\eqref{eq:effective_ham_nn} selects the local state
$|A\rangle_\alpha$ on every crossed plaquette when $J_s>J_d$.  The
corresponding product state is
\begin{equation} \label{eq:equal_amplitude_state}
    \lvert \Psi \rangle \rightarrow \prod_{\alpha} \lvert A \rangle_{\alpha} \propto \prod_{\alpha} (\lvert \vdimers \rangle + \lvert \hdimers \rangle)_{\alpha} \, .
\end{equation}
Thus the wavefunction is an equal-amplitude superposition of all
products of horizontal and vertical local valence-bond configurations
on the crossed plaquettes.  This is closely analogous in spirit to the
Rokhsar--Kivelson equal-amplitude superposition of fully packed dimer
coverings~\cite{Rokhsar1988}, but with an important difference: the
objects being superposed here are spin-$1$ AKLT-like valence-bond loop
states, not hard-core dimer configurations introduced as an effective
Hilbert space.

The sign structure of Eq.~\eqref{eq:equal_amplitude_state} is crucial.
All horizontal and vertical dimer products enter with the same sign.
Consequently, the norm $\langle \Psi|\Psi\rangle$ can be written as a
positive classical statistical sum over pairs of loop configurations.
This positivity is what makes a Metropolis Monte Carlo calculation
possible.  By contrast, a generic product of local $|B\rangle$ states
would contain relative minus signs between horizontal and vertical
dimerizations, and the corresponding overlap expansion would not have
the same sign-free interpretation.

For an observable $\hat O$, we use the standard valence-bond Monte
Carlo estimator~\cite{Tang2011a}
\begin{equation}\label{eq:observable_expectation}
\langle \hat{O} \rangle =
\frac{
\sum_{\mathcal{L}, \mathcal{L}'}
p_{\mathcal{L}, \mathcal{L}'}
\langle \mathcal{L}' | \mathcal{L} \rangle
\frac{\langle \mathcal{L}' | \hat{O} | \mathcal{L} \rangle}
{\langle \mathcal{L}' | \mathcal{L} \rangle}
}{
\sum_{\mathcal{L}, \mathcal{L}'}
p_{\mathcal{L}, \mathcal{L}'}
\langle \mathcal{L}' | \mathcal{L} \rangle
} \, .
\end{equation}
Here $\mathcal L$ and $\mathcal L'$ denote the ket and bra loop
configurations obtained by expanding Eq.~\eqref{eq:equal_amplitude_state},
and $p_{\mathcal{L},\mathcal{L}'}$ is the product of their wavefunction
amplitudes.  For the equal-amplitude state in
Eq.~\eqref{eq:equal_amplitude_state}, this amplitude product is the
same for all pairs and may be set to
$p_{\mathcal{L},\mathcal{L}'}=1$.  The nontrivial statistical weight
then comes entirely from the overlap
$\langle \mathcal L'|\mathcal L\rangle$.

We are interested in the two-point spin correlation function
$\langle \mathbf S_i\cdot\mathbf S_j\rangle$.  Since the spin-coherent
basis gives a local representation of both
$\langle \mathcal L'|\mathcal L\rangle$ and
$\langle \mathcal L'|\mathbf S_i\cdot\mathbf S_j|\mathcal L\rangle$,
and since the overlap is non-negative for the equal-amplitude state,
we can sample pairs $(\mathcal L,\mathcal L')$ by a Metropolis
algorithm.  For $\hat O=\mathbf S_i\cdot\mathbf S_j$, Eq.~\eqref{eq:observable_expectation}
becomes
\begin{equation} \label{eq:observable_expectation_2}
    \langle \mathbf{S}_i \cdot \mathbf{S}_j \rangle =
    \frac{
    \sum_{\mathcal{L}, \mathcal{L}'}
    W_{\mathcal{L}, \mathcal{L}'}
    C_{ij}(\mathcal{L},\mathcal{L}')
    }{
    \sum_{\mathcal{L}, \mathcal{L}'}
    W_{\mathcal{L}, \mathcal{L}'}
    } \, ,
\end{equation}
where
\[
    W_{\mathcal{L},\mathcal{L}'}
    =
    \langle \mathcal L'|\mathcal L\rangle
\]
is the positive Monte Carlo weight, and
\[
    C_{ij}(\mathcal{L},\mathcal{L}')
    =
    \frac{
    \langle \mathcal L'|
    \mathbf S_i\cdot\mathbf S_j
    |\mathcal L\rangle
    }{
    \langle \mathcal L'|\mathcal L\rangle
    }
\]
is the corresponding estimator.  For a proposed update from an old
configuration pair to a new configuration pair, the Metropolis
acceptance probability is
\[
    P_{\mathrm{accept}}
    =
    \min\left[
    \frac{W_{\mathrm{new}}}{W_{\mathrm{old}}},
    1
    \right].
\]

The structure of the sampled configurations differs from the usual
classical spin Monte Carlo problem.  We do not assign a classical vector
to each site.  Instead, each physical spin-$1$ site is touched by two
virtual singlets in the ket configuration and two virtual singlets in
the bra configuration.  Thus four singlet arrows meet at every site in
the overlap graph.  This is similar in spirit to amplitude-product RVB
Monte Carlo sampling~\cite{Tang2011a}, but here each layer contains two
singlet bonds incident on every physical spin-$1$ site, rather than one.

\begin{figure}[t!]
  \centering
  \includegraphics[width=\linewidth, trim=30 30 30 30, clip]{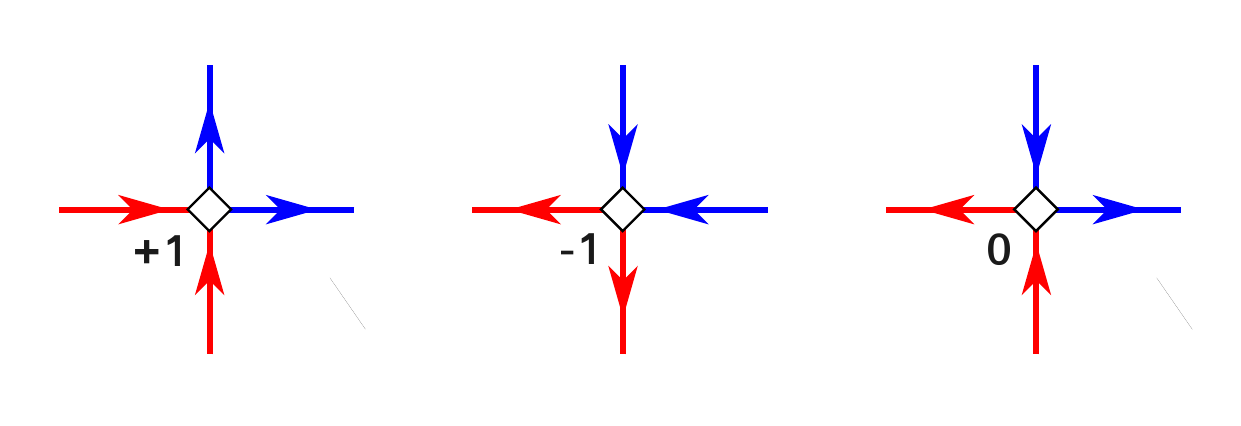}
 \caption{Representative local arrow configurations at a physical
spin-$1$ site satisfying the $2$-in-$2$-out rule.  Red lines denote
singlets from the ket configuration, while blue lines denote singlets
from the bra configuration.  Each physical site is incident on two ket
singlets and two bra singlets.  The numbers shown near the sites give
the relative values of the local $\Omega^z$ spin-insertion integrals
appearing in Eq.~\eqref{eq:spin_weights}; the common factor $2\pi/3$
has been suppressed.  These local weights enter the Monte Carlo
estimator for $\mathbf S_i\cdot\mathbf S_j$.}
\label{fig:singlets_site}
\end{figure}

The local coherent-state integral imposes a strong constraint.  Both
the overlap $\langle \mathcal L'|\mathcal L\rangle$ and the spin
matrix element
$\langle \mathcal L'|\mathbf S_i\cdot\mathbf S_j|\mathcal L\rangle$
vanish unless, at every site, the number of incoming arrows equals the
number of outgoing arrows.  We refer to this as the $2$-in-$2$-out
rule; examples are shown in Fig.~\ref{fig:singlets_site}.  This rule is
the local remnant of the $\phi$ integration in the spin-coherent
representation: any imbalance between incoming and outgoing arrows
leaves an uncancelled phase factor and hence integrates to zero.

Using the vertex weights of Eq.~\eqref{eq:vertex_weights} and the
spin-insertion weights of Eq.~\eqref{eq:spin_weights}, we can evaluate
both the Monte Carlo weight and the spin estimator for any valid pair
$(\mathcal L,\mathcal L')$.  A site that satisfies the $2$-in-$2$-out
rule contributes either the base weight $2\pi/3$ or twice this value,
$4\pi/3$.  Let $N$ be the total number of physical sites, and let
$N_\star(\mathcal L,\mathcal L')$ be the number of sites of the latter
type, namely sites where both ket arrows come in and both bra arrows go
out, or the reversed arrangement where both ket arrows go out and both
bra arrows come in.  Such a vertex has two possible local loop
resolutions, which accounts for the extra factor of $2$ relative to the
base vertex weight.  Thus
\begin{subequations} \label{eq:mc_probability_weight}
    \begin{equation}
        \label{eq:mc_probablitiy_density}
        W_{\mathcal{L}, \mathcal{L'}} =
        2^{N_\star(\mathcal L,\mathcal L')}
        \bigg( \frac{2 \pi}{3} \bigg)^{N} \, .
    \end{equation} 
    The common factor $(2\pi/3)^N$ cancels from all Metropolis ratios,
    but it is displayed here to show the origin of the statistical
    weight.

    For the two spin-insertion sites $i$ and $j$, the local
    $\Omega^z$ insertion is nonzero only for the doubled vertices just
    described.  At such a site, insertion of $\Omega^z$ replaces the
    overlap weight $4\pi/3$ by $\pm 2\pi/3$, giving a local ratio
    $\pm 1/2$.  Hence the reduced two-site insertion is
    \begin{equation} \label{eq:mc_weight}
        \chi_{ij}(\mathcal L,\mathcal L') = 
        \begin{cases}
            +\dfrac{2^{N_\star - 2} \big(\frac{2 \pi}{3} \big)^N}
            {2^{N_\star} \big( \frac{2 \pi}{3} \big)^{N}}
            = +\dfrac14,\\[1.2em]
            -\dfrac{2^{N_\star - 2} \big(\frac{2 \pi}{3} \big)^N}
            {2^{N_\star} \big( \frac{2 \pi}{3} \big)^{N}}
            = -\dfrac14,\\[1.2em]
            0.
        \end{cases}
    \end{equation}
\end{subequations}
The positive sign occurs when the two spin-insertion sites have the
same local ket-arrow orientation: both ket arrows enter at both sites,
or both ket arrows leave at both sites.  The negative sign occurs when
the ket arrows enter at one site and leave at the other.  If either
site is of a vertex type for which the $\Omega^z$ insertion vanishes,
the estimator is zero.  With the spin-$1$ coherent-state prefactor used
in the correlation function, the physical estimator entering
Eq.~\eqref{eq:observable_expectation_2} is
\begin{equation}
    C_{ij}(\mathcal L,\mathcal L')
    =
    4\,\chi_{ij}(\mathcal L,\mathcal L')
    \in \{+1,-1,0\}.
\end{equation}
There is also an extra factor of $3$ due to
\(
    \langle \mathbf{S}_i\cdot\mathbf{S}_j \rangle
    =
    3\langle {S}_i^z {S}_j^z\rangle .
\)
However, we scale down this factor of 3 since the form of Eq.~\eqref{eq:observable_expectation_2} makes the Monte Carlo problem simple and transparent: one samples a
positive weight $W_{\mathcal L,\mathcal L'}$, while the measured
quantity is a bounded random variable taking the values $+1$, $-1$, or
$0$.

The factor $2^{N_\star}$ in Eq.~\eqref{eq:mc_probablitiy_density} also
has a simple physical interpretation.  Configuration pairs with larger
$N_\star$ have more locally double-weight vertices and therefore larger
overlap weight.  Such configurations tend to decompose into many short
transition loops, whereas long ``no-turn-back'' loops reduce the number
of locally doubled vertices.  Since the spin estimator is nonzero only
when the two spin-insertion sites lie on the same transition loop, one
expects the spin correlation function to decay with distance.

\subsection{Sampling a valid loop configuration}

We begin by constructing an initial valid pair of configurations
$(|\mathcal L\rangle,\langle\mathcal L'|)$ on an $L\times L$
checkerboard lattice with periodic boundary conditions.  The ket and
bra singlets are then superposed to form an overlap graph, as
illustrated in Fig.~\ref{fig:initial_configuration_pair}.  A valid
configuration pair must obey two local constraints at every site:
there must be exactly two ket singlets and two bra singlets incident on
the site, and the resulting four arrows must satisfy the $2$-in-$2$-out
rule.  Periodic boundary conditions impose the same constraints across
the boundaries of the finite lattice.

\begin{figure}[t!]
  \centering
  \includegraphics[width=\linewidth]{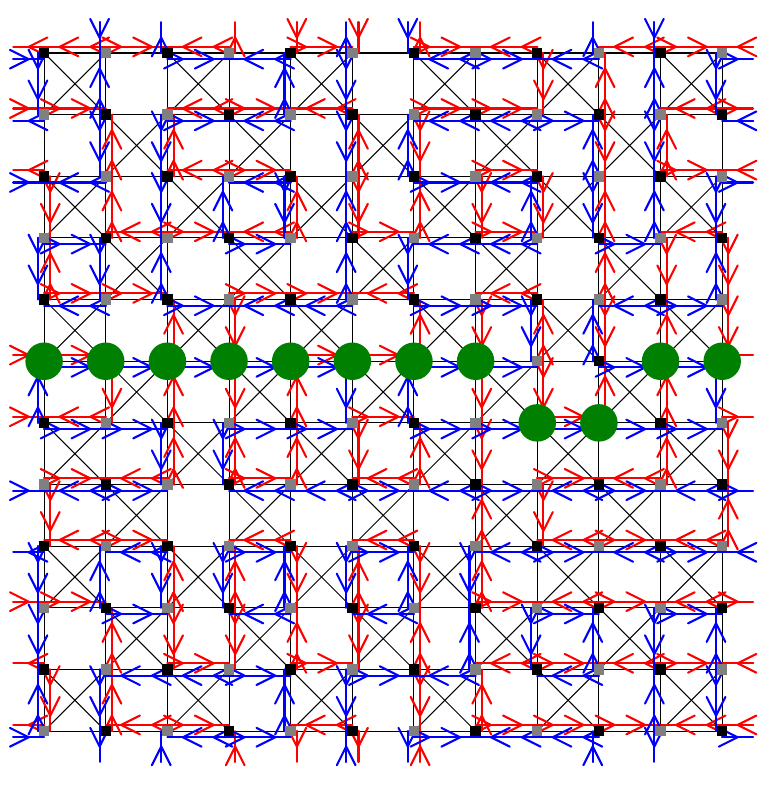}
 \caption{Valid initial configuration pair
$|\mathcal L\rangle,\langle\mathcal L'|$ on a $12\times12$
checkerboard lattice with periodic boundary conditions.  Red singlets
belong to the ket configuration $|\mathcal L\rangle$, and blue singlets
belong to the bra configuration $\langle\mathcal L'|$.  At every site,
the four incident arrows satisfy the $2$-in-$2$-out rule.  Sites marked
by green circles are no-turn-back sites.  Because the left and right
boundaries are identified, the horizontal string of such sites forms a
winding no-turn-back loop on the torus.}
\label{fig:initial_configuration_pair}
\end{figure}

The transition loops are traced by following the arrows and alternating
between ket and bra singlets at every site.  Since each step switches
between the two layers, a closed transition loop contains the same
number of ket and bra singlet segments.  Overlaying
$|\mathcal L\rangle$ and $\langle\mathcal L'|$ therefore produces a
collection of closed loops of different lengths, as seen in
Fig.~\ref{fig:initial_configuration_pair}.

Some local arrow configurations force the loop to pass from one crossed
plaquette to a neighbouring crossed plaquette.  We call such vertices
no-turn-back sites.  At these sites, once the incoming segment of a
transition loop is specified, the outgoing segment cannot return to the
same crossed plaquette.  Other allowed arrow configurations give a
choice: the transition loop may either turn back within the same crossed
plaquette or pass onward to a neighbouring one.  These possibilities are
illustrated in Fig.~\ref{fig:arrow_combination_sets}.

\begin{figure}[t!]
  \centering
  \includegraphics[width=\linewidth]{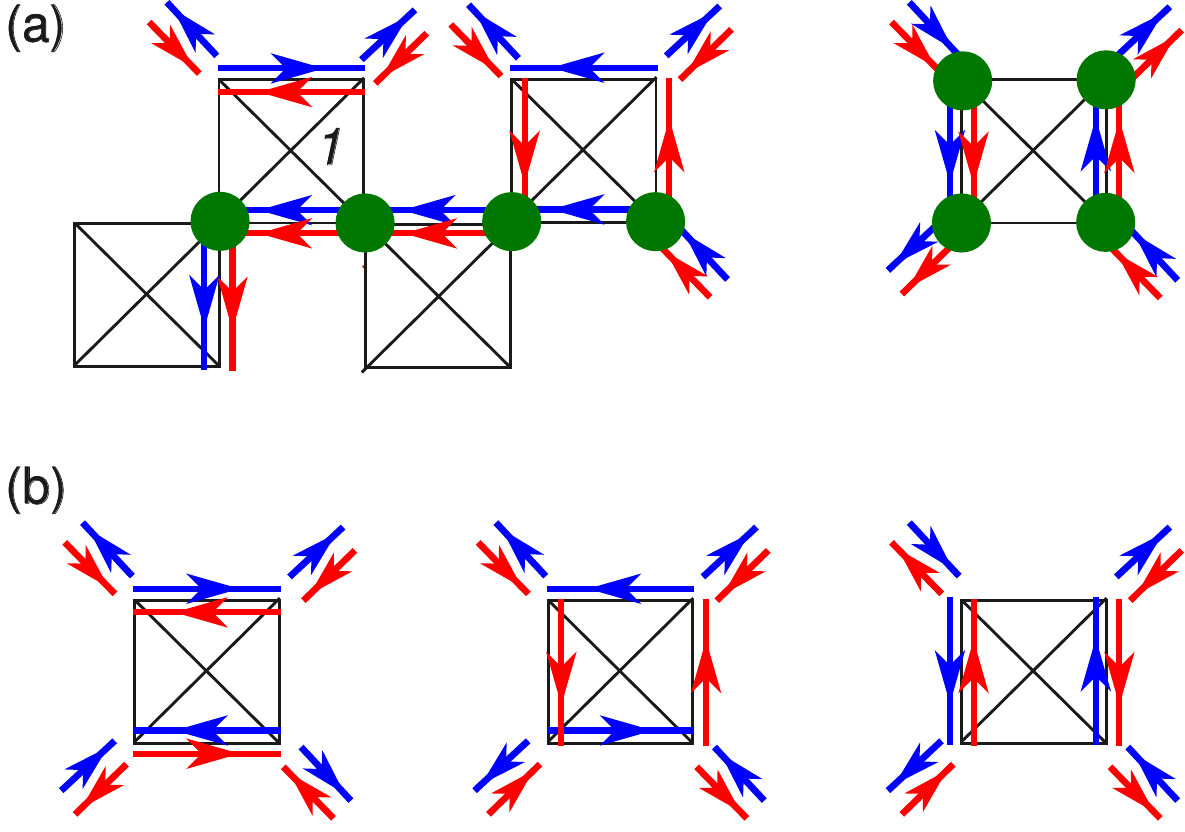}
 \caption{Local arrow configurations on crossed plaquettes.  (a)
Examples containing no-turn-back sites, marked by green circles.  At
such sites, a transition loop is forced to pass from one crossed
plaquette to a neighbouring one.  In the crossed plaquette labelled
$1$, for example, the two bottom no-turn-back sites force the loop to
connect the lower-right neighbouring plaquette to the lower-left
neighbouring plaquette.  The two top sites do not impose such a
connection and are therefore not marked.  (b) Examples of allowed arrow
configurations in which the transition loop can turn back and remain
within the same crossed plaquette.}
\label{fig:arrow_combination_sets}
\end{figure}

In many valid initial configurations, no-turn-back sites form continuous
horizontal or vertical strings.  With periodic boundary conditions such
a string becomes a winding no-turn-back loop.  Larger lattices may
contain more than one horizontal or vertical winding no-turn-back loop.
These winding structures are important because the local updates
described below do not change their number or orientation.  The Monte
Carlo simulation should therefore be understood as sampling within a
fixed no-turn-back sector.

\subsection{Box update proposal}

The Monte Carlo update must preserve all local constraints: two ket
singlets and two bra singlets must remain incident on every site, and
the $2$-in-$2$-out rule must continue to hold.  We use a local box
update designed to satisfy these conditions [see Fig.~\ref{fig:box_update}].

\begin{figure}[t!]
  \centering
  \includegraphics[width=\linewidth]{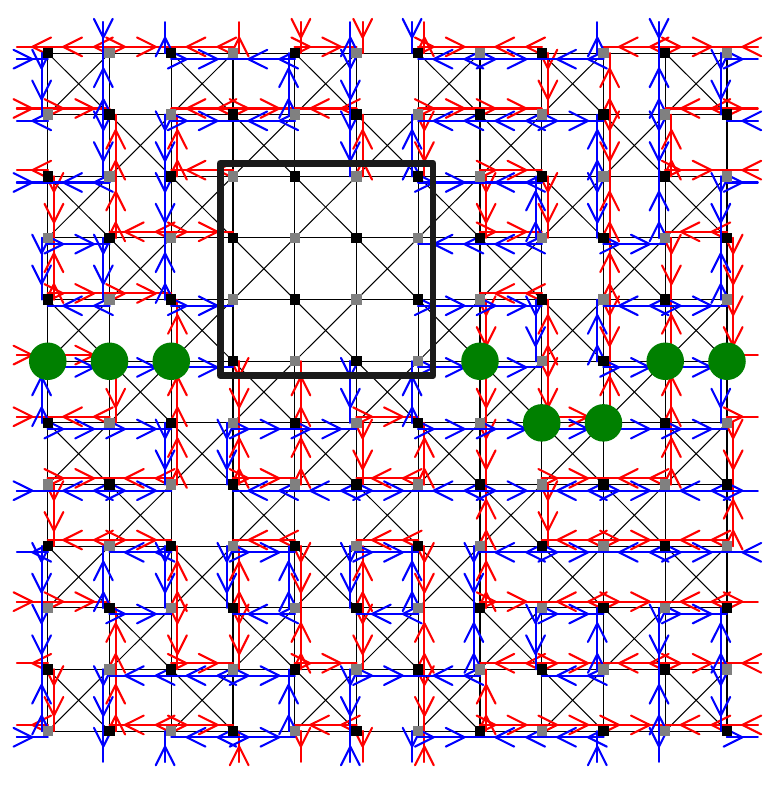}
 \caption{Local box update used in the Monte Carlo sampling.  A
rectangular $4\times4$ box is selected inside the configuration pair of
Fig.~\ref{fig:initial_configuration_pair}, and a new allowed singlet
configuration is proposed within the box.  The update preserves the
local site constraints and keeps all singlets crossing the boundary of
the box, including their arrow directions, fixed.  Since only $16$
sites are updated out of the full $12\times12$ lattice, the move is
local.}
\label{fig:box_update}
\end{figure}

In each Monte Carlo sweep, we choose a rectangular box of size
$b_1\times b_2$ at a random location.  The proposed update changes the
allowed singlet pattern inside this box while leaving the external
configuration fixed.  In particular, the singlets crossing the boundary
of the box and their arrow directions are not changed.  This ensures
that sites outside the box continue to satisfy the local constraints.

For simplicity, the box shape and size are kept fixed throughout a
given simulation.  The update is local when
$b_1 b_2 \ll L^2$.  Because it is local, it cannot create, destroy, or
change the winding direction of a no-turn-back loop that wraps around
the torus.  The number and orientation of winding no-turn-back loops
are therefore sector labels for the simulation.  This is why the first
two rows of Fig.~\ref{fig:decay_16-16} show the same type of
no-turn-back structure before and after Monte Carlo sampling.

\subsection{Monte Carlo runs}

The numerical measurement of long-distance correlations is demanding
because the signal decays exponentially, whereas the individual
estimator remains of order unity.  For any sampled configuration pair
$(\mathcal L,\mathcal L')$, the estimator
$C_{ij}(\mathcal L,\mathcal L')$ takes only the values $+1$, $-1$, or
$0$.  At large separations, the mean value is therefore obtained from a
small imbalance between positive and negative contributions, while the
fluctuations of individual measurements remain much larger than the
signal.  As a result, the tail of the correlation function converges
slowly and requires a large number of Monte Carlo measurements.

This limitation is particularly visible for the largest distances on
the $16\times16$ lattice.  For example, the diagonal correlation shown
in Fig.~\ref{fig:decay_16-16}(i) has a semi-log slope of approximately
$1.42$ per lattice step.  At the largest separation considered,
$|i-j|=9$, the corresponding average is exponentially small, while the
measured estimator still fluctuates on the scale set by the values
$+1$, $-1$, and $0$.  This separation between the size of the signal and
the size of the sample-to-sample fluctuations is the reason that
billions of Monte Carlo measurements are needed to obtain stable
long-distance data.  Larger lattices would require still longer runs to
reach comparable accuracy.  We therefore show results up to
$16\times16$.  Simulations on $12\times12$ and $14\times14$ lattices
show the same qualitative exponential decay; these data are not shown.

Statistical uncertainties were estimated by binning the Monte Carlo
time series.  The fitted slopes quoted in
Fig.~\ref{fig:decay_16-16} were obtained from the range over which the
semi-log plots are linear within the statistical resolution.

Figure~\ref{fig:decay_16-16} summarizes the results for the
$16\times16$ lattice in three no-turn-back sectors: a sector with one
horizontal winding no-turn-back loop, a sector with no winding
no-turn-back loop, and a sector with one vertical winding no-turn-back
loop.  The bottom row shows semi-log plots of the spin-spin correlation
function along the horizontal, vertical, and diagonal directions.  In
all three sectors the correlations decay approximately exponentially.
The fitted decay rate, however, depends on the sector.  When a winding
no-turn-back loop is present along a given direction, the decay in that
direction is slower, corresponding to a longer correlation length.

We do not attempt here to classify all no-turn-back sectors, nor do we
study sectors containing multiple winding no-turn-back loops.  The
numerical results should therefore be interpreted as evidence that the
equal-amplitude $|A\rangle$ product state has short-ranged spin
correlations in the sectors studied.  A systematic sector-by-sector
analysis, including nonlocal updates that change winding sectors, is
left for future work.

\begin{figure*}
    \includegraphics[width=\linewidth]{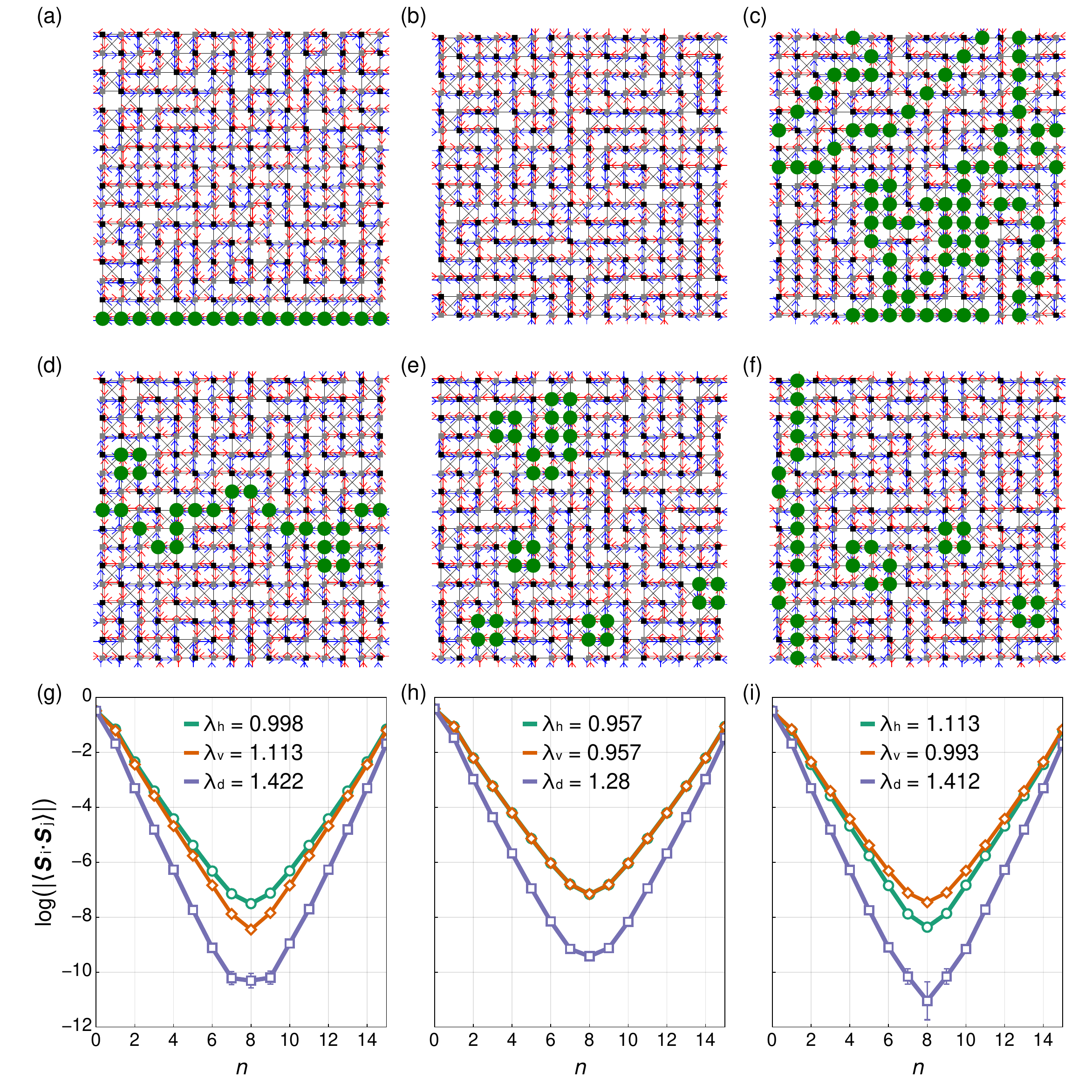}
    \caption{Monte Carlo results for the equal-amplitude $|A\rangle$
product state on a $16\times16$ checkerboard lattice.  The top row
shows representative initial configuration pairs, the middle row shows
representative final configuration pairs after Monte Carlo sampling,
and the bottom row shows semi-log plots of the spin-spin correlation
function.  The three columns correspond, respectively, to a sector with
one horizontal winding no-turn-back loop, a sector with no winding
no-turn-back loop, and a sector with one vertical winding no-turn-back
loop.  Because the box update is local, the number and orientation of
winding no-turn-back loops are unchanged during the simulation.  The
variable $n$ denotes the separation $|i-j|$ between the two sites,
measured as the shortest number of lattice steps along the chosen
direction.  The fitted slopes $\lambda_h$, $\lambda_v$, and
$\lambda_d$ correspond to horizontal, vertical, and diagonal
directions, respectively.  A winding no-turn-back loop along a given
direction reduces the corresponding decay slope compared to the transverse direction slope, indicating a longer
correlation length in that direction.}
\label{fig:decay_16-16}
\end{figure*}

\section{Discussion and Outlook}
\label{sec:discussion_outlook}

\subsection{Discussion}

This work establishes an exact microscopic parent point for constrained
valence-bond physics on corner-sharing lattices, including the
three-dimensional pyrochlore lattice.  The Hamiltonian is local,
SU(2)- and time-reversal-invariant, and frustration-free: it is built
from projectors on triangular faces that remove only the maximally
polarized spin sector.  Its ground-state manifold can therefore be
characterized directly from local zero-energy conditions, rather than
introduced as an effective dimer or loop Hilbert space.

The construction combines two usually distinct routes to exact quantum
magnetism.  As in AKLT-type parent Hamiltonians, the ground states have
an explicit valence-bond structure~\cite{Affleck1987}.  As in
Klein-type projector Hamiltonians, the ground-state space is
extensively degenerate and organized by local constraints
~\cite{Nussinov2006,Nussinov2007}.  The important difference is that
the present construction works already for spin one on the checkerboard
and pyrochlore lattices, avoiding the larger on-site spins usually
required in three-dimensional AKLT constructions on highly coordinated
lattices.

A central conceptual advance is that the constrained degenerate manifold is not imposed at an effective level, as in many Rokhsar--Kivelson-type constructions~\cite{Rokhsar1988,Moessner2006a}, but instead emerges directly from a microscopic spin-1 Hamiltonian whose ground states are spanned, in an overcomplete manner, by fully packed singlet-loop configurations generated through an AKLT-inspired fractionalization-and-recombination procedure. Local kernel constraints on each triangular face enforce the presence of at least one singlet edge, while global lattice consistency and periodic boundary conditions eliminate states that use both virtual spin-$\frac{1}{2}$ variables of a physical site inside the same tetrahedron or crossed plaquette, leaving loop states as the admissible manifold in the thermodynamic limit.

A second essential aspect of this work is analytical control over correlation diagnostics within this manifold. For a fixed loop covering, the coherent-state and Schwinger-boson representation leads to a strikingly geometric result: two spins exhibit nonzero correlations if and only if they belong to the same loop, with correlations decaying exponentially with separation along that loop. This establishes a clear separation between local entanglement encoded in loop connectivity and any long-distance correlations, which must arise from quantum superpositions of loop coverings. The manuscript carefully addresses the role of off-diagonal overlaps and their associated phase structure, an essential ingredient for assessing when a loop-gas description behaves classically and when genuinely quantum interference effects dominate~\cite{Wildeboer2020}.

A further key devolpment presented here is the explicit construction of emergent low-energy dynamics by projecting symmetry-allowed perturbations into the exactly known ground-state subspace. In two dimensions, this procedure yields an effective XY spin-$\tfrac12$ model on the square lattice with a tunable bias reflecting microscopic anisotropies, while on the pyrochlore lattice tetrahedral symmetry enforces the absence of such a bias, resulting in intrinsic collective dynamics on the diamond lattice of tetrahedra. This distinction is substantive: on the pyrochlore lattice, degeneracy lifting is governed by genuinely collective processes rather than simple single-site polarization effects, closely paralleling the structure of three-dimensional Coulomb phases~\cite{Hermele2004a}.

Finally, the Monte Carlo analysis highlights a nontrivial and physically meaningful issue: the presence of topological or winding-like sector invariants that obstruct ergodic sampling under local updates. The appearance of ``no-turn-back'' loop structures that remain invariant under local moves leads to direction-dependent correlation decay and sector-dependent observables. This is not merely an algorithmic complication but reflects the intrinsic structure of the constrained Hilbert space, underscoring the need for sector-aware diagnostics and nonlocal sampling strategies, as encountered more broadly in constrained loop and dimer models~\cite{Moessner2006a}.

\subsection{Outlook}

The construction developed here gives a controlled starting point, but
it does not by itself settle which phase is selected once the exact
degeneracy is lifted.  The most direct next step is to carry the
projection of microscopic perturbations further.  On the pyrochlore
lattice, the absence of a local pseudospin field means that the
important terms are expected to involve several tetrahedra.  It would
therefore be useful to derive explicitly the loop-reconnection and
ring-exchange processes generated at higher orders in perturbation
theory, together with their amplitudes and signs.  Such a calculation
would give an effective Hamiltonian written directly in the constrained
loop manifold, and would clarify whether the resulting dynamics is
closer to a Coulombic $U(1)$ regime, as in quantum spin ice and
diamond-lattice quantum dimer models~\cite{Hermele2004a,Sikora2011},
or to a confined valence-bond-ordered phase~\cite{Moessner2006a}.

Another open issue concerns correlations in general superpositions of
loop states.  For a fixed loop covering, the result is simple: spin
correlations are determined by loop connectivity and decay
exponentially along a loop.  In a superposition, however, off-diagonal
overlaps and relative signs can matter.  It would be valuable to
separate these two effects more systematically: the purely geometric
probability that two sites lie on the same loop, and the quantum
interference weights with which different loop coverings enter the
state.  The same distinction should also be useful for nonlocal
observables, such as string operators or flux-like quantities defined
on the tetrahedra, which are natural diagnostics for distinguishing
confined and deconfined regimes~\cite{Hermele2004a}.

The Monte Carlo results in Sec.~\ref{sec:MC} also point to a practical
issue that deserves further study.  Local updates preserve winding
structures and therefore sample only a fixed sector of the constrained
configuration space.  This is useful for diagnosing sector dependence,
but it is not sufficient for a complete numerical treatment of the loop
manifold.  Nonlocal updates, such as worm moves, directed-loop updates,
or loop-surgery moves, should make it possible to change winding
sectors while preserving the local constraints~\cite{Sandvik2005}.
Sector-resolved measurements would then help distinguish physical
anisotropies from features introduced by restricted sampling.  In
parallel, the effective pseudospin Hamiltonians derived here provide a
more compact route to numerical simulations of the degeneracy-lifting
dynamics.

Another important direction is the study of excitations.  Since the
parent Hamiltonian is a sum of local projectors, the elementary
violations of the zero-energy conditions are well defined.  Their
motion and interaction under projected perturbations should reveal
whether defects remain confined to local valence-bond rearrangements or
can behave as mobile fractionalized excitations in an appropriate
regime.  This question is especially interesting on the pyrochlore
lattice, where three-dimensional constrained dynamics can support
Coulombic correlations and emergent gauge structure in related models.
Understanding these defects would also be a first step toward computing
spectral quantities, such as dynamical spin structure factors.

A related question concerns open boundaries, especially for the pyrochlore lattice. The analogy with the AKLT chain suggests why this is interesting: cutting a spin-$1$ AKLT chain leaves behind unpaired virtual spin-$\tfrac12$ degrees of freedom at its ends. In the present construction, a surface can similarly leave dangling virtual spin-$\tfrac12$ variables, but their number and connectivity depend on the surface termination.

For example, consider a [111] surface of the pyrochlore lattice, as suggested by the orientation shown in Fig.~\ref{fig:lattices}. If the surface is terminated so that the outermost sites are the apical sites of tetrahedra above a kagome plane, the dangling virtual spin-$\tfrac12$ degrees of freedom reside at these apical sites and form a triangular lattice. From the perspective of the pyrochlore lattice, these tetrahedra are next-nearest neigbours; consequently, the boundary degrees of freedom at their apical sites are not coupled by either the bulk projector terms or the additional Heisenberg interactions considered in this work. Their low-energy physics is therefore determined by residual surface interactions. Depending on these interactions, the surface may realize different effective spin-$\tfrac12$ models, ranging from magnetically ordered states to more strongly fluctuating regimes.

Other terminations can lead to different boundary degrees of freedom. For instance, a termination through kagom\'e planes may change not only the connectivity but also the location of the residual spin-$\tfrac12$ variables. The surface problem is therefore not universal in the same sense as the bulk parent construction, but it provides a natural setting in which to ask how local projector constraints, residual interactions and details of termination organize effective low-energy degrees of freedom at a boundary.

Finally, the model suggests a possible route for quantum simulation. The microscopic Hamiltonian contains projector interactions rather than only simple two-spin exchanges, so it is unlikely to arise directly as a minimal material Hamiltonian.  Nevertheless, constrained dynamics and effective multi-spin terms can be engineered in several quantum simulation platforms, including cold atoms, Rydberg arrays, and
superconducting circuits~\cite{Bloch2012}. The present construction may therefore serve as a useful target Hamiltonian for exploring how valence-bond constraints, loop dynamics, and emergent gauge-like descriptions can arise from local quantum degrees of freedom.

\section*{Acknowledgments}
We thank Roderich Moessner, Anders Sandvik, Ludovic Jaubert, Kabir Ramola and Sayan Sircar for discussions. 
The work of Y.I. was performed, in part, at the Aspen Center for Physics, which is supported by a grant from the Simons Foundation (1161654, Troyer). This research was also supported in part by grant NSF PHY-2309135 to the Kavli Institute for Theoretical Physics. Y.I. acknowledges support from the Abdus Salam International Centre for Theoretical Physics through the Associates Programme, from the Simons Foundation through Grant No.~284558FY19, from IIT Madras through the Institute of Eminence program for establishing QuCenDiEM (Project No. SP22231244CPETWOQCDHOC), and the International Centre for Theoretical Sciences for participation in the Discussion Meeting --- Fractionalized Quantum Matter (code: ICTS/DMFQM2025/07). K.S. acknowledges support from IIT Madras as a Visiting Faculty Fellow during which this project was initiated.

\section*{Data Availability Statement} The data and analysis scripts used to generate the Monte Carlo results shown in Fig.~\ref{fig:decay_16-16} are available at Zenodo \href{https://doi.org/10.5281/zenodo.21186494}{https://doi.org/10.5281/zenodo.21186494}.
Additional data and details are available from the corresponding author upon reasonable request.

\appendix

\section{Loop representation for the overlaps and matrix elements between singlet loop states}
\label{app:stat_mech_loops}

In this appendix we develop the loop representation used in the main
text and in Appendices~\ref{app:loops_correlations} and
\ref{app:matrix_elements_calculation}.  The purpose is to rewrite
overlaps between singlet-loop states, and later matrix elements with
spin insertions, as sums over classical loop configurations.  This
representation is useful for two reasons.  First, it gives a
statistical-mechanical form for the norm of a wavefunction written as
a superposition of singlet-loop states.  When the amplitudes in this
superposition are non-negative, this norm may be interpreted as a
partition function.  Second, the same representation gives a simple
book-keeping scheme for spin--spin correlations and for the dimer
matrix elements used in the effective Hamiltonian.

We shall use a complete local basis of valence-bond states on each
crossed plaquette or tetrahedron.  It is therefore not necessary to
treat all three opposite-pair dimerizations as independent basis
states.  On the checkerboard lattice we choose the two basis states
which we call ``horizontal'' and ``vertical'' dimerizations; the
``diagonal'' dimerization is then a linear combination of these two.
On the pyrochlore lattice the labels horizontal, vertical, and
diagonal have no intrinsic geometric meaning.  One may choose any two
of the three local dimerizations as a basis on each tetrahedron.  The
only essential point is that the chosen singlet-loop states span the
ground-state manifold.

The starting point is the coherent-state overlap of a singlet-loop
state, already used in Sec.~\ref{sec:gs_correlations}.  For a loop
configuration $\mathcal L$,
\begin{equation}\label{eq:overlap1_repeat}
  \Psi^{\mathcal{L}}(\Omega)\equiv \langle \mathcal{L}|\Omega\rangle \propto \prod_{\langle i,j\rangle \in \mathcal{L}}\! \left(u_i  v_j - v_i  u_j\right) \, ,
\end{equation}
where
\[
    u_i =
    \cos\left(\frac{\theta_i}{2}\right)e^{i\phi_i/2},
    \qquad
    v_i =
    \sin\left(\frac{\theta_i}{2}\right)e^{-i\phi_i/2},
\]
and the spin coherent state at site $i$ is parametrized by
\[
    \boldsymbol{\Omega}_i
    =
    (\sin\theta_i\cos\phi_i,
     \sin\theta_i\sin\phi_i,
     \cos\theta_i).
\]
The proportionality sign in Eq.~\eqref{eq:overlap1_repeat} indicates
that overall normalization factors, common to all loop resolutions in
what follows, have been suppressed.

Consider now a general wavefunction written as a superposition of loop
states.  Its norm contains interference terms between two loop
configurations $\mathcal L'$ and $\mathcal L$ of the form
\begin{equation*}
\prod_{\langle i,j\rangle \in \mathcal{L'}}\! \left(u_i^*  v_j^* - v_i^*  u_j^*\right)\prod_{\langle m,n\rangle \in \mathcal{L}}\! \left(u_m  v_n - v_m  u_n\right) \, .
\end{equation*}
The fully packed nature of the loop states is important.  Every site
$i$ is touched by two singlet bonds in $\mathcal L'$ and by two
singlet bonds in $\mathcal L$.  Let $j$ and $k$ be the two neighbours
of $i$ in $\mathcal L'$, and let $m$ and $n$ be the two neighbours of
$i$ in $\mathcal L$.  The two sets $\{j,k\}$ and $\{m,n\}$ may overlap,
depending on the two loop configurations.  The full overlap integrand
then contains, at site $i$, the local factor
\begin{equation*}
\left(u_i^*  v_j^* - v_i^*  u_j^*\right)\left(u_i^*  v_k^* - v_i^*  u_k^*\right)\left(u_i  v_m - v_i  u_m\right)\left(u_i  v_n - v_i  u_n\right),
\end{equation*}
and no other factor in the integrand contains $u_i$ or $v_i$.  Expanding
this expression gives $2^4=16$ local terms.  The purpose of the arrow
representation below is to keep track of which of these terms survive
the angular integration and what weights they carry.

We first define an arrow convention for the ket factors.  A blue arrow
from site $p$ to site $q$ represents the factor $u_p v_q$:
\begin{center}
\begin{pspicture}(0,0.5)(5,1.5)
\psline[linewidth=0.06cm,linearc=1.2,linecolor=blue](1,1)(2,0)(3,1)
\psarc[linewidth=0.12cm,linecolor=blue]{->}(2,1.2){0.705}{-89}{-85}
\psdots[dotscale=1.5](1,1)(3,1)
\rput(0.6,1){$\mathlarger{\mathlarger{p}}$}
\rput(3.4,1){$\mathlarger{\mathlarger{q}}$}
\rput(2,1){$\mathlarger{\mathlarger{u_p v_q}}$}
\end{pspicture}
\end{center}
The direction of the arrow is the direction of the phase flow:
$u_pv_q$ carries the phase factor
$e^{i(\phi_p-\phi_q)/2}$, so the arrow points from the site carrying
the positive $\phi$ phase to the site carrying the negative $\phi$
phase.  Consequently, the factor $v_p u_q$ is represented by the
opposite blue arrow, from $q$ to $p$.

For the bra factors, a red arrow from site $p$ to site $q$ represents
$v_p^*u_q^*$:
\begin{center}
\begin{pspicture}(0,-0.3)(5,0.8)
\psline[linewidth=0.06cm,linearc=1.2,linecolor=red](1,0)(2,1)(3,0)
\psarc[linewidth=0.12cm,linecolor=red]{-<}(2,-0.2){0.705}{87}{93}
\psdots[dotscale=1.5](1,0)(3,0)
\rput(0.6,0){$\mathlarger{\mathlarger{p}}$}
\rput(3.4,0){$\mathlarger{\mathlarger{q}}$}
\rput(2,0){$\mathlarger{\mathlarger{v_p^* u_q^*}}$}
\end{pspicture}
\end{center}
This is again the direction of the phase flow, since
$v_p^*u_q^*$ carries the phase factor
$e^{i(\phi_p-\phi_q)/2}$.  It is useful to regard the blue ket arrows
as living in a bottom layer and the red bra arrows as living in a top
layer.  The colours will then mainly serve to remind us which layer an
arrow belongs to.

With this convention, the factor associated with a given site depends
on both the layer of the arrow and whether the arrow is incoming or
outgoing.  The four possibilities are summarized in
Table~\ref{tab:arrows}.
\begin{table}[h!]
  \centering
  \begin{tabular}{ |c| c| }
  \hline
     \textbf{Arrow type and direction} & \textbf{Factor} \\
     \hline\hline
    top level (red), outgoing & $v^*=\sin ({\theta}/{2}) \exp({i\phi/2})$ \\
    \hline
     top level (red), incoming & $u^*=\cos ({\theta}/{2}) \exp({-i\phi/2})$\\
    \hline
    bottom level (blue), outgoing & $u=\cos ({\theta}/{2}) \exp({i\phi/2})$\\
    \hline
    bottom level (blue), incoming & $v=\sin ({\theta}/{2}) \exp({-i\phi/2})$
    \\ [1ex]
    \hline
\end{tabular}
  \caption{Local arrow conventions used in the two-layer representation
of the coherent-state overlap.  Red arrows belong to the top layer and
come from the complex-conjugated bra factors, while blue arrows belong
to the bottom layer and come from the ket factors.  At a given site, an
outgoing arrow carries the phase factor $e^{+i\phi/2}$ and an incoming
arrow carries $e^{-i\phi/2}$.  The corresponding $\theta$-dependent
factors are listed in the second column.}
\label{tab:arrows}
\end{table}

We also need a sign convention.  On the checkerboard lattice, after
choosing horizontal and vertical dimers as the local basis, the
relevant bonds lie on the underlying square lattice, which is
bipartite.  We associate a factor of $(-1)$ with a blue bottom-layer
arrow directed from sublattice $B$ to sublattice $A$, and with a red
top-layer arrow directed from sublattice $A$ to sublattice $B$.  An
analogous convention can be chosen on the pyrochlore lattice once one
omits one of the three local dimerizations from the chosen basis, as
illustrated in Fig.~\ref{fig:pyrochlore_configs}.  Equivalently, using
the site-colouring scheme of Ref.~[\onlinecite{Parameswaran2009}], we
may restrict the singlet basis to singlets between red-blue/yellow-green
sites and red-yellow/green-blue sites, and designate red and green
sites as sublattice $A$ and blue and yellow sites as sublattice $B$ for
the purpose of this sign convention.

We now identify which terms in the local expansion survive the
coherent-state integration.  At every site there are two red arrows and
two blue arrows, each of which may be incoming or outgoing.  From
Table~\ref{tab:arrows}, every outgoing arrow contributes a phase
$e^{+i\phi_i/2}$, whereas every incoming arrow contributes
$e^{-i\phi_i/2}$.  Thus a local term carries the phase
\[
    \exp\left[
    \frac{i}{2}
    \left(N_{\rm out}-N_{\rm in}\right)\phi_i
    \right],
\]
where $N_{\rm out}$ and $N_{\rm in}$ are the numbers of outgoing and
incoming arrows at site $i$.  The integral over $\phi_i$ vanishes
unless $N_{\rm out}=N_{\rm in}$.  Explicitly,
\begin{equation*}
\int_{0}^{2\pi} e^{\pm i\phi}d\phi = \int_{0}^{2\pi} e^{\pm 2i\phi}d\phi = 0 .
\end{equation*}

The same selection rule applies in the numerator of a spin--spin
correlation function.  Indeed, by rotational invariance,
\begin{multline}
\label{eq:overlap_corr}
\int \prod_{k} d \boldsymbol{\Omega}_k  \left(\Psi^{\mathcal{L}'}(\Omega)\right)^\ast \! \Psi^{\mathcal{L}}(\Omega)\; \boldsymbol{\Omega}_i\cdot\boldsymbol{\Omega}_j \\
= 3 \int \prod_{k} d \boldsymbol{\Omega}_k  \left(\Psi^{\mathcal{L}'}(\Omega)\right)^\ast \!\Psi^{\mathcal{L}}(\Omega)\; {\Omega}_i^z{\Omega}_j^z .
\end{multline}
Here $\Omega_i^z\Omega_j^z=\cos\theta_i\cos\theta_j$ contains no
$\phi$ dependence.  Hence the $\phi$ integration imposes the same
condition as in the norm calculation.

We therefore conclude that only local arrow configurations with two
incoming and two outgoing arrows survive at every site.  There are
\({4\choose 2}=6\) such configurations.  Three of them are
\begin{center}
\begin{pspicture}(0,0.5)(7.5,1.5)
\psarc[linewidth=0.1cm,linecolor=blue]{->}(1.707,1.707){1}{225}{265}
\psarc[linewidth=0.08cm,linecolor=blue](1.707,1.707){1}{225}{270}
\psarc[linewidth=0.1cm,linecolor=blue]{>-}(0.293,1.707){1}{280}{315}
\psarc[linewidth=0.08cm,linecolor=blue](0.293,1.707){1}{270}{315}
\psarc[linewidth=0.1cm,linecolor=red]{<-}(1.707,0.293){1}{95}{135}
\psarc[linewidth=0.08cm,linecolor=red](1.707,0.293){1}{90}{135}
\psarc[linewidth=0.1cm,linecolor=red]{-<}(0.293,0.293){1}{45}{80}
\psarc[linewidth=0.08cm,linecolor=red](0.293,0.293){1}{45}{90}
\psarc[linewidth=0.1cm,linecolor=blue]{->}(4.207,1.707){1}{225}{265}
\psarc[linewidth=0.08cm,linecolor=blue](4.207,1.707){1}{225}{270}
\psarc[linewidth=0.1cm,linecolor=blue]{>-}(2.793,1.707){1}{280}{315}
\psarc[linewidth=0.08cm,linecolor=blue](2.793,1.707){1}{270}{315}
\psarc[linewidth=0.1cm,linecolor=red]{>-}(4.207,0.293){1}{100}{135}
\psarc[linewidth=0.08cm,linecolor=red](4.207,0.293){1}{90}{135}
\psarc[linewidth=0.1cm,linecolor=red]{->}(2.793,0.293){1}{45}{85}
\psarc[linewidth=0.08cm,linecolor=red](2.793,0.293){1}{45}{90}
\psarc[linewidth=0.1cm,linecolor=blue]{->}(6.707,1.707){1}{225}{265}
\psarc[linewidth=0.08cm,linecolor=blue](6.707,1.707){1}{225}{270}
\psarc[linewidth=0.1cm,linecolor=blue]{<-}(5.293,1.707){1}{275}{315}
\psarc[linewidth=0.08cm,linecolor=blue](5.293,1.707){1}{270}{315}
\psarc[linewidth=0.1cm,linecolor=red]{>-}(6.707,0.293){1}{100}{135}
\psarc[linewidth=0.08cm,linecolor=red](6.707,0.293){1}{90}{135}
\psarc[linewidth=0.1cm,linecolor=red]{-<}(5.293,0.293){1}{45}{80}
\psarc[linewidth=0.08cm,linecolor=red](5.293,0.293){1}{45}{90}
\psdots[dotscale=1.5](1,1)(3.5,1)(6,1)
\end{pspicture}
\end{center}
and the other three are obtained by reversing all arrows in these
pictures.  These are the six local vertices of a six-vertex model,
although here the vertices are not arranged as a two-dimensional
square-lattice vertex model.  Rather, they live on the doubled
transition graph built from the two loop configurations
$\mathcal L'$ and $\mathcal L$.

A useful feature of the sign convention above is that the minus signs
always occur in pairs for each allowed vertex, independent of whether
the site belongs to sublattice $A$ or $B$.  Thus every allowed local
vertex has a non-negative weight.  This positivity is what makes the
norm of a positive-amplitude loop superposition interpretable as a
classical partition function.

The weights of the allowed vertices are obtained by integrating over
$\theta_i$.  With the numbering corresponding to the order of the
vertices shown above, the first two weights are
\begin{subequations}\label{eq:vertex_weights}
  \begin{multline}\label{eq:vertex_weight12}
    w_1=w_2 =\int d\Omega_i \cos^2{\frac{\theta_i}{2}} \sin^2{\frac{\theta_i}{2}} \\
    = 2\pi \int_0^\pi d\theta_i \sin\theta_i \cos^2{\frac{\theta_i}{2}} \sin^2{\frac{\theta_i}{2}} = \frac{2\pi}{3}
  \end{multline}
and the third is
    \begin{multline}\label{eq:vertex_weight3}
    w_3 = \int d\Omega_i \cos^4{\frac{\theta_i}{2}}  \\
    = 2\pi \int_0^\pi d\theta_i \sin\theta_i \cos^4{\frac{\theta_i}{2}} = \frac{4\pi}{3}.
  \end{multline}
\end{subequations}
The vertices obtained by reversing all arrows have the same weights
with $u$ and $v$ interchanged: $w_4=w_5=2\pi/3$ and $w_6=4\pi/3$.
For the last equality we used
\[
    \int d\Omega_i\,
    \cos^4\left(\frac{\theta_i}{2}\right)
    =
    \int d\Omega_i\,
    \sin^4\left(\frac{\theta_i}{2}\right).
\]
Thus
\begin{equation}\label{eq:vertex_rel_weights}
w_1=w_2=w_4=w_5=\frac{2\pi}{3} = \frac{w_3}{2}= \frac{w_6}{2} \, .
\end{equation}

Equation~\eqref{eq:vertex_rel_weights} has an important graphical
interpretation.  We may assign the common weight
\[
    w=\frac{2\pi}{3}
\]
to each elementary resolution connecting the top and bottom layers.
The vertices with weight $w$ have one such resolution, while the
vertices with weight $2w$ have two possible resolutions.  In all
resolutions, the arrows change layers at each site.  This is shown
graphically below.

\begin{center}
\begin{pspicture}(0,0.5)(5,1.5)
\psarc[linewidth=0.1cm,linecolor=blue]{->}(1.707,1.707){1}{225}{265}
\psarc[linewidth=0.08cm,linecolor=blue](1.707,1.707){1}{225}{270}
\psarc[linewidth=0.1cm,linecolor=blue]{>-}(0.293,1.707){1}{280}{315}
\psarc[linewidth=0.08cm,linecolor=blue](0.293,1.707){1}{270}{315}
\psarc[linewidth=0.1cm,linecolor=red]{<-}(1.707,0.293){1}{95}{135}
\psarc[linewidth=0.08cm,linecolor=red](1.707,0.293){1}{90}{135}
\psarc[linewidth=0.1cm,linecolor=red]{-<}(0.293,0.293){1}{45}{80}
\psarc[linewidth=0.08cm,linecolor=red](0.293,0.293){1}{45}{90}
\psarc[linewidth=0.1cm,linecolor=gray]{->}(4.207,1.707){1}{255}{265}
\psarc[linewidth=0.08cm,linecolor=gray](4.207,1.707){1}{235}{270}
\psarc[linewidth=0.1cm,linecolor=gray]{>-}(2.793,1.707){1}{280}{290}
\psarc[linewidth=0.08cm,linecolor=gray](2.793,1.707){1}{270}{315}
\psarc[linewidth=0.1cm,linecolor=gray]{<-}(4.207,0.293){1}{95}{110}
\psarc[linewidth=0.08cm,linecolor=gray](4.207,0.293){1}{90}{135}
\psarc[linewidth=0.1cm,linecolor=gray]{-<}(2.793,0.293){1}{70}{80}
\psarc[linewidth=0.08cm,linecolor=gray](2.793,0.293){1}{55}{90}
\psdots[dotscale=1.5](1,1)
\rput(2.2,1){$\mathlarger{\mathlarger{=}}$}
\end{pspicture}
\end{center}

\begin{center}
\begin{pspicture}(0,0.5)(5,1.5)
\psarc[linewidth=0.1cm,linecolor=blue]{->}(1.707,1.707){1}{225}{265}
\psarc[linewidth=0.08cm,linecolor=blue](1.707,1.707){1}{225}{270}
\psarc[linewidth=0.1cm,linecolor=blue]{>-}(0.293,1.707){1}{280}{315}
\psarc[linewidth=0.08cm,linecolor=blue](0.293,1.707){1}{270}{315}
\psarc[linewidth=0.1cm,linecolor=red]{>-}(1.707,0.293){1}{100}{135}
\psarc[linewidth=0.08cm,linecolor=red](1.707,0.293){1}{90}{135}
\psarc[linewidth=0.1cm,linecolor=red]{->}(0.293,0.293){1}{45}{85}
\psarc[linewidth=0.08cm,linecolor=red](0.293,0.293){1}{45}{90}
\psarc[linewidth=0.1cm,linecolor=gray]{->}(4.207,1.707){1}{255}{265}
\psarc[linewidth=0.08cm,linecolor=gray](4.207,1.707){1}{255}{270}
\psarc[linewidth=0.1cm,linecolor=gray]{>-}(2.793,1.707){1}{280}{290}
\psarc[linewidth=0.08cm,linecolor=gray](2.793,1.707){1}{270}{290}
\psarc[linewidth=0.1cm,linecolor=gray]{>-}(4.207,0.293){1}{100}{110}
\psarc[linewidth=0.08cm,linecolor=gray](4.207,0.293){1}{90}{110}
\psarc[linewidth=0.1cm,linecolor=gray]{->}(2.793,0.293){1}{75}{85}
\psarc[linewidth=0.08cm,linecolor=gray](2.793,0.293){1}{75}{90}
\psarc[linewidth=0.08cm,linecolor=gray](3.1,1){0.25}{-70}{90}
\psarc[linewidth=0.08cm,linecolor=gray](3.9,1){0.25}{100}{270}
\psdots[dotscale=1.5](1,1)
\rput(2.2,1){$\mathlarger{\mathlarger{=}}$}
\end{pspicture}
\end{center}

\begin{center}
\begin{pspicture}(0,0.5)(7.5,1.5)
\psarc[linewidth=0.1cm,linecolor=blue]{->}(1.707,1.707){1}{225}{265}
\psarc[linewidth=0.08cm,linecolor=blue](1.707,1.707){1}{225}{270}
\psarc[linewidth=0.1cm,linecolor=blue]{<-}(0.293,1.707){1}{275}{315}
\psarc[linewidth=0.08cm,linecolor=blue](0.293,1.707){1}{270}{315}
\psarc[linewidth=0.1cm,linecolor=red]{>-}(1.707,0.293){1}{100}{135}
\psarc[linewidth=0.08cm,linecolor=red](1.707,0.293){1}{90}{135}
\psarc[linewidth=0.1cm,linecolor=red]{-<}(0.293,0.293){1}{45}{80}
\psarc[linewidth=0.08cm,linecolor=red](0.293,0.293){1}{45}{90}
\psarc[linewidth=0.1cm,linecolor=gray]{->}(4.207,1.707){1}{255}{265}
\psarc[linewidth=0.08cm,linecolor=gray](4.207,1.707){1}{235}{270}
\psarc[linewidth=0.1cm,linecolor=gray]{<-}(2.793,1.707){1}{275}{288}
\psarc[linewidth=0.08cm,linecolor=gray](2.793,1.707){1}{270}{315}
\psarc[linewidth=0.1cm,linecolor=gray]{>-}(4.207,0.293){1}{100}{110}
\psarc[linewidth=0.08cm,linecolor=gray](4.207,0.293){1}{90}{135}
\psarc[linewidth=0.1cm,linecolor=gray]{-<}(2.793,0.293){1}{70}{80}
\psarc[linewidth=0.08cm,linecolor=gray](2.793,0.293){1}{55}{90}
\psarc[linewidth=0.1cm,linecolor=gray]{->}(6.707,1.707){1}{255}{265}
\psarc[linewidth=0.08cm,linecolor=gray](6.707,1.707){1}{255}{270}
\psarc[linewidth=0.1cm,linecolor=gray]{<-}(5.293,1.707){1}{275}{288}
\psarc[linewidth=0.08cm,linecolor=gray](5.293,1.707){1}{270}{290}
\psarc[linewidth=0.1cm,linecolor=gray]{>-}(6.707,0.293){1}{100}{110}
\psarc[linewidth=0.08cm,linecolor=gray](6.707,0.293){1}{90}{110}
\psarc[linewidth=0.1cm,linecolor=gray]{-<}(5.293,0.293){1}{70}{80}
\psarc[linewidth=0.08cm,linecolor=gray](5.293,0.293){1}{70}{90}
\psarc[linewidth=0.08cm,linecolor=gray](5.6,1){0.25}{-90}{72}
\psarc[linewidth=0.08cm,linecolor=gray](6.4,1){0.25}{100}{270}
\psdots[dotscale=1.5](1,1)
\rput(2.2,1){$\mathlarger{\mathlarger{=}}$}
\rput(4.75,1){$\mathlarger{\mathlarger{+}}$}
\end{pspicture}
\end{center}

The grey segments produced by these resolutions form closed oriented
loops.  Each such loop alternates between the two layers at every site
of the original lattice: it follows a dimer from the ket configuration
$\mathcal L$, changes layer at a site, follows a dimer from the bra
configuration $\mathcal L'$, changes layer again, and so on.  These
loops should not be confused with the original AKLT-like singlet loops
of a single configuration.  They are transition loops built from the
pair of configurations $\mathcal L$ and $\mathcal L'$.  Since
$\mathcal L$ and $\mathcal L'$ need not be the same, these transition
loops need not lie in any single geometric plane.

The common weight $w=2\pi/3$ is now associated with every legitimate
local resolution.  Once this convention is adopted, reversing all
arrows on any one closed transition loop leaves the weight unchanged.
Therefore each closed loop contributes a factor of $2$, corresponding
to its two possible orientations.

We thus obtain the following statistical-mechanical representation of
the overlap integral:
\begin{multline}\label{eq:overlap_int}
\int \prod_{k} d \boldsymbol{\Omega}_k  \left(\Psi^{\mathcal{L}'}(\Omega)\right)^\ast \Psi^{\mathcal{L}}(\Omega) \\
= \left(\frac{2\pi}{3}\right)^N
\sum_{\lambda\in\mathfrak R(\mathcal L',\mathcal L)}
2^{\Lambda(\lambda)} .
\end{multline}
Here $N$ is the number of physical sites.  The set
$\mathfrak R(\mathcal L',\mathcal L)$ denotes all legitimate
undirected resolutions of the doubled transition graph formed from
the singlet bonds in $\mathcal L$ and $\mathcal L'$.  For a given
resolution $\lambda$, $\Lambda(\lambda)$ is the number of closed
transition loops.  The factor $2^{\Lambda(\lambda)}$ counts the two
allowed orientations of each closed loop.  In what follows, the common
factor $\left(2\pi/3\right)^N$ will often be suppressed, since it is
the same for all legitimate resolutions and only multiplies the
partition function by an overall constant.

\section{Spin--spin correlations in the loop representation}
\label{app:loops_correlations}

In Appendix~\ref{app:stat_mech_loops} we developed a loop
representation for overlaps between singlet-loop states.  The purpose
of the present appendix is to extend that representation to
spin--spin correlation functions.  This is useful for two related
reasons.  First, it gives a graphical rule for inserting spin
operators into the loop ensemble.  Second, it makes transparent why
spin correlations are controlled by loop connectivity: if the two spin
insertions are not connected by an appropriate loop in the doubled
representation, their contributions cancel after summing over loop
orientations.

We begin from Eq.~\eqref{eq:overlap_corr}.  The integrations over all
$\boldsymbol{\Omega}_k$ with $k\neq i,j$ are the same as in the overlap
calculation of Appendix~\ref{app:stat_mech_loops}.  Thus all sites
away from the spin insertions carry the same vertex weights as before.
The only new ingredients are the local integrals at the insertion
sites $i$ and $j$.  Taking the spin component to be $S^z$, and using
$\Omega_i^z=\cos\theta_i$, the relevant local integrals are
\begin{subequations}\label{eq:spin_weights}
  \begin{multline}\label{eq:spin_weight1}
\int d\Omega_i \,\Omega_i^z\, \cos^2{\frac{\theta_i}{2}} \sin^2{\frac{\theta_i}{2}}  \\
    = 2\pi \int_0^\pi d\theta_i \sin\theta_i \cos\theta_i
    \cos^2{\frac{\theta_i}{2}} \sin^2{\frac{\theta_i}{2}}
    = 0 \, ,
  \end{multline}
  \begin{multline}\label{eq:spin_weight2}
   \int d\Omega_i \,\Omega_i^z\, \cos^4{\frac{\theta_i}{2}}  \\
    = 2\pi \int_0^\pi d\theta_i \sin\theta_i \cos\theta_i
    \cos^4{\frac{\theta_i}{2}}
    = \frac{2\pi}{3} \, ,
  \end{multline}
    \begin{multline}\label{eq:spin_weight3}
   \int d\Omega_i \,\Omega_i^z\, \sin^4{\frac{\theta_i}{2}}  \\
    = 2\pi \int_0^\pi d\theta_i \sin\theta_i \cos\theta_i
    \sin^4{\frac{\theta_i}{2}}
    = - \frac{2\pi}{3} \, .
  \end{multline}
\end{subequations}
The same formulas hold for the integration over $\boldsymbol{\Omega}_j$.
The vanishing of Eq.~\eqref{eq:spin_weight1} has a simple meaning: a
spin insertion gives zero whenever the local vertex contains equal
powers of $|u_i|^2$ and $|v_i|^2$.  In contrast,
Eqs.~\eqref{eq:spin_weight2} and \eqref{eq:spin_weight3} distinguish
whether the local configuration is weighted by $|u_i|^4$ or by
$|v_i|^4$, and the two possibilities come with opposite signs.  This
sign is the graphical remnant of the factor
$\Omega_i^z=\cos\theta_i$.

These local rules immediately imply a connectivity criterion.  If the
two sites $i$ and $j$ are not connected by either of the two loops
passing through site $i$ in the doubled loop representation, their
spin correlations vanish.  To see this, recall first that reversing
the orientation of any closed loop does not change the weight of a
configuration in the partition function.  In the numerator of the
spin--spin correlation function, however, reversing the orientation of
a loop through site $i$ changes the local factor at that site.  From
Table~\ref{tab:arrows}, reversing a loop through a given site
interchanges
\[
    |v|^2 \leftrightarrow |u|^2 ,
    \qquad
    \text{or equivalently}
    \qquad
    \sin^2\frac{\theta}{2}
    \leftrightarrow
    \cos^2\frac{\theta}{2}.
\]
Reversing both loops passing through the same site therefore either
interchanges
\[
    |v|^4 \leftrightarrow |u|^4 ,
    \qquad
    \text{or equivalently}
    \qquad
    \sin^4\frac{\theta}{2}
    \leftrightarrow
    \cos^4\frac{\theta}{2},
\]
or leaves the mixed factor
$|u|^2|v|^2$ unchanged.  In the first case the contribution changes
sign, since Eq.~\eqref{eq:spin_weight2} is exchanged with
Eq.~\eqref{eq:spin_weight3}.  In the second case the contribution is
zero both before and after the reversal, by Eq.~\eqref{eq:spin_weight1}.
Thus, if no loop carrying the sign change also passes through site
$j$, each contribution can be paired with another contribution of the
same statistical weight and opposite sign.  Their sum is zero.

Consequently, the spin--spin correlation function is bounded by the
probability that the two sites are connected by at least one loop in
the doubled loop representation.  This is the statistical-mechanical
counterpart of the fixed-loop result in Sec.~\ref{sec:gs_correlations}:
spin correlations are not determined only by the Euclidean distance
between two sites, but by the probability that the two spin insertions
remain connected in the loop ensemble.

The same graphical rules also determine the sign of the correlations
whenever the loop representation has non-negative statistical weights.
For a single loop connecting sites $i$ and $j$, the sign is set by the
parity of the number of steps separating the two sites along that loop.
With the sublattice convention used in Appendix~\ref{app:stat_mech_loops}
for the pyrochlore lattice, sites on different sublattices are
separated by an odd number of loop steps, while sites on the same
sublattice are separated by an even number.  If the sites are on
different sublattices, the non-vanishing local factors at $i$ and $j$
are of opposite type, one of Eq.~\eqref{eq:spin_weight2} and the other
of Eq.~\eqref{eq:spin_weight3}, giving a negative product.  If the
sites are on the same sublattice, the two non-vanishing local factors
are of the same type, giving a positive product.  Thus, whenever the
plasma analogy is a genuine positive-weight statistical ensemble, the
correlations have a staggered, N\'eel-like sign structure.

There is, however, an important caveat.  The preceding sign argument
uses the interpretation of the squared wavefunction as a positive
statistical weight.  If the superposition of loop states contains
relative negative signs, then some off-diagonal terms in
$|\Psi(\Omega)|^2$ also carry negative signs, and the partition-function
interpretation is lost.  A simple example is the state made solely of
diagonal dimers.  Its correlations are nonzero along the diagonals and
alternate from site to site along each diagonal.  In terms of the
underlying square-lattice sublattices, however, all sites along a
diagonal belong to the same sublattice.  There is no contradiction:
when the diagonal-dimer state is expressed in the vertical/horizontal
dimer basis, it contains relative negative amplitudes, so the
positive-weight sign argument no longer applies.

We now give a more explicit graphical representation of the spin
insertions.  The graphical form of the three local integrals in
Eq.~\eqref{eq:spin_weights} is as follows:
\begin{center}
\begin{pspicture}(0,0.5)(5,1.5)
\psarc[linewidth=0.1cm,linecolor=blue]{->}(1.707,1.707){1}{225}{265}
\psarc[linewidth=0.08cm,linecolor=blue](1.707,1.707){1}{225}{270}
\psarc[linewidth=0.1cm,linecolor=blue]{>-}(0.293,1.707){1}{280}{315}
\psarc[linewidth=0.08cm,linecolor=blue](0.293,1.707){1}{270}{315}
\psarc[linewidth=0.1cm,linecolor=red]{<-}(1.707,0.293){1}{95}{135}
\psarc[linewidth=0.08cm,linecolor=red](1.707,0.293){1}{90}{135}
\psarc[linewidth=0.1cm,linecolor=red]{-<}(0.293,0.293){1}{45}{80}
\psarc[linewidth=0.08cm,linecolor=red](0.293,0.293){1}{45}{90}
\psarc[linewidth=0.1cm,linecolor=blue]{-<}(4.207,1.707){1}{225}{260}
\psarc[linewidth=0.08cm,linecolor=blue](4.207,1.707){1}{225}{270}
\psarc[linewidth=0.1cm,linecolor=blue]{<-}(2.793,1.707){1}{275}{315}
\psarc[linewidth=0.08cm,linecolor=blue](2.793,1.707){1}{270}{315}
\psarc[linewidth=0.1cm,linecolor=red]{<-}(4.207,0.293){1}{95}{135}
\psarc[linewidth=0.08cm,linecolor=red](4.207,0.293){1}{90}{135}
\psarc[linewidth=0.1cm,linecolor=red]{-<}(2.793,0.293){1}{45}{80}
\psarc[linewidth=0.08cm,linecolor=red](2.793,0.293){1}{45}{90}
\psdots[dotscale=2,dotstyle=diamond](1,1)(3.5,1)
\rput(2.2,1){$\mathlarger{\mathlarger{=}}$}
\rput(4.9,1){$\mathlarger{\mathlarger{=\; 0}}$}
\end{pspicture}
\end{center}
with two more vanishing contributions obtained by reversing all arrows
of the two vertices above.  The diamond-shaped vertex denotes the
insertion of the operator $S^z$, and is used to distinguish this
vertex from the ordinary overlap vertices with weights $w$ in
Eqs.~\eqref{eq:vertex_weights}.

The two nonzero spin-insertion vertices are
\begin{center}
\begin{pspicture}(0,0.5)(3.5,1.5)
\psarc[linewidth=0.1cm,linecolor=blue]{->}(1.707,1.707){1}{225}{265}
\psarc[linewidth=0.08cm,linecolor=blue](1.707,1.707){1}{225}{270}
\psarc[linewidth=0.1cm,linecolor=blue]{<-}(0.293,1.707){1}{275}{315}
\psarc[linewidth=0.08cm,linecolor=blue](0.293,1.707){1}{270}{315}
\psarc[linewidth=0.1cm,linecolor=red]{>-}(1.707,0.293){1}{100}{135}
\psarc[linewidth=0.08cm,linecolor=red](1.707,0.293){1}{90}{135}
\psarc[linewidth=0.1cm,linecolor=red]{-<}(0.293,0.293){1}{45}{80}
\psarc[linewidth=0.08cm,linecolor=red](0.293,0.293){1}{45}{90}
\psdots[dotscale=2,dotstyle=diamond](1,1)
\rput(2.5,1){$\mathlarger{\mathlarger{=\;\frac{2\pi}{3}}} \, $;}
\end{pspicture}
\end{center}
\begin{center}
\begin{pspicture}(0,0.5)(3.5,1.5)
\psarc[linewidth=0.1cm,linecolor=blue]{-<}(1.707,1.707){1}{225}{260}
\psarc[linewidth=0.08cm,linecolor=blue](1.707,1.707){1}{225}{270}
\psarc[linewidth=0.1cm,linecolor=blue]{>-}(0.293,1.707){1}{280}{315}
\psarc[linewidth=0.08cm,linecolor=blue](0.293,1.707){1}{270}{315}
\psarc[linewidth=0.1cm,linecolor=red]{<-}(1.707,0.293){1}{95}{135}
\psarc[linewidth=0.08cm,linecolor=red](1.707,0.293){1}{90}{135}
\psarc[linewidth=0.1cm,linecolor=red]{->}(0.293,0.293){1}{45}{85}
\psarc[linewidth=0.08cm,linecolor=red](0.293,0.293){1}{45}{90}
\psdots[dotscale=2,dotstyle=diamond](1,1)
\rput(2.7,1){$\mathlarger{\mathlarger{=\,-\frac{2\pi}{3}}} \, $.}
\end{pspicture}
\end{center}
Two points should be emphasized.  First, the value of the
spin-insertion vertex is determined solely by the net flow of arrows
between the two layers.  Second, the corresponding ordinary overlap
vertices have weight $w_3=w_6=4\pi/3$, whereas the absolute value of
the spin-insertion vertex is only $2\pi/3$.  Thus a spin-insertion
vertex carries one half of the magnitude of the corresponding ordinary
overlap vertex, together with a sign fixed by the direction of the net
flow.

This observation gives the following rule for the numerator of
Eq.~\eqref{eq:corr_function_gen}.  All sites other than $i$ and $j$
are treated exactly as in the overlap calculation.  At the two
spin-insertion sites, one may use either of two equivalent
prescriptions.

First, one may resolve each diamond-shaped vertex in the same two ways
as an ordinary overlap vertex, but assign an additional factor of
$1/2$ to each spin-insertion site.  The two insertions therefore
contribute an overall factor of $1/4$.  The sign of each contribution
is then determined by the net flow of arrows between the two layers, as
explained above.

Equivalently, one may choose only one of the two possible ways of
connecting the upper-layer dimers to the lower-layer dimers at each
spin-insertion site, and omit the explicit factor of $1/2$ at that
site.  The final answer is independent of which of the two local
resolutions is chosen.  This equivalence is the local graphical form
of the statement that the two resolutions may either be summed and
weighted by $1/2$, or represented by a single effective spin-insertion
resolution.

We now test this representation for a fixed loop state,
$\mathcal L=\mathcal L'$.  In this case, two spins can be correlated
only if they belong to the same loop of dimers in $\mathcal L$.  These
AKLT-like dimer loops should not be confused with the doubled loops
$\lambda$ appearing in the overlap representation of
Appendix~\ref{app:stat_mech_loops}.  As in the main text, all dimer
loops that do not contain either $i$ or $j$ factor out and cancel
between numerator and denominator.  It is therefore enough to consider
the single dimer loop passing through sites $i$ and $j$.  Let this loop
have length $\ell$, and let $s$ be the number of steps along one of the
two arcs connecting $i$ and $j$.  The other arc has length $\ell-s$.

We first evaluate the denominator, i.e., the overlap partition function
for this doubled loop.  Suppose initially that at every site the
connection is of the pass-through type,
\begin{center}
\begin{pspicture}(2.5,0.5)(5,1.5)
\psarc[linewidth=0.08cm,linecolor=gray](4.207,1.707){1}{235}{270}
\psarc[linewidth=0.08cm,linecolor=gray](2.793,1.707){1}{270}{315}
\psarc[linewidth=0.08cm,linecolor=gray](4.207,0.293){1}{90}{135}
\psarc[linewidth=0.08cm,linecolor=gray](2.793,0.293){1}{55}{90}
\end{pspicture}
\end{center}
so that no local turn-back occurs.  The doubled representation then
contains two loops tracing the original dimer loop.  Each of these two
loops can be oriented independently, giving a contribution $2^2=4$.
Here and below the common factor $2\pi/3$ per site has been pulled out
and will be restored only as an overall multiplicative constant.

The remaining denominator configurations contain at least one
turn-back of the form
\begin{center}
\begin{pspicture}(2.5,0.5)(5,1.5)
\psarc[linewidth=0.08cm,linecolor=gray](4.207,1.707){1}{250}{270}
\psarc[linewidth=0.08cm,linecolor=gray](2.793,1.707){1}{270}{290}
\psarc[linewidth=0.08cm,linecolor=gray](4.207,0.293){1}{90}{110}
\psarc[linewidth=0.08cm,linecolor=gray](2.793,0.293){1}{70}{90}
\psarc[linewidth=0.08cm,linecolor=gray](3.1,1){0.23}{-80}{80}
\psarc[linewidth=0.08cm,linecolor=gray](3.9,1){0.23}{100}{260}
\end{pspicture}
\end{center}
If there are $t$ such turn-backs on a loop of length $\ell$, their
positions can be chosen in ${\ell\choose t}$ ways.  The turn-backs cut
the doubled loop representation into $t$ closed loops, and each such
loop can be independently oriented in two ways.  Hence the contribution
with $t\geq 1$ turn-backs is ${\ell\choose t}2^t$.  Summing over all
possible $t$ gives
\begin{equation*}
  \sum_{t=1}^{\ell}{\ell \choose t} 2^t
  =
  (1+2)^\ell - 1
  =
  3^\ell - 1 \, .
\end{equation*}
Adding the no-turn-back contribution gives the denominator
\begin{equation}
\label{eq:AKLT_part_function}
Z = 4 + 3^\ell - 1 = 3^\ell+3 = 3^\ell \left(1+3^{1-\ell}\right) .
\end{equation}
The common factor $(2\pi/3)^\ell$ has been omitted from
Eq.~\eqref{eq:AKLT_part_function}, since the same factor appears in
the numerator and cancels in the normalized correlation function.

We next evaluate the numerator of Eq.~\eqref{eq:corr_function_gen} for
$S_i^zS_j^z$.  We consider three cases, according to how the
diamond-shaped vertices at sites $i$ and $j$ are resolved.

\emph{(i) Both spin-insertion vertices are pass-through vertices.}
\begin{center}
\begin{pspicture}(2.5,0.5)(5,1.5)
\psarc[linewidth=0.08cm,linecolor=gray](4.207,1.707){1}{235}{270}
\psarc[linewidth=0.08cm,linecolor=gray](2.793,1.707){1}{270}{315}
\psarc[linewidth=0.08cm,linecolor=gray](4.207,0.293){1}{90}{135}
\psarc[linewidth=0.08cm,linecolor=gray](2.793,0.293){1}{55}{90}
\end{pspicture}
\end{center}
First suppose that there are no additional turn-backs anywhere on the
loop.  There are four possible orientations of the two doubled loops,
but only the two configurations with opposite loop orientations give a
nonzero spin-insertion contribution.  Their combined contribution is
\[
    2(-1)^s ,
\]
where the factor $(-1)^s$ is the staggered sign accumulated along the
arc connecting $i$ and $j$.

Now allow additional turn-backs away from $i$ and $j$.  A nonzero
connected contribution is obtained only when all additional turn-backs
lie on the same arc between $i$ and $j$.  If turn-backs occur on both
arcs, the two spin insertions no longer lie on a common doubled loop,
and the orientation-reversal cancellation described above makes the
net contribution vanish.  For $t$ turn-backs, the number of allowed
placements is therefore
\[
    {s-1\choose t}+{\ell-s-1\choose t},
\]
where the first term places all turn-backs on the arc of length $s$
and the second term places all turn-backs on the complementary arc of
length $\ell-s$.  The corresponding orientation factor is $2^t$, and
the contribution is
\begin{equation*}
  2^{t}(-1)^s \left[{{s-1} \choose t}+ {{\ell-s-1} \choose t}\right] .
\end{equation*}
The allowed values of $t$ are understood to be limited by the upper
entries of the corresponding binomial coefficients.  Thus the total
contribution in this case is
\begin{multline*}
 (-1)^s\left\{ 2+ \sum_{t=1}^{s-1} 2^{t} {{s-1} \choose {t}} + \sum_{t=1}^{\ell-s-1} 2^{t} {{\ell-s-1} \choose {t}}\right\}\\
 = (-1)^s\left\{\sum_{t=0}^{s-1} 2^{t} {{s-1} \choose {t}} + \sum_{t=0}^{\ell-s-1} 2^{t} {{\ell-s-1} \choose {t}}\right\}\\
 = (-1)^s\left(3^{s-1}+3^{\ell-s-1}\right) \, .
\end{multline*}

\emph{(ii) One spin-insertion vertex is pass-through and the other is
a turn-back.}  Take, for example, site $i$ to be pass-through,
\begin{center}
\begin{pspicture}(2.5,0.5)(5,1.5)
\psarc[linewidth=0.08cm,linecolor=gray](4.207,1.707){1}{235}{270}
\psarc[linewidth=0.08cm,linecolor=gray](2.793,1.707){1}{270}{315}
\psarc[linewidth=0.08cm,linecolor=gray](4.207,0.293){1}{90}{135}
\psarc[linewidth=0.08cm,linecolor=gray](2.793,0.293){1}{55}{90}
\end{pspicture}
\end{center}
and site $j$ to be a turn-back:
\begin{center}
\begin{pspicture}(2.5,0.5)(5,1.5)
\psarc[linewidth=0.08cm,linecolor=gray](4.207,1.707){1}{250}{270}
\psarc[linewidth=0.08cm,linecolor=gray](2.793,1.707){1}{270}{290}
\psarc[linewidth=0.08cm,linecolor=gray](4.207,0.293){1}{90}{110}
\psarc[linewidth=0.08cm,linecolor=gray](2.793,0.293){1}{70}{90}
\psarc[linewidth=0.08cm,linecolor=gray](3.1,1){0.23}{-80}{80}
\psarc[linewidth=0.08cm,linecolor=gray](3.9,1){0.23}{100}{260}
\end{pspicture}
\end{center}
If there are no additional turn-backs, there is a single loop carrying
both spin insertions.  It has two possible orientations, giving
$2(-1)^s$.  If there are $t$ additional turn-backs, they must again all
lie on the same arc between $i$ and $j$; otherwise the connected
contribution cancels.  The configuration then has $t+1$ loops, but for
each orientation of the loop containing both spin insertions, the
orientation of the other loop through the turn-back site is fixed by
the requirement that the spin-insertion vertex carry a net flow between
the layers.  Thus the number of independent orientation choices is
again $2^t$.  The total contribution is therefore
\begin{multline*}
 (-1)^s\left\{ 2+ \sum_{t=1}^{s-1} 2^{t} {{s-1} \choose {t}} + \sum_{t=1}^{\ell-s-1} 2^{t} {{\ell-s-1} \choose {t}}\right\}\\
 = (-1)^s\left(3^{s-1}+3^{\ell-s-1}\right) \, .
\end{multline*}
The same contribution is obtained when the turn-back is placed at site
$i$ and the pass-through vertex at site $j$.

\emph{(iii) Both spin-insertion vertices are turn-backs.}
If there are no additional turn-backs, the two turn-backs at sites
$i$ and $j$ form two loops.  Only the two configurations in which these
two loops have opposite orientations give a nonzero spin-insertion
contribution, so the result is again $2(-1)^s$.  Configurations in
which the two loops have the same orientation give zero.

With $t$ additional turn-backs, all additional turn-backs must again
lie on the same arc between $i$ and $j$ in order to keep the spin
insertions connected.  The number of placements is
\[
    {s-1\choose t}+{\ell-s-1\choose t}.
\]
There are now $t+2$ loops.  However, the loop connecting the two spin
insertions and the two loops adjacent to it have only two overall
orientation choices that give a nonzero spin-insertion contribution.
The remaining $t-1$ loops, when present, can be oriented independently.
Equivalently, the total number of independent orientation choices is
$2^t$.  Hence the net contribution is again
\begin{multline*}
 (-1)^s\left\{ 2+ \sum_{t=1}^{s-1} 2^{t} {{s-1} \choose {t}} + \sum_{t=1}^{\ell-s-1} 2^{t} {{\ell-s-1} \choose {t}}\right\}\\
 = (-1)^s\left(3^{s-1}+3^{\ell-s-1}\right) \, .
\end{multline*}

Thus all possible choices of the local resolutions at the two
spin-insertion vertices give the same result.  We may therefore either
choose one representative resolution at each insertion site, or sum
over all four local choices and multiply by the compensating factor
$1/4$.  Including the coherent-state prefactor for
$S_i^zS_j^z$, the normalized correlation function is
 \begin{multline}\label{eq:corr_function_loop_Sz}
 \frac{\langle\Psi \lvert {S}_i^z\cdot {S}_j^z\rvert\Psi\rangle}{\langle\Psi |\Psi\rangle}
 = 4 \,(-1)^s\,\frac{3^{s-1}+3^{\ell-s-1}}{3^\ell \left(1+3^{1-\ell}\right)}\\
  =\frac{4}{3} \,(-1)^s\,\frac{3^{-s}+3^{s-\ell}}{1+ 3^{1-\ell}} \, .
\end{multline}
The factor of $4$ in the first line is the coherent-state prefactor
appropriate to two spin-$1$ operators at distinct sites.  Finally, by
SU(2) symmetry,
\[
    \langle\Psi \lvert \mathbf{S}_i\cdot\mathbf{S}_j\rvert\Psi\rangle
    =
    3\langle\Psi \lvert {S}_i^z {S}_j^z\rvert\Psi\rangle .
\]
We therefore recover the fixed-loop result quoted in the main text,
 \begin{equation*}\label{eq:corr_function_loop_again}
 \frac{\langle\Psi \lvert \mathbf{S}_i\cdot\mathbf{S}_j\rvert\Psi\rangle}{\langle\Psi |\Psi\rangle}
  =4 \,(-1)^s\,\frac{3^{-s}+3^{s-\ell}}{1+ 3^{1-\ell}} \, .
\end{equation*}

We close by applying the loop representation to a state for which the
plasma analogy gives a positive statistical ensemble.  Consider the
product state
\[
    \lvert\Psi_0\rangle =
    \bigotimes_{\alpha =1}^{N_\boxtimes} \lvert A_\alpha\rangle ,
\]
which contains no relative negative signs when written as a
superposition of vertical and horizontal dimer configurations.  It is
an equal-amplitude superposition of all product states consisting of
vertical and horizontal dimers on all crossed plaquettes or
tetrahedra.  Its norm can therefore be represented as a loop partition
function,
\begin{equation}\label{eq:part_function}
  Z=\int \prod_{k} d \boldsymbol{\Omega}_k  \left| \Psi_0(\Omega) \right|^2
  \propto \sum_{\text{all pairings } \eta} 2^{\Lambda(\eta)} \, ,
\end{equation}
where $\eta$ denotes a global choice of dimer configurations in the
top and bottom layers, together with a local resolution connecting the
two top-layer dimers to the two bottom-layer dimers at every site.  In
the state $|\Psi_0\rangle$, each crossed plaquette or tetrahedron is in
the equal-amplitude superposition of the two basis dimerizations
appearing in $|A\rangle$.  Hence, at a given site, each of the two
adjacent crossed plaquettes or tetrahedra contributes one of two
possible dimers.  There are therefore \(2\times2=4\) possible incident
dimer pairs in the top layer and similarly \(4\) in the bottom layer.
Once a top-layer pair and a bottom-layer pair are fixed, there are two
ways to connect them locally.  Thus, the local resolution problem has
the structure \(4\times4\times2\), subject to the global compatibility
constraints imposed by the neighbouring plaquettes or tetrahedra.
For a global pairing configuration $\eta$, $\Lambda(\eta)$ denotes the
number of closed loops in the resulting doubled-loop configuration.
Each such loop can be oriented in two ways, giving the factor
$2^{\Lambda(\eta)}$.  The common overall factor $(2\pi/3)^N$ has again
been omitted.

\section{Dimer matrix elements calculation}
\label{app:matrix_elements_calculation}

In this appendix we describe how the dimer-basis matrix elements used
in Sec.~\ref{sec:eff_Hamiltians} are estimated from the coherent-state
expression in Eq.~\eqref{eq:matrix_element_gen}.  The same procedure
applies to the nearest-neighbour matrix elements in
Eqs.~\eqref{eq:diag_element_13}--\eqref{eq:offdiag_element_1x}, to the
off-diagonal nearest-neighbour elements in
Eqs.~\eqref{eq:offdiag_nn_elements_1} and
\eqref{eq:offdiag_nn_elements_2}, and to the next-nearest-neighbour
matrix elements listed in Tables~\ref{table:matrix_elements} and
\ref{table:rotated_matrix_elements}.  We give two representative
calculations: the nearest-neighbour off-diagonal element
$\langle \hdimers | \mathbf{S}_1 \cdot \mathbf{S}_2 | \vdimers \rangle$
on a single crossed plaquette, and the next-nearest-neighbour element
$\langle \hdimers, \hdimers | \mathbf{S}_3 \cdot \mathbf{S}_7
| \vdimers, \hdimers \rangle$ on a pair of adjacent crossed
plaquettes.  All other entries are obtained in the same way, up to
lattice symmetries and the signs fixed by the arrow conventions.

We first consider
$\langle \hdimers | \mathbf{S}_1 \cdot \mathbf{S}_2 | \vdimers \rangle$.
According to Eq.~\eqref{eq:matrix_element_gen}, this requires the
matrix element
$\langle \psi^{\shdimers} | \mathbf{S}_1 \cdot \mathbf{S}_2
| \psi^{\svdimers} \rangle$ in the numerator and the norm factors
$\sqrt{\langle \psi^{\svdimers} | \psi^{\svdimers} \rangle}
\sqrt{\langle \psi^{\shdimers} | \psi^{\shdimers} \rangle}$
in the denominator.  Here $| \psi^{\svdimers} \rangle$ denotes the
spin-coherent representation of the loop state whose singlet covering
on the crossed plaquette under consideration is $| \vdimers \rangle$;
the notation $| \psi^{\shdimers} \rangle$ is defined analogously.

The relevant singlet coverings are shown in
Fig.~\ref{fig:one-chequer_off-diag}.  Panel (a) represents the
numerator, with ket singlets in red and bra singlets in blue.  Panels
(b) and (c) represent the two norm factors in the denominator.  As in
Appendices~\ref{app:stat_mech_loops} and
\ref{app:loops_correlations}, the local coherent-state integrals can
be evaluated by resolving each vertex into allowed
$2\text{-in}$-$2\text{-out}$ arrow configurations.  In the present
local estimate we retain only the resolutions that connect the two
sites through the local segment shown in the figure.  Contributions in
which sites $1$ and $2$ are connected only through a longer indirect
loop segment are discarded; the purple turn-back curves in
Fig.~\ref{fig:one-chequer_off-diag} indicate these omitted
connections.

\begin{figure}[h!]
  \centering
  \includegraphics[width=\linewidth]{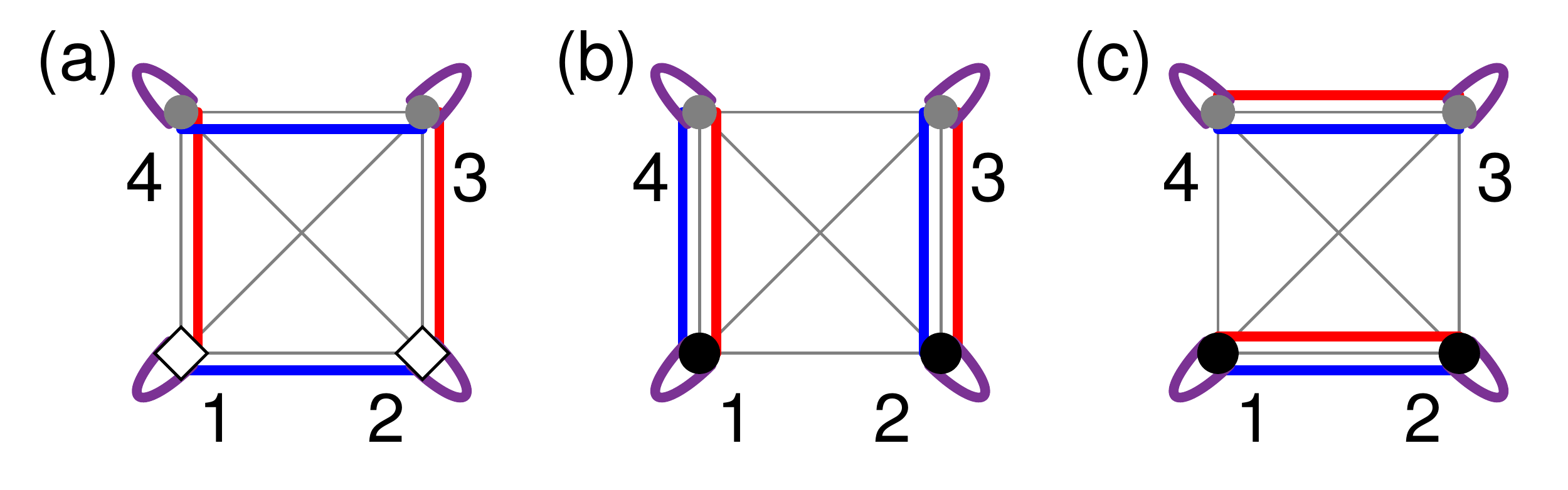}
  \caption{Local singlet coverings used to evaluate the off-diagonal
matrix element on the lower crossed plaquette of
Fig.~\ref{fig:two_chequers}.  Red lines denote ket singlets and blue
lines denote bra singlets.  Purple arcs indicate local turn-backs,
where the loop does not connect to another site of the crossed
plaquette.  Empty diamonds and black circles follow the vertex
conventions of Appendices~\ref{app:stat_mech_loops} and
\ref{app:loops_correlations}; grey sites do not carry spin insertions
in this example.  (a) Numerator configuration for the ket state
$|\vdimers\rangle$ and bra state $\langle\hdimers|$.  (b) Norm
configuration for $|\vdimers\rangle$ and $\langle\vdimers|$.  (c) Norm
configuration for $|\hdimers\rangle$ and $\langle\hdimers|$.}
\label{fig:one-chequer_off-diag}
\end{figure}

For the numerator in Fig.~\ref{fig:one-chequer_off-diag}(a) to be
nonzero, the two ket strands at site $1$ must both have the same arrow
orientation, either both incoming or both outgoing.  The two bra strands
at that site then have the opposite orientation.  This fixes the
allowed arrow structure at site $2$: the two ket strands at site $2$
must have the opposite orientation to those at site $1$, and the bra
strands again have the opposite orientation to the ket strands.

There are two global choices for this pattern, corresponding to the
two possible common orientations of the ket arrows at site $1$.  For
each such choice, there are two local bra pathways at site $1$ and two
local bra pathways at site $2$.  However, the sites $1$ and $2$ carry
spin insertions, so these local resolutions are not counted with the
same weight as ordinary overlap vertices.  As explained in
Appendix~\ref{app:loops_correlations}, one may either keep an explicit
factor of $1/2$ for each such local pathway, or equivalently count only
one effective resolution at each spin-insertion site.  Thus the number
of numerator contributions from sites $1$ and $2$ is
\[
    2 \times \left(\frac{1}{2}\times 2\right)
      \times \left(\frac{1}{2}\times 2\right)
    =
    2 .
\]
This is the point at which the spin-insertion sites differ from
ordinary norm vertices: the factors of $2$ from the two possible local
pathways at each insertion site are compensated by the corresponding
factors of $1/2$.

The remaining prefactors have a separate origin.  The factor of $4$
comes from the coherent-state representation of
$\mathbf{S}_1\cdot\mathbf{S}_2$ in Eq.~\eqref{eq:matrix_element_gen}.
The factor of $3$ comes from rotational invariance:
\[
    \langle \mathbf{S}_1\cdot\mathbf{S}_2\rangle
    =
    3\langle S^z_1 S^z_2\rangle ,
\]
so it is enough to count the $z$-component contribution and multiply
by three.  Finally, the sign is fixed by the spin-insertion weights in
Eq.~\eqref{eq:spin_weight3}.  In the present matrix element, one of
the two insertion sites carries the opposite sign, giving an overall
factor of $-1$.  Therefore the numerator contributes
\[
    4 \times 3 \times 2 \times (-1).
\]
The common factors of $2\pi/3$ produced by the local coherent-state
integrals at each site cancel pairwise between numerator and
denominator.

We now turn to the denominator.  The norm terms in
Figs.~\ref{fig:one-chequer_off-diag}(b) and
\ref{fig:one-chequer_off-diag}(c) contain no spin operators.  Hence
there is no requirement that the two ket strands, or the two bra
strands, at a given site have the same orientation.  Following the
enumeration procedure leading to Eq.~\eqref{eq:AKLT_part_function} in
Appendix~\ref{app:loops_correlations}, the number of allowed
$2\text{-in}$-$2\text{-out}$ arrow resolutions at sites $1$ and $2$,
after excluding the longer indirect loop connections, is $36$ for each
of the two norm diagrams. For the norm diagram in Fig.~\ref{fig:one-chequer_off-diag}(b), at site $1$, we can admit resolutions of either one turn back or no turn back at all. The same applies for site $2$ as well. This produces a factor of $2+ (2 \times 2^2) = 6$ for each of the sites and hence we get $36$ for $\langle \psi^{\svdimers} | \psi^{\svdimers} \rangle$. Whereas the norm diagram in panel (c) of the same figure allows for zero turn back at either of the sites, or two ways of turning back once (either at site $1$ or site $2$, while the other site is a pass through or there is no turn back), or two turn backs at both sites. So, the total number of different resolutions is 
 $2+(2 \times 2^2) + 2^3 = 18$. Along with an extra multiplicative factor of $2$ resulting from the two orientations of loops connecting sites $3$ and $4$, we get the same factor of $36$. 
 
Resolutions at all other sites are common
to numerator and denominator and cancel in the normalized matrix
element.  The only uncancelled local factors are those associated with
the two spin insertions at sites $1$ and $2$.  Therefore,
\begin{equation} \label{eq:one-chequer_off-diag_elements}
    \langle \hdimers | \mathbf{S}_1 \cdot \mathbf{S}_2 | \vdimers \rangle = \frac{4 \times 3 \times 2 \times (-1)}{\sqrt{36} \times \sqrt{36}} = -\frac{2}{3} \, .
\end{equation}

The same logic applies to next-nearest-neighbour matrix elements on
adjacent crossed plaquettes.  We illustrate it using
$\langle \hdimers, \hdimers | \mathbf{S}_3 \cdot \mathbf{S}_7
| \vdimers, \hdimers \rangle$, shown in
Fig.~\ref{fig:two-chequers_off-diag}.  Here
$| \vdimers, \hdimers \rangle$ denotes a state with
$|\vdimers\rangle$ on the lower crossed plaquette and
$|\hdimers\rangle$ on the upper crossed plaquette; the bra state is
defined similarly.

\begin{figure}[h!]
  \centering
  \includegraphics[width=\linewidth]{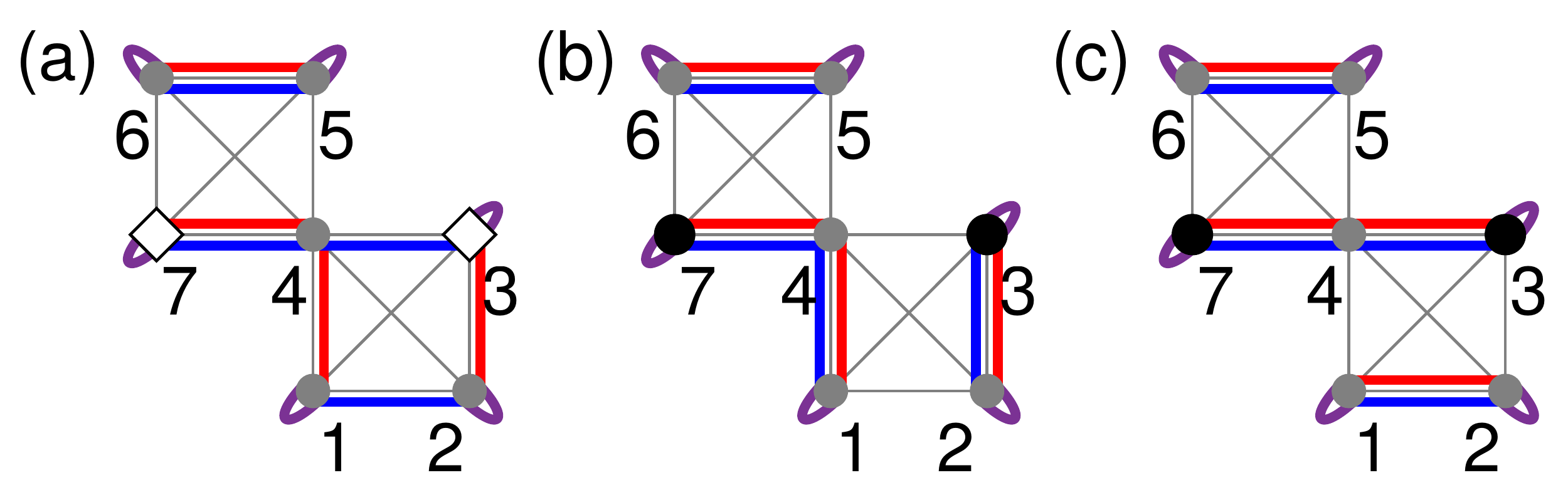}
  \caption{Singlet coverings used to evaluate a representative
off-diagonal matrix element on the crossed-plaquette pair of
Fig.~\ref{fig:two_chequers}.  Red and blue lines denote ket and bra
singlets, respectively.  In a state such as
$|\vdimers,\hdimers\rangle$, the first dimer configuration refers to
the lower crossed plaquette and the second to the upper crossed
plaquette; the same convention is used for bra states.  (a) Numerator
configuration for the ket state $|\vdimers,\hdimers\rangle$ and bra
state $\langle\hdimers,\hdimers|$.  (b) Norm configuration for
$|\vdimers,\hdimers\rangle$ and $\langle\vdimers,\hdimers|$.  (c) Norm
configuration for $|\hdimers,\hdimers\rangle$ and
$\langle\hdimers,\hdimers|$.}
\label{fig:two-chequers_off-diag}
\end{figure}

For the numerator in Fig.~\ref{fig:two-chequers_off-diag}(a) to be
nonzero, the same spin-insertion constraints as above must hold at
the two sites carrying the operators, now sites $3$ and $7$.  In
addition, the local loop resolution at the intermediate site $4$ must
connect the path from site $3$ to site $7$.  That is, a path entering
site $4$ from site $3$ is not allowed to turn back to site $3$; it must
continue through site $4$ toward site $7$. Otherwise, no single loop will contain both the sites $3$ and $7$. Thus, this no-break condition is the local statement that the spin insertions at sites $3$ and $7$ are
connected by the retained segment of the loop.  With this condition
imposed, the numerator has the same effective spin-insertion count as
in the single-plaquette calculation. Thus it contributes
\[
    4 \times 3 \times 2 .
\]
There is no minus sign in this case, because the allowed connected
resolutions have the same spin-insertion sign at sites $3$ and $7$:
whenever the two ket strands are incoming at site $3$, they are also
incoming at site $7$, and similarly for the outgoing configuration.

The no-break condition at site $4$ is not imposed in the denominator
norms, since no spin operators are inserted there and the norm only
requires allowed local loop resolutions.  The number of allowed local
resolutions at sites $3$, $4$, and $7$, with the same exclusion of
longer indirect loop connections used in the numerator, is $108$ for
each of the two denominator diagrams. Equivalently, this is the
corresponding norm-counting factor for the three sites involved in the
local two-plaquette process.  To explain this factor in more detail, consider the norm diagram in Fig.~\ref{fig:two-chequers_off-diag}(b). We have
$2+(2\times 2^2)+2^3=18$ loop resolutions, corresponding to no
turn-backs at either site $4$ or site $7$, one turn-back at one of
these two sites, or two turn-backs at both sites.  For site $3$, the
enumeration is the same as in the norm diagram of
Fig.~\ref{fig:one-chequer_off-diag}(b), giving a factor of $6$.
Together, this gives $18\times6=108$.  All other site factors again
cancel between numerator and denominator. A similar procedure for the norm diagram of Fig.~\ref{eq:two-chequer_off-diag_elements}(c) gives the same value $108$. Hence,
\begin{equation} \label{eq:two-chequer_off-diag_elements}
    \langle \hdimers, \hdimers | \mathbf{S}_3 \cdot \mathbf{S}_7 | \vdimers, \hdimers \rangle = \frac{4 \times 3 \times 2}{\sqrt{108} \times \sqrt{108}} = \frac{2}{9} \, .
\end{equation}

\section{Basis conversion for next-nearest-neighbour spin interactions}
\label{app:matrix_elements_basis_conversion}

In this appendix we spell out the conversion from the local dimer basis
used in Table~\ref{table:matrix_elements} to the orthonormal
pseudospin basis used in the effective Hamiltonians of
Sec.~\ref{sec:nnn_effective_H}.  The purpose is to make explicit how
the leading next-nearest-neighbour matrix elements between physical
spin-$1$ moments become matrices acting on the two local singlet
doublets associated with a neighbouring pair of crossed plaquettes or
tetrahedra.

The dimer-basis matrix elements in Table~\ref{table:matrix_elements}
are written in terms of the local valence-bond states
$|\hdimers\rangle$ and $|\vdimers\rangle$.  These two states are
normalized but not orthogonal:
\[
    \langle \hdimers|\hdimers\rangle
    =
    \langle \vdimers|\vdimers\rangle
    =
    1,
    \qquad
    \langle \hdimers|\vdimers\rangle
    =
    \frac{1}{2}.
\]
The orthonormal local basis is
\begin{subequations}\label{eq:basis_states_repeat}
\begin{equation}\label{eq:A_state_repeat}
\lvert A\rangle = \frac{\sqrt{3}}{3}
\left(\lvert \hdimers\rangle + \lvert \vdimers\rangle\right) \, ,
\end{equation}
\begin{equation}\label{eq:B_state_repeat}
\lvert B\rangle =
\left(\lvert \hdimers\rangle - \lvert \vdimers\rangle\right) \, .
\end{equation}
\end{subequations}
The inverse relations, which are the ones used in the conversion below,
are
\[
    |\hdimers\rangle =
    \frac{\sqrt{3}|A\rangle+|B\rangle}{2},
    \qquad
    |\vdimers\rangle =
    \frac{\sqrt{3}|A\rangle-|B\rangle}{2}.
\]

For a neighbouring pair of crossed plaquettes or tetrahedra, labelled
$1$ and $2$, we use the ordered orthonormal basis
\[
    \mathcal B_{AB}
    =
    \left\{
    |A\rangle_1|A\rangle_2,\,
    |A\rangle_1|B\rangle_2,\,
    |B\rangle_1|A\rangle_2,\,
    |B\rangle_1|B\rangle_2
    \right\}.
\]
The corresponding dimer product basis is taken to be
\[
    \mathcal B_{\rm dim}
    =
    \left\{
    |\vdimers,\vdimers\rangle,\,
    |\vdimers,\hdimers\rangle,\,
    |\hdimers,\vdimers\rangle,\,
    |\hdimers,\hdimers\rangle
    \right\},
\]
where the first entry refers to the lower crossed plaquette or
tetrahedron, and the second entry refers to the upper crossed
plaquette or tetrahedron in Figs.~\ref{fig:two_chequers} and
\ref{fig:two_tetrahedra}.  In this notation,
\[
\begin{pmatrix}
|A,A\rangle\\
|A,B\rangle\\
|B,A\rangle\\
|B,B\rangle
\end{pmatrix}
\quad
\text{is obtained from}
\quad
\begin{pmatrix}
|\vdimers,\vdimers\rangle\\
|\vdimers,\hdimers\rangle\\
|\hdimers,\vdimers\rangle\\
|\hdimers,\hdimers\rangle
\end{pmatrix}
\]
by expanding each $|A\rangle$ and $|B\rangle$ using
Eqs.~\eqref{eq:basis_states_repeat}.  Equivalently, if
$M^{\rm dim}_{pq}$ denotes the matrix of
$\mathbf S_p\cdot\mathbf S_q$ in the dimer product basis, and
$M^{AB}_{pq}$ denotes the matrix in the orthonormal
$|A\rangle,|B\rangle$ basis, then
\[
    M^{AB}_{pq}
    =
    T^{\mathsf T} M^{\rm dim}_{pq} T ,
\]
where all coefficients are real and
\[
T=
\begin{pmatrix}
\frac{1}{3} & -\frac{\sqrt{3}}{3} & -\frac{\sqrt{3}}{3} & 1\\
\frac{1}{3} &  \frac{\sqrt{3}}{3} & -\frac{\sqrt{3}}{3} & -1\\
\frac{1}{3} & -\frac{\sqrt{3}}{3} &  \frac{\sqrt{3}}{3} & -1\\
\frac{1}{3} &  \frac{\sqrt{3}}{3} &  \frac{\sqrt{3}}{3} & 1
\end{pmatrix}.
\]
The rows of $T$ are ordered as
$\mathcal B_{\rm dim}$, and the columns are ordered as
$\mathcal B_{AB}$.

As an example, consider the operator
$\mathbf S_1\cdot\mathbf S_5$.  Using the entries of
Table~\ref{table:matrix_elements}, one finds
\[
\begin{split}
{}_1\langle A|{}_2\langle A|
\mathbf S_1\cdot\mathbf S_5
|A\rangle_1|A\rangle_2
&=
\frac{1}{9}
\sum_{d,d'\in\mathcal B_{\rm dim}}
\langle d|\mathbf S_1\cdot\mathbf S_5|d'\rangle  \\
&=
\frac{1}{9}\left(\frac{16}{9}\right)
=
\frac{16}{81}.
\end{split}
\]
Similarly,
\[
\begin{aligned}
{}_1\langle A|{}_2\langle A|
\mathcal O_{15}
|A\rangle_1|B\rangle_2
&=
-\frac{8\sqrt{3}}{81},\\
{}_1\langle A|{}_2\langle A|
\mathcal O_{15}
|B\rangle_1|B\rangle_2
&=
\frac{4}{27},
\end{aligned}
\qquad
\mathcal O_{15}\equiv \mathbf S_1\cdot\mathbf S_5 .
\]
The remaining entries are obtained in the same way.  Since, the
microscopic operator is Hermitian and the basis coefficients are real,
the matrices below are real symmetric.  We nevertheless keep
$\operatorname{Re}$ in the intermediate expressions to emphasize that
these are the real parts of the corresponding off-diagonal matrix
elements.  In the matrix equations below, the notation
$\langle \mathbf S_p\cdot\mathbf S_q\rangle$ denotes the matrix
representation of the operator $\mathbf S_p\cdot\mathbf S_q$ in the
ordered basis $\mathcal B_{AB}$.

Using Eqs.~\eqref{eq:basis_states_repeat} and the entries of
Table~\ref{table:matrix_elements}, we obtain
\begin{subequations}\label{AB_elements_15}
\begin{align}
 &{}_1\langle A|{}_2\langle A| \left(\mathbf{S}_1\cdot\mathbf{S}_5 \right)| A \rangle_1 | A \rangle_2 = +\frac{16}{81} \, ,\\
 &{}_1\langle A|{}_2\langle B| \left(\mathbf{S}_1\cdot\mathbf{S}_5 \right)| A \rangle_1 | B \rangle_2 = 0 \, ,\\
 &{}_1\langle B|{}_2\langle A| \left(\mathbf{S}_1\cdot\mathbf{S}_5 \right)| B \rangle_1 | A \rangle_2 = 0 \, ,\\
 &{}_1\langle B|{}_2\langle B| \left(\mathbf{S}_1\cdot\mathbf{S}_5 \right)| B \rangle_1 | B \rangle_2 = 0 \, ,\\
 &\operatorname{Re}\left({}_1\langle A|{}_2\langle A| \left(\mathbf{S}_1\cdot\mathbf{S}_5 \right)| A \rangle_1 | B \rangle_2\right)  = -\frac {8\sqrt{3}}{81} \, ,\\
 &\operatorname{Re}\left({}_1\langle A|{}_2\langle A| \left(\mathbf{S}_1\cdot\mathbf{S}_5 \right)| B \rangle_1 | A \rangle_2\right)  = -\frac {8\sqrt{3}}{81} \, ,\\
 &\operatorname{Re}\left({}_1\langle A|{}_2\langle A| \left(\mathbf{S}_1\cdot\mathbf{S}_5 \right)| B \rangle_1 | B \rangle_2\right) =+\frac{4}{27} \, ,\\
 &\operatorname{Re}\left({}_1\langle A|{}_2\langle B| \left(\mathbf{S}_1\cdot\mathbf{S}_5 \right)| B \rangle_1 | A \rangle_2\right) = +\frac{4}{27} \, ,\\
 &\operatorname{Re}\left({}_1\langle A|{}_2\langle B| \left(\mathbf{S}_1\cdot\mathbf{S}_5 \right)| B \rangle_1 | B \rangle_2\right)  = 0 \, ,\\
 &\operatorname{Re}\left({}_1\langle B|{}_2\langle A| \left(\mathbf{S}_1\cdot\mathbf{S}_5 \right)| B \rangle_1 | B \rangle_2\right)  = 0 \, .
\end{align}
\end{subequations}
\begin{align} \label{AB_matrix_15}
\begin{split}
    \langle \mathbf{S}_1 \cdot \mathbf{S}_5 \rangle =& \begin{pmatrix}
                    \frac{16}{81} & -\frac{8 \sqrt{3}}{81} & -\frac{8 \sqrt{3}}{81} & \frac{4}{27}\\
                    -\frac{8 \sqrt{3}}{81} & 0 & \frac{4}{27} & 0\\
                    -\frac{8 \sqrt{3}}{81} & \frac{4}{27} & 0 & 0\\
                    \frac{4}{27} & 0 & 0 & 0
                \end{pmatrix} .
\end{split}
\end{align}

\begin{subequations}\label{AB_elements_16}
\begin{align}
 &{}_1\langle A|{}_2\langle A| \left(\mathbf{S}_1\cdot\mathbf{S}_6 \right)| A \rangle_1 | A \rangle_2 = -\frac{8}{81} \, ,\\
 &{}_1\langle A|{}_2\langle B| \left(\mathbf{S}_1\cdot\mathbf{S}_6 \right)| A \rangle_1 | B \rangle_2 = +\frac{8}{27} \, ,\\
 &{}_1\langle B|{}_2\langle A| \left(\mathbf{S}_1\cdot\mathbf{S}_6 \right)| B \rangle_1 | A \rangle_2 = 0 \, ,\\
 &{}_1\langle B|{}_2\langle B| \left(\mathbf{S}_1\cdot\mathbf{S}_6 \right)| B \rangle_1 | B \rangle_2 = 0 \, ,\\
 &\operatorname{Re}\left({}_1\langle A|{}_2\langle A| \left(\mathbf{S}_1\cdot\mathbf{S}_6 \right)| A \rangle_1 | B \rangle_2\right)  = 0 \, ,\\
 &\operatorname{Re}\left({}_1\langle A|{}_2\langle A| \left(\mathbf{S}_1\cdot\mathbf{S}_6 \right)| B \rangle_1 | A \rangle_2\right)  = +\frac {4\sqrt{3}}{81} \, ,\\
 &\operatorname{Re}\left({}_1\langle A|{}_2\langle A| \left(\mathbf{S}_1\cdot\mathbf{S}_6 \right)| B \rangle_1 | B \rangle_2\right) = 0 \, ,\\
 &\operatorname{Re}\left({}_1\langle A|{}_2\langle B| \left(\mathbf{S}_1\cdot\mathbf{S}_6 \right)| B \rangle_1 | A \rangle_2\right) = 0 \, ,\\
 &\operatorname{Re}\left({}_1\langle A|{}_2\langle B| \left(\mathbf{S}_1\cdot\mathbf{S}_6 \right)| B \rangle_1 | B \rangle_2\right)  = -\frac{4 \sqrt{3}}{27} \, ,\\
 &\operatorname{Re}\left({}_1\langle B|{}_2\langle A| \left(\mathbf{S}_1\cdot\mathbf{S}_6 \right)| B \rangle_1 | B \rangle_2\right)  = 0 \, .
\end{align}
\end{subequations}
\begin{align} \label{AB_matrix_16}
\begin{split}
    \langle \mathbf{S}_1 \cdot \mathbf{S}_6 \rangle =& \begin{pmatrix}
                    -\frac{8}{81} & 0 & \frac{4 \sqrt{3}}{81} & 0\\
                    0 & \frac{8}{27} & 0 & -\frac{4 \sqrt{3}}{27}\\
                    \frac{4 \sqrt{3}}{81} & 0 & 0 & 0\\
                    0 & -\frac{4 \sqrt{3}}{27} & 0 & 0
                \end{pmatrix} .
\end{split}
\end{align}

\begin{subequations}\label{AB_elements_17}
\begin{align}
 &{}_1\langle A|{}_2\langle A| \left(\mathbf{S}_1\cdot\mathbf{S}_7 \right)| A \rangle_1 | A \rangle_2 = +\frac{16}{81} \, ,\\
 &{}_1\langle A|{}_2\langle B| \left(\mathbf{S}_1\cdot\mathbf{S}_7 \right)| A \rangle_1 | B \rangle_2 = 0 \, ,\\
 &{}_1\langle B|{}_2\langle A| \left(\mathbf{S}_1\cdot\mathbf{S}_7 \right)| B \rangle_1 | A \rangle_2 = 0 \, ,\\
 &{}_1\langle B|{}_2\langle B| \left(\mathbf{S}_1\cdot\mathbf{S}_7 \right)| B \rangle_1 | B \rangle_2 = 0 \, ,\\
 &\operatorname{Re}\left({}_1\langle A|{}_2\langle A| \left(\mathbf{S}_1\cdot\mathbf{S}_7 \right)| A \rangle_1 | B \rangle_2\right)  = \frac {8\sqrt{3}}{81} \, ,\\
 &\operatorname{Re}\left({}_1\langle A|{}_2\langle A| \left(\mathbf{S}_1\cdot\mathbf{S}_7 \right)| B \rangle_1 | A \rangle_2\right)  = -\frac {8\sqrt{3}}{81} \, ,\\
 &\operatorname{Re}\left({}_1\langle A|{}_2\langle A| \left(\mathbf{S}_1\cdot\mathbf{S}_7 \right)| B \rangle_1 | B \rangle_2\right) = -\frac{4}{27} \, ,\\
 &\operatorname{Re}\left({}_1\langle A|{}_2\langle B| \left(\mathbf{S}_1\cdot\mathbf{S}_7 \right)| B \rangle_1 | A \rangle_2\right) = -\frac{4}{27} \, ,\\
 &\operatorname{Re}\left({}_1\langle A|{}_2\langle B| \left(\mathbf{S}_1\cdot\mathbf{S}_7 \right)| B \rangle_1 | B \rangle_2\right)  = 0 \, ,\\
 &\operatorname{Re}\left({}_1\langle B|{}_2\langle A| \left(\mathbf{S}_1\cdot\mathbf{S}_7 \right)| B \rangle_1 | B \rangle_2\right)  = 0 \, .
\end{align}
\end{subequations}
\begin{align} \label{AB_matrix_17}
\begin{split}
    \langle \mathbf{S}_1 \cdot \mathbf{S}_7 \rangle =& \begin{pmatrix}
                    \frac{16}{81} & \frac{8 \sqrt{3}}{81} & -\frac{8 \sqrt{3}}{81} & -\frac{4}{27}\\
                    \frac{8 \sqrt{3}}{81} & 0 & -\frac{4}{27} & 0\\
                    -\frac{8 \sqrt{3}}{81} & -\frac{4}{27} & 0 & 0\\
                    -\frac{4}{27} & 0 & 0 & 0
                \end{pmatrix} .
\end{split}
\end{align}

\begin{subequations}\label{AB_elements_25}
\begin{align}
 &{}_1\langle A|{}_2\langle A| \left(\mathbf{S}_2\cdot\mathbf{S}_5 \right)| A \rangle_1 | A \rangle_2 = -\frac{8}{81} \, ,\\
 &{}_1\langle A|{}_2\langle B| \left(\mathbf{S}_2\cdot\mathbf{S}_5 \right)| A \rangle_1 | B \rangle_2 = 0 \, ,\\
 &{}_1\langle B|{}_2\langle A| \left(\mathbf{S}_2\cdot\mathbf{S}_5 \right)| B \rangle_1 | A \rangle_2 = +\frac{8}{27} \, ,\\
 &{}_1\langle B|{}_2\langle B| \left(\mathbf{S}_2\cdot\mathbf{S}_5 \right)| B \rangle_1 | B \rangle_2 = 0 \, ,\\
 &\operatorname{Re}\left({}_1\langle A|{}_2\langle A| \left(\mathbf{S}_2\cdot\mathbf{S}_5 \right)| A \rangle_1 | B \rangle_2\right)  = +\frac{4 \sqrt{3}}{81} \, ,\\
 &\operatorname{Re}\left({}_1\langle A|{}_2\langle A| \left(\mathbf{S}_2\cdot\mathbf{S}_5 \right)| B \rangle_1 | A \rangle_2\right)  = 0 \, ,\\
 &\operatorname{Re}\left({}_1\langle A|{}_2\langle A| \left(\mathbf{S}_2\cdot\mathbf{S}_5 \right)| B \rangle_1 | B \rangle_2\right) = 0 \, ,\\
 &\operatorname{Re}\left({}_1\langle A|{}_2\langle B| \left(\mathbf{S}_2\cdot\mathbf{S}_5 \right)| B \rangle_1 | A \rangle_2\right) = 0 \, ,\\
 &\operatorname{Re}\left({}_1\langle A|{}_2\langle B| \left(\mathbf{S}_2\cdot\mathbf{S}_5 \right)| B \rangle_1 | B \rangle_2\right)  = 0 \, ,\\
 &\operatorname{Re}\left({}_1\langle B|{}_2\langle A| \left(\mathbf{S}_2\cdot\mathbf{S}_5 \right)| B \rangle_1 | B \rangle_2\right)  = -\frac{4 \sqrt{3}}{27} \, .
\end{align}
\end{subequations}
\begin{align} \label{AB_matrix_25}
\begin{split}
    \langle \mathbf{S}_2 \cdot \mathbf{S}_5 \rangle =& \begin{pmatrix}
                    -\frac{8}{81} & \frac{4 \sqrt{3}}{81} & 0 & 0\\
                    \frac{4 \sqrt{3}}{81} & 0 & 0 & 0\\
                    0 & 0 & \frac{8}{27} & -\frac{4 \sqrt{3}}{27}\\
                    0 & 0 & -\frac{4 \sqrt{3}}{27} & 0
                \end{pmatrix} .
\end{split}
\end{align}

\begin{subequations}\label{AB_elements_26}
\begin{align}
 &{}_1\langle A|{}_2\langle A| \left(\mathbf{S}_2\cdot\mathbf{S}_6 \right)| A \rangle_1 | A \rangle_2 = +\frac{4}{81} \, ,\\
 &{}_1\langle A|{}_2\langle B| \left(\mathbf{S}_2\cdot\mathbf{S}_6 \right)| A \rangle_1 | B \rangle_2 = -\frac{4}{27} \, ,\\
 &{}_1\langle B|{}_2\langle A| \left(\mathbf{S}_2\cdot\mathbf{S}_6 \right)| B \rangle_1 | A \rangle_2 = -\frac{4}{27} \, ,\\
 &{}_1\langle B|{}_2\langle B| \left(\mathbf{S}_2\cdot\mathbf{S}_6 \right)| B \rangle_1 | B \rangle_2 = +{4}/{9} \, ,\\
 &\operatorname{Re}\left({}_1\langle A|{}_2\langle A| \left(\mathbf{S}_2\cdot\mathbf{S}_6 \right)| A \rangle_1 | B \rangle_2\right)  = 0 \, ,\\
 &\operatorname{Re}\left({}_1\langle A|{}_2\langle A| \left(\mathbf{S}_2\cdot\mathbf{S}_6 \right)| B \rangle_1 | A \rangle_2\right)  = 0 \, ,\\
 &\operatorname{Re}\left({}_1\langle A|{}_2\langle A| \left(\mathbf{S}_2\cdot\mathbf{S}_6 \right)| B \rangle_1 | B \rangle_2\right) = 0 \, ,\\
 &\operatorname{Re}\left({}_1\langle A|{}_2\langle B| \left(\mathbf{S}_2\cdot\mathbf{S}_6 \right)| B \rangle_1 | A \rangle_2\right) = 0 \, ,\\
 &\operatorname{Re}\left({}_1\langle A|{}_2\langle B| \left(\mathbf{S}_2\cdot\mathbf{S}_6 \right)| B \rangle_1 | B \rangle_2\right)  = 0 \, ,\\
 &\operatorname{Re}\left({}_1\langle B|{}_2\langle A| \left(\mathbf{S}_2\cdot\mathbf{S}_6 \right)| B \rangle_1 | B \rangle_2\right)  = 0 \, .
\end{align}
\end{subequations}
\begin{align} \label{AB_matrix_26}
\begin{split}
    \langle \mathbf{S}_2 \cdot \mathbf{S}_6 \rangle =& \begin{pmatrix}
                    \frac{4}{81} & 0 & 0 & 0\\
                    0 & -\frac{4}{27} & 0 & 0\\
                    0 & 0 & -\frac{4}{27} & 0\\
                    0 & 0 & 0 & \frac{4}{9}
                \end{pmatrix} .
\end{split}
\end{align}

\begin{subequations}\label{AB_elements_27}
\begin{align}
 &{}_1\langle A|{}_2\langle A| \left(\mathbf{S}_2\cdot\mathbf{S}_7 \right)| A \rangle_1 | A \rangle_2 = -\frac{8}{81} \, ,\\
 &{}_1\langle A|{}_2\langle B| \left(\mathbf{S}_2\cdot\mathbf{S}_7 \right)| A \rangle_1 | B \rangle_2 = 0 \, ,\\
 &{}_1\langle B|{}_2\langle A| \left(\mathbf{S}_2\cdot\mathbf{S}_7 \right)| B \rangle_1 | A \rangle_2 = +\frac{8}{27} \, ,\\
 &{}_1\langle B|{}_2\langle B| \left(\mathbf{S}_2\cdot\mathbf{S}_7 \right)| B \rangle_1 | B \rangle_2 = 0 \, ,\\
 &\operatorname{Re}\left({}_1\langle A|{}_2\langle A| \left(\mathbf{S}_2\cdot\mathbf{S}_7 \right)| A \rangle_1 | B \rangle_2\right)  = -\frac{4 \sqrt{3}}{81} \, ,\\
 &\operatorname{Re}\left({}_1\langle A|{}_2\langle A| \left(\mathbf{S}_2\cdot\mathbf{S}_7 \right)| B \rangle_1 | A \rangle_2\right) = 0 \, ,\\
 &\operatorname{Re}\left({}_1\langle A|{}_2\langle A| \left(\mathbf{S}_2\cdot\mathbf{S}_7 \right)| B \rangle_1 | B \rangle_2\right) = 0 \, ,\\
 &\operatorname{Re}\left({}_1\langle A|{}_2\langle B| \left(\mathbf{S}_2\cdot\mathbf{S}_7 \right)| B \rangle_1 | A \rangle_2\right) = 0 \, ,\\
 &\operatorname{Re}\left({}_1\langle A|{}_2\langle B| \left(\mathbf{S}_2\cdot\mathbf{S}_7 \right)| B \rangle_1 | B \rangle_2\right)  = 0 \, ,\\
 &\operatorname{Re}\left({}_1\langle B|{}_2\langle A| \left(\mathbf{S}_2\cdot\mathbf{S}_7 \right)| B \rangle_1 | B \rangle_2\right)  = +\frac{4 \sqrt{3}}{27} \, .
\end{align}
\end{subequations}
\begin{align} \label{AB_matrix_27}
\begin{split}
    \langle \mathbf{S}_2 \cdot \mathbf{S}_7 \rangle =& \begin{pmatrix}
                    -\frac{8}{81} & -\frac{4 \sqrt{3}}{81} & 0 & 0\\
                    -\frac{4 \sqrt{3}}{81} & 0 & 0 & 0\\
                    0 & 0 & \frac{8}{27} & \frac{4 \sqrt{3}}{27}\\
                    0 & 0 & \frac{4 \sqrt{3}}{27} & 0
                \end{pmatrix} .
\end{split}
\end{align}

\begin{subequations}\label{AB_elements_35}
\begin{align}
 &{}_1\langle A|{}_2\langle A| \left(\mathbf{S}_3\cdot\mathbf{S}_5 \right)| A \rangle_1 | A \rangle_2 = +\frac{16}{81} \, ,\\
 &{}_1\langle A|{}_2\langle B| \left(\mathbf{S}_3\cdot\mathbf{S}_5 \right)| A \rangle_1 | B \rangle_2 = 0 \, ,\\
 &{}_1\langle B|{}_2\langle A| \left(\mathbf{S}_3\cdot\mathbf{S}_5 \right)| B \rangle_1 | A \rangle_2 = 0 \, ,\\
 &{}_1\langle B|{}_2\langle B| \left(\mathbf{S}_3\cdot\mathbf{S}_5 \right)| B \rangle_1 | B \rangle_2 = 0 \, ,\\
 &\operatorname{Re}\left({}_1\langle A|{}_2\langle A| \left(\mathbf{S}_3\cdot\mathbf{S}_5 \right)| A \rangle_1 | B \rangle_2\right)  = -\frac{8 \sqrt{3}}{81} \, ,\\
 &\operatorname{Re}\left({}_1\langle A|{}_2\langle A| \left(\mathbf{S}_3\cdot\mathbf{S}_5 \right)| B \rangle_1 | A \rangle_2\right)  = +\frac{8 \sqrt{3}}{81} \, ,\\
 &\operatorname{Re}\left({}_1\langle A|{}_2\langle A| \left(\mathbf{S}_3\cdot\mathbf{S}_5 \right)| B \rangle_1 | B \rangle_2\right) = -\frac{4}{27} \, ,\\
 &\operatorname{Re}\left({}_1\langle A|{}_2\langle B| \left(\mathbf{S}_3\cdot\mathbf{S}_5 \right)| B \rangle_1 | A \rangle_2\right) = -\frac{4}{27} \, ,\\
 &\operatorname{Re}\left({}_1\langle A|{}_2\langle B| \left(\mathbf{S}_3\cdot\mathbf{S}_5 \right)| B \rangle_1 | B \rangle_2\right)  = 0 \, ,\\
 &\operatorname{Re}\left({}_1\langle B|{}_2\langle A| \left(\mathbf{S}_3\cdot\mathbf{S}_5 \right)| B \rangle_1 | B \rangle_2\right)  = 0 \, .
\end{align}
\end{subequations}
\begin{align} \label{AB_matrix_35}
\begin{split}
    \langle \mathbf{S}_3 \cdot \mathbf{S}_5 \rangle =& \begin{pmatrix}
                    \frac{16}{81} & -\frac{8 \sqrt{3}}{81} & \frac{8 \sqrt{3}}{81} & -\frac{4}{27}\\
                    -\frac{8 \sqrt{3}}{81} & 0 & -\frac{4}{27} & 0\\
                    \frac{8 \sqrt{3}}{81} & -\frac{4}{27} & 0 & 0\\
                    -\frac{4}{27} & 0 & 0 & 0
                \end{pmatrix} .
\end{split}
\end{align}

\begin{subequations}\label{AB_elements_36}
\begin{align}
 &{}_1\langle A|{}_2\langle A| \left(\mathbf{S}_3\cdot\mathbf{S}_6 \right)| A \rangle_1 | A \rangle_2 = -\frac{8}{81} \, ,\\
 &{}_1\langle A|{}_2\langle B| \left(\mathbf{S}_3\cdot\mathbf{S}_6 \right)| A \rangle_1 | B \rangle_2 = +\frac{8}{27} \, ,\\
 &{}_1\langle B|{}_2\langle A| \left(\mathbf{S}_3\cdot\mathbf{S}_6 \right)| B \rangle_1 | A \rangle_2 = 0 \, ,\\
 &{}_1\langle B|{}_2\langle B| \left(\mathbf{S}_3\cdot\mathbf{S}_6 \right)| B \rangle_1 | B \rangle_2 = 0 \, ,\\
 &\operatorname{Re}\left({}_1\langle A|{}_2\langle A| \left(\mathbf{S}_3\cdot\mathbf{S}_6 \right)| A \rangle_1 | B \rangle_2\right)  = 0 \, ,\\
 &\operatorname{Re}\left({}_1\langle A|{}_2\langle A| \left(\mathbf{S}_3\cdot\mathbf{S}_6 \right)| B \rangle_1 | A \rangle_2\right)  = -\frac{4 \sqrt{3}}{81} \, ,\\
 &\operatorname{Re}\left({}_1\langle A|{}_2\langle A| \left(\mathbf{S}_3\cdot\mathbf{S}_6 \right)| B \rangle_1 | B \rangle_2\right) = 0 \, ,\\
 &\operatorname{Re}\left({}_1\langle A|{}_2\langle B| \left(\mathbf{S}_3\cdot\mathbf{S}_6 \right)| B \rangle_1 | A \rangle_2\right) = 0 \, ,\\
 &\operatorname{Re}\left({}_1\langle A|{}_2\langle B| \left(\mathbf{S}_3\cdot\mathbf{S}_6 \right)| B \rangle_1 | B \rangle_2\right)  = +\frac{4 \sqrt{3}}{27} \, ,\\
 &\operatorname{Re}\left({}_1\langle B|{}_2\langle A| \left(\mathbf{S}_3\cdot\mathbf{S}_6 \right)| B \rangle_1 | B \rangle_2\right)  = 0 \, .
\end{align}
\end{subequations}
\begin{align} \label{AB_matrix_36}
\begin{split}
    \langle \mathbf{S}_3 \cdot \mathbf{S}_6 \rangle =& \begin{pmatrix}
                    -\frac{8}{81} & 0 & -\frac{4 \sqrt{3}}{81} & 0\\
                    0 & \frac{8}{27} & 0 & \frac{4 \sqrt{3}}{27}\\
                    -\frac{4 \sqrt{3}}{81} & 0 & 0 & 0\\
                    0 & \frac{4 \sqrt{3}}{27} & 0 & 0
                \end{pmatrix} .
\end{split}
\end{align}

\begin{subequations}\label{AB_elements_37}
\begin{align}
 &{}_1\langle A|{}_2\langle A| \left(\mathbf{S}_3\cdot\mathbf{S}_7 \right)| A \rangle_1 | A \rangle_2 = +\frac{16}{81} \, ,\\
 &{}_1\langle A|{}_2\langle B| \left(\mathbf{S}_3\cdot\mathbf{S}_7 \right)| A \rangle_1 | B \rangle_2 = 0 \, ,\\
 &{}_1\langle B|{}_2\langle A| \left(\mathbf{S}_3\cdot\mathbf{S}_7 \right)| B \rangle_1 | A \rangle_2 = 0 \, ,\\
 &{}_1\langle B|{}_2\langle B| \left(\mathbf{S}_3\cdot\mathbf{S}_7 \right)| B \rangle_1 | B \rangle_2 = 0 \, ,\\
 &\operatorname{Re}\left({}_1\langle A|{}_2\langle A| \left(\mathbf{S}_3\cdot\mathbf{S}_7 \right)| A \rangle_1 | B \rangle_2\right)  = +\frac{8 \sqrt{3}}{81} \, ,\\
 &\operatorname{Re}\left({}_1\langle A|{}_2\langle A| \left(\mathbf{S}_3\cdot\mathbf{S}_7 \right)| B \rangle_1 | A \rangle_2\right)  = +\frac{8 \sqrt{3}}{81} \, ,\\
 &\operatorname{Re}\left({}_1\langle A|{}_2\langle A| \left(\mathbf{S}_3\cdot\mathbf{S}_7 \right)| B \rangle_1 | B \rangle_2\right) = +\frac{4}{27} \, ,\\
 &\operatorname{Re}\left({}_1\langle A|{}_2\langle B| \left(\mathbf{S}_3\cdot\mathbf{S}_7 \right)| B \rangle_1 | A \rangle_2\right) = +\frac{4}{27} \, ,\\
 &\operatorname{Re}\left({}_1\langle A|{}_2\langle B| \left(\mathbf{S}_3\cdot\mathbf{S}_7 \right)| B \rangle_1 | B \rangle_2\right)  = 0 \, ,\\
 &\operatorname{Re}\left({}_1\langle B|{}_2\langle A| \left(\mathbf{S}_3\cdot\mathbf{S}_7 \right)| B \rangle_1 | B \rangle_2\right)  = 0 \, .
\end{align}
\end{subequations}
\begin{align} \label{AB_matrix_37}
\begin{split}
    \langle \mathbf{S}_3 \cdot \mathbf{S}_7 \rangle =& \begin{pmatrix}
                    \frac{16}{81} & \frac{8 \sqrt{3}}{81} & \frac{8 \sqrt{3}}{81} & \frac{4}{27}\\
                    \frac{8 \sqrt{3}}{81} & 0 & \frac{4}{27} & 0\\
                    \frac{8 \sqrt{3}}{81} & \frac{4}{27} & 0 & 0\\
                    \frac{4}{27} & 0 & 0 & 0
                \end{pmatrix} .
\end{split}
\end{align}

The subscripts $1$ and $2$ on the states
$\lvert A \rangle$ and $\lvert B \rangle$ refer to the lower and upper
crossed plaquettes, or equivalently to the lower and upper tetrahedra,
in Figs.~\ref{fig:two_chequers} and \ref{fig:two_tetrahedra}.  The
matrices in Eqs.~\eqref{AB_matrix_15}--\eqref{AB_matrix_37} are then
combined with the coupling definitions in
Eqs.~\eqref{eq:checkerboard_coupling_def} and
\eqref{eq:pyrochlore_coupling_def}.  Summing the appropriate
symmetry-related matrices gives the projected next-nearest-neighbour
effective Hamiltonians in
Eqs.~\eqref{eq:eff_Hamiltonian_checkerboard} and
\eqref{eq:eff_Hamiltonian_pyrochlore}.

A useful check is that the irrational coefficients proportional to
$\sqrt{3}$ cancel once the spin-pair matrices are summed within the
symmetry-allowed coupling sectors.  Thus, the final effective
Hamiltonians in the main text contain only rational coefficients.  This
cancellation, which renders only the diagonal and anti-diagonal entries non-zero, is a nontrivial consequence of combining the local
orthonormalization of the nonorthogonal dimer basis with the spatial
symmetries of the checkerboard and pyrochlore lattices.

%

\end{document}